\documentclass[12pt,authoryear,review]{elsarticle}

\makeatletter

  \def\ps@pprintTitle{%
 \let\@oddhead\@empty
 \let\@evenhead\@empty
 \def\@oddfoot{\centerline{\thepage}}%
 \let\@evenfoot\@oddfoot}
\makeatother

\usepackage{amsmath,amssymb,amsfonts}
\DeclareMathOperator*{\argmax}{arg\,max}
\usepackage{graphicx}
\usepackage{caption}
\usepackage{subcaption}
\usepackage{dsfont}
\usepackage[dvipsnames]{xcolor}
\usepackage[rightcaption]{sidecap}
\usepackage{epstopdf}
\usepackage{todonotes}
\usepackage{hyperref}
\usepackage{bbm}
\usepackage{mathtools}
\usepackage{comment}
\usepackage{soul}

\newtheorem{theorem}{Theorem}

\newtheorem{proposition}[theorem]{Proposition}

\newtheorem{definition}[theorem]{Definition}
\newtheorem{remark}[theorem]{Remark}

\newenvironment{proof}[1][Proof]{\begin{trivlist}
\item[\hskip \labelsep {\bfseries #1}]}{\end{trivlist}}

\newcommand{\bs}{\boldsymbol}

\journal{TBA}

\begin{document}
\begin{frontmatter}

\title{Liquidity Competition Between Brokers and an Informed Trader
\tnoteref{t1}}
\tnotetext[t1]{}

\author[author1]{Ryan Donnelly}\ead{ryan.f.donnelly@kcl.ac.uk}
\author[author1]{Zi Li}\ead{zi.2.li@kcl.ac.uk}
\address[author1] {Department of Mathematics, King's College London, \\ Strand, London, WC2R 2LS, United Kingdom}

\date{}

\begin{abstract}
We study a multi-agent setting in which brokers transact with an informed trader. Through a sequential Stackelberg-type game, brokers manage trading costs and adverse selection with an informed trader. In particular, supplying liquidity to the informed traders allows the brokers to speculate based on the flow information. They simultaneously attempt to minimize inventory risk and trading costs with the lit market based on the informed order flow, also known as the internalization-externalization strategy. We solve in closed form for the trading strategy that the informed trader uses with each broker and propose a system of equations which classify the equilibrium strategies of the brokers. By solving these equations numerically we may study the resulting strategies in equilibrium. Finally, we formulate a competitive game between brokers in order to determine the liquidity prices subject to precommitment supplied to the informed trader and provide a numerical example in which the resulting equilibrium is not Pareto efficient.
\end{abstract}

\begin{keyword}
    informed trading, market making, equilibrium, algorithmic trading
\end{keyword}
\end{frontmatter}
\setstcolor{red}
\section{Introduction}

In financial markets, it is common that different parties have different information, a phenomenon referred to as information asymmetry. When dealing with order flow from informed traders, counterparties are generally adversely selected against and provide liquidity at a loss. There is extensive literature on informed and uninformed trading. Early models that deal with information asymmetry include \cite{glosten1985bid} and \cite{kyle1985continuous}. In \cite{kyle1985continuous}, the informed trader has superior information about the payoff of the asset, while noise traders trade randomly. Both the informed and uninformed traders interact with a dealer who specifies execution prices contingent on total order flow. The problem is analyzed by solving for an equilibrium between the dealer's price strategy and the informed trader's order flow strategy. Recently, \cite{herdegen2023liquidity} studies a one-shot Nash competition between an arbitrary number of identical dealers that compete for the order flow of a client. When quoting their price schedules, the dealers do not know the client’s type but only the distribution of clients. \cite{cartea2025brokers} characterizes a perfect information Stackelberg equilibrium between a broker and her clients — an informed trader and an uniformed trader in an over-the-counter market. Different from \cite{herdegen2023liquidity}, the broker can distinguish the type of traders and therefore quote bespoke prices as well as extract superior information from the informed trading flow. In addition to the game setting, some work studies brokers' optimization of trading performance despite the presence of an exogenous informed flow. \cite{barzykin2024unwinding} considers a central trading desk which aggregates the inflow of clients’ orders with unobserved toxicity, and formulates this optimal unwinding problem as a partially observable stochastic control problem. The academic work on detection and prediction of toxic flows also receives a significant amount of attention, see \cite{cartea2023detecting} and \cite{easley2011microstructure}.

In addition to utilizing information to make profitable trades, agents are concerned with minimizing trading costs and managing inventory risks. This type of optimal execution problem has also been heavily studied. \cite{bertsimas1998optimal} and \cite{almgren2001optimal} provide a groundwork for this type of problems in discrete time by modeling the permanent and temporary price impact caused by an agent's trades. Analogous models and results in continuous time are provided in \cite{gueant2012optimal} and \cite{forsyth2012optimal}. Optimal execution models can also be adapted to scenarios that include trading based on access to information, for example statistical arbitrage. One of the early studies in this line of research is \cite{cartea2014buy}, which develops a market-making model with a mean-reverting alpha component in the underlying asset price. Other market signals based on book volume imbalance have also been used, e.g. \cite{cartea2018enhancing,donnelly2020optimal}. 

The actions of an individual agent in a financial market will impact the behaviour of other agents. It is essential to formulate and study models in which multiple agents act simultaneously. There are numerous works in optimal execution which model the interaction between multiple agents, including \cite{schied2019market}, \cite{huang2019mean}, and \cite{donnelly2020optimal}. Our model is a generalization of the single broker and the single informed trader model considered in \cite{cartea2025brokers}. While \cite{bergault2025mean} is an extension of \cite{cartea2025brokers} where a large group of informed traders and one broker are considered in a mean-field setting, here we consider a setting in which a finite number of brokers trade strategically with a single informed trader.

This paper studies the optimal strategies of a group of brokers who act as liquidity providers to other agents. These brokers have as clients an informed trader and an uninformed trader, both of whom trade with all of the brokers simultaneously on bespoke quotes streamed by each broker. The informed trader has privileged information about the trend component of the asset price, while the uninformed trader has no informational advantage and their trades are exogenous and non-directional. In our model, the brokers can extract this privileged information from the informed order flow. Therefore, they can speculate based on the signal in the lit market while externalizing flow from both the informed and uninformed traders. The brokers and the informed trader maximize their expected wealth, while minimizing inventory holdings. When the finite time horizon is reached, each agent that has unsold units of the asset recovers a salvage cost for their remaining inventory. Once the optimal actions of each agent are determined, we demonstrate how the brokers select the liquidity price they offer to the informed trader through an equilibrium mechanism.


The remainder of this paper is organized as follows. Section \ref{sec:model} introduces the model and optimization problems solved by the broker and the informed trader. We first study the informed trader's problem, and derive his optimal strategy in closed-form. With the informed trader's optimal strategy at hand, we then turn to the brokers' problem. In Section \ref{sec:numerical} we conduct a numerical simulation of the strategies of all agents. In Section \ref{sec:liquidity_game} we demonstrate how brokers will select their liquidity prices in equilibrium in order to maximize their own value, and investigate how risk-aversion parameters affect the selection of these liquidity prices in equilibrium. Section \ref{sec:conclusion} concludes.

\section{Model}\label{sec:model}

Consider a market with $N$ brokers, an informed trader, and an uninformed trader.\footnote{The model treats each the informed and uninformed trader as an individual agent, but they could represent an aggregate of all informed and uninformed trading activity.} In this market only the $N$ brokers can provide liquidity by streaming quotes to the informed and uninformed traders. The informed and uninformed traders submit trades to all brokers simultaneously. The brokers know the identity of clients, and so they can supply different quotes to informed and uninformed traders. Additionally, the brokers are able to submit trades to a lit market.

Let $T>0$ be a finite-time horizon and $\mathcal{T}=[0,T]$. We work in a probability space $\left(\Omega,\mathcal{F}, \mathbb{F}=\left(\mathcal{F}_t\right)_{t\in\mathcal{T}},\mathbb{P}\right)$ satisfying the usual conditions and supporting $N+2$ Brownian motions $W^S$, $W^{\alpha}$, and $W^{U,j}$ for $j=1,\dots,N$. The Brownian motions $W^S$ and $W^\alpha$ are independent of all others, but the collection $W^{U,j}$ for $j=1,\dots,N$ has constant correlation matrix $\rho = (\rho_{i,j})$. The mid-price $\left(S^\nu_t\right)_{t\in\mathcal{T}}$ of the asset in the lit market satisfies
\begin{align}
    dS^\nu_t &= \left(\sum_{j=1}^N b_j\,\nu_t^{j}+\alpha_t\right)dt+\sigma\,dW^S_t\,,\quad S^\nu_0\in\mathbb{R}^+\,,\label{eqn:dS}\\
    d\alpha_t &= -\theta\,\alpha_t\,dt+\eta\,dW_t^{\alpha}\,,\quad \alpha_0\in\mathbb{R}\,,\label{eqn:dAlpha}
\end{align}
where $\sigma$, $\eta$, $\theta$, and $b_j$ for $j\in\{1,\dots,N\}$ are positive constants. In equation \eqref{eqn:dS}, the first $N$ terms in the drift represent the impacts each broker has on the mid-price, where $(\nu^{j}_t)_{t\in\mathcal{T}}$ represents broker $j$’s trading speed in the lit market and $b_j$ is the price impact parameter for broker $j$. The additional information held by the informed trader flows into the price of the asset via the last term of the drift in equation \eqref{eqn:dS}, represented by $\alpha_t$. This additional information is generally known as `alpha' and we assume that it is mean-reverting with dynamics given by equation \eqref{eqn:dAlpha}. 

Temporary price impacts occur when the brokers trade in the lit market. Specifically, the transaction price achieved by broker $j$ is
\begin{align*}
    \hat{S}^{j}_t = S^\nu_t+k_j\,\nu^{j}_t\,,\quad j\in\{1,\dots,N\}\,,
\end{align*}
where $k_j>0$ is the temporary price impact parameter. The informed and uninformed agents do not trade in the lit market, but instead submit trades to each of the brokers who provide them with liquidity. We assume that the identity of each trader is known to the brokers and a customized quote is made correspondingly. In particular, at time $t$, broker $j$ receives orders from the informed trader at speed $\omega^{j}_t$ and from the uninformed trader at speed $u^{j}_t$. Broker $j$ quotes the prices $\hat{S}^{I,j}_t$ and $\hat{S}^{U,j}_t$ for each type of trader. The cost of liquidity provided by each broker to both types of agent are specified as a function of each trader's rate of trading. The quotes (i.e., execution prices if there is a trade) for the informed and uninformed traders are
\begin{align}
    \hat{S}^{I,j}_t &= S^\nu_t + \kappa_j\,\omega^{j}_t\quad\text{and}\quad\hat{S}^{U,j}_t = S^\nu_t + c_j\,u^{j}_t\,,\quad j\in\{1,\dots,N\}\,,\label{eqn:ExecutionPrices}
\end{align}
respectively, where the liquidity cost parameters $\kappa_j>0$ are known by the informed trader and $c_j>0$ are known by the uninformed; when the trading rate is positive (negative) the trader buys (sells) the asset.

In the remainder of this section, subsection \ref{subsec:ITraderStratgy_multi} specifies and solves the informed trader’s problem, and subsection \ref{subsec:BrokerStratgy_multi} specifies and solves the problem faced by the brokers.

\subsection{Informed Trader's Strategy}\label{subsec:ITraderStratgy_multi}

In this section we adopt much of the analysis of \cite{cartea2025brokers} but where the informed trader has a multivariate control process distributing trades among the multiple brokers. After deriving the optimal trading speeds, we show that each broker can deduce total order flow even though they only directly observe the trades submitted to themselves.

The informed trader knows the alpha component \eqref{eqn:dAlpha} of the mid-price but not the brokers' trading rates $\nu^{j}$ in the lit market. The filtration of the informed trader $\left(\mathcal{F}^I_t\right)_{t\in\mathcal{T}}$ is given by
\begin{align*}
    \mathcal{F}^I_t\coloneqq\sigma\left[\left\{S_u\right\}_{u\le t}, \left\{\alpha_u\right\}_{u\le t}\right]\,.
\end{align*}
Here we follow an analysis similar to \cite{cartea2025brokers} in which the informed trader incorporates model ambiguity with respect to the assets dynamics because he is unable to observe the trades of the brokers in the lit market. To this end, we specify a new probability measure $\mathbb{P}^I$ fixed by the informed trader under which the dynamics of the asset price are
\begin{align*}
    dS_t &= \alpha_t\,dt+\sigma\,d\widetilde{W}^S_t\,,\quad S_0\in\mathbb{R}^+\,,\\
    d\alpha_t &= -\theta\,\alpha_t\,dt+\eta\,dW_t^{\alpha}\,,\quad \alpha_0\in\mathbb{R}\,,
\end{align*}
where $(\widetilde{W}^S_t)_{t\in\mathcal{T}}$ and $\left(W^{\alpha}_t\right)_{t\in\mathcal{T}}$ are $\mathbb{P}^I-$Brownian motions independent of each other, and he works on a complete filtered probability space $(\Omega,\mathcal{F},\mathbb{F}^I=(\mathcal{F}^I_t)_{t\in\mathcal{T}},\mathbb{P}^I)$. The set of admissible strategies for the informed trader is
\begin{align*}
    \mathcal{A}^I\coloneqq\Bigg\{&\bs{\omega}=\left(\bs{\omega}_t\right)_{t\in\mathcal{T}}=\left(\omega^{1}_t,\dots,\omega^{N}_t\right)_{t\in\mathcal{T}}\bigg\vert\bs{\omega}\text{ is }\mathbb{P}^I-\text{progressively measurable,}\\ 
    &\hspace{2mm}\text{and }\mathbb{E}^{\mathbb{P}^I}\left[\int^T_0|\bs{\omega}_s|^2 ds\right]<\infty\Bigg\}\,,
\end{align*}
where $\mathbb{E}^{\mathbb{P}^I}[\cdot]$ denotes $\mathbb{P}^I-$expectation. Recall that $\omega_t^j$ is the speed at which the informed trader submits trades to broker $j\in\{1,\dots,N\}$ at time $t$, and so the informed trader’s inventory process $(Q^{I,\omega}_t)_{t\in\mathcal{T}}$ and cash process $(X^{I,\omega}_t)_{t\in\mathcal{T}}$ satisfy
\begin{align*}
    dQ_t^{I,\omega} &= \sum_{j=1}^N \omega^{j}_t\,dt\,, \qquad Q_0^{I,\omega}=0\,,\\
    dX_t^{I,\omega} &= -\sum_{j=1}^N \left(S_t + \kappa_j\,\omega^{j}_t\right)\omega^{j}_t\,dt\,, \quad X_0^{I,\omega}=0\,.
\end{align*}
where $\kappa_j>0$ is the price that broker $j$ charges the informed trader for liquidity for any $j\in\{1,\dots,N\}$.

The informed trader knows that his information about the drift of the mid-price process under measure $\mathbb{P}^I$ is incomplete, and therefore considers alternative models of the mid-price dynamics under a candidate measure $\mathbb{Q}^I$ equivalent to $\mathbb{P}^I$. Define a Radon-Nikodym derivative process by
\begin{align*}
    \frac{d\mathbb{Q}^{I}(y^S)}{d\mathbb{P}^I}\Big\vert_t &= \exp\left\{-\frac{1}{2}\int^t_0\left(y^S_u\right)^2du-\int^t_0 y^S_u\,d\widetilde{W}^{S}_u\right\}\,,
\end{align*}
and define the full class of candidate measures considered by the informed trader by
\begin{align*}
    \mathcal{Q}^I = \Bigg\{
    &\mathbb{Q}^I(y^S):y^S\text{ is }\mathbb{F}^I-\text{adapted, and }\left(\frac{d\mathbb{Q}^{I}(y^S)}{d\mathbb{P}^I}\Big\vert_t\right)_{t\in\mathcal{T}}\text{ is a martingale under }\mathbb{P}^I\Bigg\}\,.
\end{align*}
In the new measure $\mathbb{Q}^I(y^S)$, the dynamics of the mid-price is changed to
\begin{align*}
    dS_t=\left(\alpha_t-\sigma y^S_t\right)dt+\sigma\,d\overline{W}^S_t\,,
\end{align*}
where $(\overline{W}^S_t)_{t\in\mathcal{T}}$ is a $\mathbb{Q}^I-$Browian motion. Within the set of candidate measures $\mathcal{Q}^I$, the informed trader aims to rank the candidates and choose the one that makes his strategy the more robust to the misspecification. Therefore, a penalty is introduced to measure the `cost' of deviation from the reference measure, i.e. rejecting the reference measure $\mathbb{P}^I$ and accepting a candidate measure $\mathbb{Q}^I$. A popular choice for the penalty function is the relative entropic penalty function
\begin{align*}
    \mathcal{H}^I_{t,T}\left(\mathbb{Q}^I\vert\mathbb{P}^I\right) = \frac{1}{\psi_I}\log\left\{\left(\frac{d\mathbb{Q}^I}{d\mathbb{P}^I}\right)_T\middle/\left(\frac{d\mathbb{Q}^I}{d\mathbb{P}^I}\right)_t\right\}\,,
\end{align*}
where $\psi_I>0$ is the ambiguity aversion parameter that shows the confidence of the informed trader towards his reference model. In the case where $\psi_I$ is very small, the informed trader is extremely confident about the reference model as any deviation from the reference model is costly. In the limiting case where $\psi_I\to0$, the informed trader rejects any alternative models. On the other hand, when $\psi_I$ is very large, considering alternative models results in a very small penalty, indicating that the informed trader is extremely ambiguous about the reference model. In the extreme case where $\psi_I\to\infty$, the informed trader focuses on the worst scenario when trading.

In this way, the informed trader's optimization problem becomes
\begin{align}
    \begin{split}
        H^I\left(t,\alpha,S,q^I,x^I\right)=\sup_{\bs{\omega}\in\mathcal{A}^I}\inf_{\mathbb{Q}^I\in\mathcal{Q}^I}\mathbb{E}^{\mathbb{Q}^I}_{t,\alpha,S,q^I,x^I}\biggl[X^{I,\omega}_T+Q^{I,\omega}_T\,S_T-a_I\left(Q^{I,\omega}_T\right)^2\hspace{10mm}\\
          -\phi_I\int^T_t\left(Q^{I,\omega}_u\right)^2 du+\mathcal{H}^I_{t,T}\left(\mathbb{Q}^I\vert\mathbb{P}^I\right)\biggr]\,,
    \end{split}\label{eqn:InformedValueFunction}
\end{align}
where $\mathbb{E}^{\mathbb{Q}^I}_{t,\alpha,S,q^I,x^I}[\cdot]$ denotes $\mathbb{Q}^I-$expectation conditional on $\alpha_t=\alpha$, $S_t=S$, $X^{I,\omega}_t=x^I$ and $Q^{I,\omega}_t=q^I$, and $a_I\ge0$ and $\phi_I\ge0$ are the terminal and running inventory penalty parameters, respectively. By the dynamic programming principle, the associated HJB–Isaacs (HJBI) equation is
\begin{align}
    \begin{split}
        &\partial_t H^I+\mathcal{L}^\alpha H^I-\phi_I\left(q^I\right)^2+\alpha\,\partial_S H^I+\frac{1}{2}\sigma^2\partial^2_{S}H^I+\inf_{y^S}\left\{-\sigma\,y^S\,\partial_S H^I+\frac{\left(y^S\right)^2}{2\,\psi_I}\right\}\\
        &\hspace{5mm} + \sup_{\omega^{1},\dots,\omega^{N}}\left\{\left(\sum_{j=1}^N \omega^{j}\right)\partial_{q^I}H^I-\left(\sum_{j=1}^N(S+\kappa_j\,\omega^{j})\,\omega^{j}\right)\partial_{x^I}H^I\right\}=0\,,\label{eqn:HJBI}
    \end{split}
\end{align}
with terminal condition $H^I(T,\alpha,S,q^I,x^I)=x^I+S\,q^I-a_I(q^I)^2$, where
\begin{align}
    \mathcal{L}^\alpha = -\theta\,\alpha\,\partial_\alpha+\frac{1}{2}\eta^2\partial^2_{\alpha}\label{eqn:InfinitesimalOperatorAlpha}
\end{align}
is the infinitesimal generator of the process $\alpha$. The optimizers within \eqref{eqn:HJBI} are
\begin{align}
    \begin{split}
        y^{S*} = \psi^I\,\sigma\,\partial_S H
        \quad\text{and}\quad\quad \omega^{j*} = \frac{1}{\kappa_j}\frac{\partial_{q^I}H^I-S\,\partial_{x^I}H^I}{2\,\partial_{x^I}H^I}.\label{eqn:OptimalInformedTradingSpeedBeforeAnsatz}
    \end{split}
\end{align}
Substituting these controls back into \eqref{eqn:HJBI}, we get the partial differential equation (PDE)
\begin{align}
    \begin{split}
        &\partial_t H^I+\mathcal{L}^\alpha H^I-\phi_I\left(q^I\right)^2+\alpha\,\partial_S H^I+\frac{1}{2}\sigma^2\partial^2_{S}H^I\\
        &\hspace{15mm} -\frac{1}{2}\,\psi_I\left(\sigma\,\partial_S H^I\right)^2+\left(\sum_{j=1}^N\frac{1}{\kappa_j}\right)\frac{\left(\partial_{q^I}H^I-S\,\partial_{x^I}H^I\right)^2}{4\,\partial_{x^I}H^I}=0\,.
    \end{split}\label{eqn:HJBI_sub}
\end{align}

\begin{proposition}[Solution to HJBI Equation]
    The HJBI equation \eqref{eqn:HJBI_sub} admits the ansatz
    \begin{align}
        H^I\left(t,\alpha,S,q^I,x^I\right)=x^I+S\,q^I+h^I_0\left(t,\alpha\right)+h^I_1\left(t,\alpha\right)\,q^I+h^I_2\left(t\right)\left(q^I\right)^2\,,\label{eqn:InformedValueFunctionCandidate}
    \end{align}
    with the functions $h_0^I$, $h_1^I$, and $h_2^I$ given by
    \begin{align*}
        h^I_0\left(t,\alpha\right) &= f^I_0(t)+\alpha^2\,f^I_2(t)\,,\\
        h^I_1\left(t,\alpha\right) &= \alpha\,m^I(t)\,,\\
        h^I_2(t) &= -\sqrt{\kappa\,\Phi}\,\,\,\frac{\zeta\,e^{\gamma(T-t)}+e^{-\gamma(T-t)}}{\zeta\,e^{\gamma(T-t)}-e^{-\gamma(T-t)}}\,\,,
    \end{align*}
    where
    \begin{align*}
        \frac{1}{\kappa} &= \sum_{j=1}^N \frac{1}{\kappa_j}\,,\quad \Phi=\frac{1}{2}\,\psi_I\sigma^2+\phi_I\,,\\
        \gamma&=\sqrt{\frac{\Phi}{\kappa}}\,,\quad\quad\,\zeta = \frac{a_I+\sqrt{\kappa\,\Phi}}{a_I-\sqrt{\kappa\,\Phi}}\,,
    \end{align*}
    and
    \begin{align*}
        m^I(t) &= \frac{\zeta}{\theta+\gamma}\biggl(\frac{e^{-\theta(T-t)}-e^{\gamma(T-t)}}{e^{-\gamma(T-t)}-\zeta\,e^{\gamma(T-t)}}\biggr)-\frac{1}{\theta-\gamma}\biggl(\frac{e^{-\theta(T-t)}-e^{-\gamma(T-t)}}{e^{-\gamma(T-t)}-\zeta\,e^{\gamma(T-t)}}\biggr)\\
        f^I_2(t) &= \int^T_t\frac{\left(m^I(u)\right)^2}{4\,\kappa}\,e^{-2\,\theta\,(u-t)}\,du\,,\\
        f^I_0(t) &= \int^T_t\eta^2 f^I_2(u)\,du\,.
    \end{align*}
\end{proposition}

\begin{proof}
    This can be checked by direct substitution of \eqref{eqn:InformedValueFunctionCandidate} into \eqref{eqn:HJBI_sub}. \qed
\end{proof}

\begin{theorem}[Verification Theorem]\label{the:VerificationForInformedTrader} The control problem in \eqref{eqn:InformedValueFunction} has a classical solution. The value function in \eqref{eqn:InformedValueFunction} is given by \eqref{eqn:InformedValueFunctionCandidate} and the optimal trading strategy of the informed trader with broker $j\in\{1,\dots,N\}$ is the admissible control given in feedback form by
    \begin{align}
        \omega^{j*}_t = \omega^j_0(t)\,\alpha_t + \omega^j_1(t)\,Q^{I,\omega^*}_t\,,\quad\text{for } j\in\{1,\dots, N\}\,,\label{eqn:OptimalInformedTradingSpeed}
    \end{align}
    where
    \begin{align}
        \omega^j_0(t) = \frac{1}{\kappa_j}\,\frac{m^I(t)}{2}\,,\quad\quad \omega^j_1(t) = -\frac{1}{\kappa_j}\sqrt{\kappa\,\Phi}\,\frac{\zeta\,e^{\gamma(T-t)}+e^{-\gamma(T-t)}}{\zeta\,e^{\gamma(T-t)}-e^{-\gamma(T-t)}}\,.\label{eqn:FunctionG}
    \end{align}
    The optimal change of measure is the admissible control given in feedback form by
    \begin{align}
        y^{S*}_t = \psi_I\,\sigma\,Q^{I,\omega^*}_t\,.\label{eqn:OptimalMeasure}
    \end{align}
\end{theorem}
\begin{proof}
    This follows analogously to the proof of Theorem 2.2 in \cite{cartea2025brokers} where multivariate controls are substituted as appropriate.
\end{proof}
\begin{remark}
    Taking a closer look at functions in \eqref{eqn:FunctionG}, we note that $\omega^j_1(t)$ is always negative. This means that in the informed trader's optimal strategy \eqref{eqn:OptimalInformedTradingSpeed}, the second term $\omega^j_1(t)\,Q^{I,\omega^*}_t$ is responsible for managing the risk associated with holding inventory through time and limiting the liquidation cost at maturity. Furthermore, as time approaches maturity, function $\omega^j_0(t)$ converges to zero. Hence, near the terminal time, the informed trader's strategy is dominated by inventory liquidation, with diminishing emphasis on exploiting the alpha component. \label{rem:FunctionG}
\end{remark}

Considering the earlier discussion before \eqref{eqn:InformedValueFunction}, we observe that the value function and optimal controls are well defined for all positive finite $\psi_I$. However, allowing $\psi_I\rightarrow\infty$ would introduce the problem where the optimal drift would approach $\pm \infty$ based on the sign of $Q^{I,\omega^*}_t$. To address this, one can place large but finite bounds on the drift which permits an infinite ambiguity parameter. While the exact form of the optimal controls in \eqref{eqn:OptimalInformedTradingSpeed} would change in this case, the overall qualitative behaviour of the optimal speeds would remain similar.

In addition, from \eqref{eqn:OptimalInformedTradingSpeedBeforeAnsatz} and \eqref{eqn:OptimalInformedTradingSpeed} we can observe that
\begin{align}
    \kappa_i\,\omega^{i*}_t = \kappa_j\,\omega^{j*}_t\, \quad \mbox{ for all } i,j\in\{1,\dots,N\}\,.\label{eqn:proportional_trades}
\end{align}
This means that under the optimal strategy the informed trader's temporary price impact of trades to each broker, given by $\kappa_{i} \omega_t^{i}$ for $i\in\{1,\dots,N\}$, and, in turn, transaction prices at all brokers, are equal. Furthermore, denoting the total trading speed process of the informed trader by $(\omega_t)_{t\in\mathcal{T}}$, we observe under the optimal strategy
\begin{align}
    \omega^{*}_t = \sum^N_{i=1}\omega^{i*}_t=\left(1+\sum^N_{i\neq j}\frac{\kappa_j}{\kappa_i}\right)\,\omega^{j*}_t\,, \quad \forall\, j\in\{1,\dots,N\}\,.\label{eqn:total_speed}
\end{align}

Thus, when broker $j$ supplies liquidity to the informed trader, she is able to deduce the total trading speed of the informed trader, and can therefore infer the informed trader's inventory by integrating the total flow.\footnote{This requires the assumption that broker $j$ knows the informed trader's initial inventory, which is stated below and is subjected to a robustness analysis in a later section.}. Additionally, because trades submitted to each broker by the informed trader are proportional as in equation \eqref{eqn:proportional_trades}, each broker can deduce the trading speed that the insider submits to the other brokers.

\subsection{Brokers' Strategy}\label{subsec:BrokerStratgy_multi}

We consider $N$ brokers who stream quotes to the informed and uninformed traders. We assume that the brokers know the informed trader’s model and the value of the model parameters, and the initial inventory of the informed trader\footnote{The initial inventory of the informed trader to be known by the brokers is a mild restriction in practice, as traders typically aim to end the day with a flat position.}. The brokers observes the mid-price process $(S^\nu_t)_{t\in\mathcal{T}}$ of the LOB in the lit market, but do not directly observe the alpha component \eqref{eqn:dAlpha}. The signal is costly to obtain. The brokers `pay' for information to learn this signal by trading with the informed trader. Based on the discussion after \eqref{eqn:proportional_trades} the brokers are able to deduce the alpha component by observing the order flow that they receive from the informed trader. In particular, broker $i$ assumes that the informed trader behaves optimally and derives the form of the informed trader’s optimal strategy in \eqref{eqn:OptimalInformedTradingSpeed}. With the informed trading speed $\omega^{i*}_t$ submitted to her, she computes the total informed trading speed $\omega^*_t$ via \eqref{eqn:total_speed}, and then infers the informed trader's inventory as
$Q_t^{I,\omega^*} = Q^{I,\omega^*}_0+\int^t_0\,\omega^*_u\,du$. Equipped with this information, it is straightforward to extract the alpha component of the mid-price as
\begin{align}
    \alpha_t = \frac{\omega^{i*}_t-\omega^i_1(t)\,Q^{I,\omega^*}_t}{\omega^i_0(t)}\,.\label{eqn:AlphaExtraction}
\end{align}
This broker $i$ also provides liquidity to the uninformed trader whose rate of trading follows the Ornstein–Uhlenbeck process
\begin{align}
    du^i_t=-\theta_i\,u^i_t\,dt+\eta_i\,dW^{U,i}_t\,,\quad u^i_0\in\mathbb{R}\,,\label{eqn:uninformed_trading}
\end{align}
where $\theta_i$ and $\eta_i$ are positive constants.\footnote{Recall $(W^{U,i}_t)_{t\in\mathcal{T}}$ is a Brownian motion independent of $W^S$ and $W^\alpha$, but $W^{U,i}$ is correlated to $W^{U,j}$ with parameter $\rho_{i,j}$.} We assume for simplicity that $(u^j_t)_{t\in\mathcal{T}}$ is also visible to broker $i$ for each $j\neq i$. When broker $i$ trades in the lit market, her cash process $(X^{i,\nu^i,\omega^i}_t)_{t\in\mathcal{T}}$ follows
\begin{align*}
    dX^{i,\nu^i,\omega^i}_t = u^i_t\left(S^\nu_t+c_i\,u^i_t\right)dt+\omega^i_t\left(S^\nu_t+\kappa_i\,\omega^i_t\right)dt-\nu^i_t\left(S^\nu_t+k_i\,\nu^i_t\right)dt\,,\quad X^{i,\nu^i,\omega^i}_0=0\,,
\end{align*}
and her inventory process $(Q^{i,\nu^i,\omega^i}_t)_{t\in\mathcal{T}}$ follows
\begin{align*}
    dQ^{i,\nu^i,\omega^i}_t=\left(\nu^i_t-u^i_t-\omega^i_t\right)dt\,,\quad Q^{i,\nu^i,\omega^i}_0=0\,.
\end{align*}
We aim to find an equilibrium of trading strategies between brokers. Each broker assumes that the others are acting using a Markov strategy, and since each broker observes all uninformed trading, an individual broker is able to deduce the trading speeds of all other brokers in the lit market. Thus, we may assume each broker knows the cash, inventory, and trading speed of other brokers. In this way, the filtration $(\mathcal{F}^i_t)_{t\in\mathcal{T}}$ of broker $i$ is given by 
\begin{align*}
    \mathcal{F}^i_t\coloneqq\sigma\left[\left\{S^\nu_u\right\}_{u\le t}, \left\{\alpha_u\right\}_{u\le t}, \left\{X^{j,\nu^j,\omega^j}_u\right\}_{u\le t}, \left\{Q^{I,\omega}_u\right\}_{u\le t}, \left\{Q^{j,\nu^j,\omega^j}_u\right\}_{u\le t}, \left\{u^j_u\right\}_{u\le t}:j\in\{1,\dots,N\}\right]\,,
\end{align*}
and she works in a completed filtered probability space $(\Omega,\mathcal{F},(\mathcal{F}_t)_{t\in\mathcal{T}},\mathbb{P})$ under which $(S^\nu_t)_{t\in\mathcal{T}}$ and $(\alpha_t)_{t\in\mathcal{T}}$ satisfy \eqref{eqn:dS} and \eqref{eqn:dAlpha}, respectively. The set of admissible strategies is given by
\begin{align*}
    \mathcal{A}^i\coloneqq\left\{\nu^i\coloneqq\left(\nu^i_t\right)_{t\in\mathcal{T}}\big\vert\nu^i\text{ is a Markov control and }\mathbb{E}\left[\int^T_0\left(\nu^i_s\right)^2 ds\right]<\infty\right\}\,.
\end{align*}
The value function of broker $i$ is given by
\begin{align}
    \begin{split}
        &\hspace{2.5mm}H^i\left(t,\alpha,S,q^I, \bs{q}^B, \bs{x}^B, \bs{u}; \nu^{-i}\right)\\
        &=\sup_{\nu^i\in\mathcal{A}^i}\mathbb{E}_{t,\alpha,S,q^I,\bs{q}^B, \bs{x}^B, \bs{u}}\left[X^{i,\nu^i,\omega^i}_T+Q^{i,\nu^i,\omega^i}_T\,S^\nu_T-a_i\left(Q^{i,\nu^i,\omega^i}_T\right)^2-\phi_i\int^T_t\left(Q^{i,\nu^i,\omega^i}_u\right)^2 du\right]\,,
    \end{split}\label{eqn:BrokerValueFunction}
\end{align}
where $\nu^{-i}$ represent a fixed collection of strategies for all brokers other than broker $i$, $a_i\ge0$ is a terminal liquidation parameter, $\phi_i\ge0$ is a running inventory penalty parameter, and $\bs{q}^B$, $\bs{x}^B$, and $\bs{u}$ represent the collection of all $q^j$, $x^j$, and $u^j$ for $j\in\{1,\dots,N\}$, respectively. As the informed trader adopts the strategy of the form given in \eqref{eqn:OptimalInformedTradingSpeed}, his trading speed $\omega^j$ at time $t$ to broker $j$ is a function of $t$, $\alpha_t$, and $Q^I_t$. With a slight abuse of notation we denote this function by $\omega^j(t,\alpha,Q^I)$. The associated HJB equation for broker $i$ is
\begin{align}
    \begin{split}
        &\partial_tH^i - \phi_i (q^i)^2 + \mathcal{L}^\alpha H^i + \mathcal{L}^{u}H^i - \sum_{j=1}^N (u^j + \omega^j)\partial_{q^j}H^i\\
        &+ \alpha \partial_SH^i + \frac{1}{2}\sigma^2\partial^2_{S}H^i + \sum_{j=1}^N \biggl(u^j(S + c_ju^j) + \omega^j(S + \kappa_j\omega^j) \biggr)\partial_{x^j}H^i\\ 
        &+  \sum_{j=1}^N\omega^j\partial_{q^I}H^i  + \sup_{\nu^i}\biggl\{ \sum_{j=1}^N \biggl(\nu^j\partial_{q^j}H^i + b_j\nu^j\partial_SH^i - \nu^j(S+k_j\nu^j)\partial_{x^j}H^i\biggr) \biggr\} = 0\,,\label{eqn:HJB}
    \end{split}
\end{align}
with terminal condition 
$H^i(T,\alpha,S,q^I,\bs{q}^B,\bs{x}^B,\bs{u}; \nu^{-i})=x^i+S\,q^i-a_i(q^i)^2$, 
where the infinitesimal operator $\mathcal{L}^\alpha$ of the alpha component is in \eqref{eqn:InfinitesimalOperatorAlpha} and 
\begin{align*}
    \mathcal{L}^{u} = -\sum_{j=1}^N\theta_j\,u^j\,\partial_{u^j}+\frac{1}{2}\sum_{j=1}^N\sum_{\ell = 1}^N\eta_j\,\eta_\ell\,\rho_{j,\ell}\,\partial^2_{u^j,u^\ell}\,,
\end{align*}
is the infinitesimal generator of the collection of trading rates $\{u^j\}$ for $j=1,\dots,N$ of the uninformed trader. The optimizer within \eqref{eqn:HJB} is given by
\begin{align}
    \nu^{i*}=\frac{\partial_{q^i}H^i-S\,\partial_{x^i}H^i+b_i\,\partial_S H^i}{2\,k_i\,\partial_{x^i}H^i}\,.\label{eqn:nu_i_star}
\end{align}
Given that all brokers are attempting to maximize their performance simultaneously, we substitute expressions analogous to \eqref{eqn:nu_i_star} for all brokers into the HJB equation. That is, for $i\neq j$ we take
\begin{align}
    \nu^{j*}=\frac{\partial_{q^j}H^j-S\,\partial_{x^j}H^j+b_j\,\partial_S H^j}{2\,k_j\,\partial_{x^j}H^j}\,.\label{eqn:nu_j_star}
\end{align}
Substituting \eqref{eqn:nu_i_star} and \eqref{eqn:nu_j_star} into \eqref{eqn:HJB} and making the ansatz
\begin{align}
    H^i(t,\alpha,S,q^I,\bs{q}^B,\bs{x}^B,\bs{u}; \nu^{-i*}) &= x^i + q^i\,S + h^i(t, \alpha, q^I, \bs{q}^B, \bs{u})\label{eqn:BrokerAnsatz1}
\end{align}
results in a system of PDEs for $h^i$, $i\in\{1, \dots, N\}$ given by
\begin{align}
    \begin{split}
    \partial_t h^i + \alpha q^i - \phi_i(q^i)^2 - \sum_{j=1}^N (u^j + \omega^j )\partial_{q^j}h^i + \kappa_i(\omega^i)^2 + c_i(u^i)^2 + \mathcal{L}^\alpha h^i + \sum_{j=1}^N\mathcal{L}^{u^j}h^i \hspace{10mm}\\
    + \sum_{j=1}^N \omega^j \partial_{q^I}h^i + \sum_{j=1}^N \frac{1}{2k_j}(\partial_{q^j}h^j + b_j q^j)(\partial_{q^j}h^i + b_j q^i) -\frac{1}{4k_i}(\partial_{q^i}h^i + b_i q^i)^2 = 0\,,
    \end{split}\label{eqn:HJB_h}
\end{align}
with terminal condition $h^i(T, \alpha, q^I, \bs{q}^B, \bs{u}) = -a_i(q^i)^2$.

\begin{proposition}\label{prop:Broker}
    Define
    \begin{align*}
        \delta_{j,k} &= \left\{\begin{array}{rr} 1\,, & j=k\\ 0\,, &j\neq k\end{array}\right.\,, & \delta_{j,k,\ell} &= \left\{\begin{array}{rr} 1\,, & j=k=\ell \\ 0\,, & \mbox{otherwise}\end{array}\right.\,,
    \end{align*}
    and suppose that for all $i,r,s\in\{1,\dots,N\}$ and $t\in[0,T]$ the functions $f^i$, $g^i$, $m^i$, $n^i_{r,s}$, $p^i_{r,s}$, $d^i_{r}$, $v^i$, $w^i_{r}$, $x^i_{r}$, $y^i_{r,s}$, and $z^i_{r}$ satisfy the system of ODEs
    {\footnotesize
    \begin{align}\label{eqn:big_ODE}
        \begin{split}
            \partial_tf^i + \eta^2\, m^i + \frac{1}{2}\sum_{j=1}^N \sum_{\ell=1}^N\rho_{j,\ell}\,\eta_j\,\eta_\ell\, P_{j,\ell}^i &= 0\,,\\
            \partial_tg^i + \kappa_i(\omega_1^i)^2 - \frac{1}{4k_i}(d_i^i)^2 + \sum_{j=1}^N\biggl( 2\omega_1^jg^i + \frac{1}{2k_j}d_j^id_j^j - \omega_1^jd_j^i\biggr) &= 0\,,\\
            \partial_t m^i + \kappa_i(\omega_0^i)^2 - 2\theta m^i - \frac{1}{4k_i}(x_i^i)^2 + \sum_{j=1}^N \biggl( \omega_0^jv^i + \frac{1}{2k_j}x_j^i x_j^j - \omega_0^j x_j^i\biggr) &= 0\,,\\
            \partial_t n_{r,s}^i - \biggl(\phi_i + \frac{b_i^2}{4k_i}\biggr)\delta_{r,s,i} + \biggl(\frac{b_r^2}{2k_r} - \frac{b_i}{2k_i}N_{r,i}^i + \sum_{j=1}^N \frac{b_j}{2k_j}N_{r,j}^j\biggr)\delta_{s,i}\hspace{20mm}\\
            + \frac{b_s}{2k_s}N_{r,s}^i - \frac{1}{4k_i}N_{r,i}^iN_{s,i}^i + \sum_{j=1}^N \frac{1}{2k_j}N_{r,j}^iN_{s,j}^j &= 0\,,\\
            \partial_tp_{r,s}^i + c_i\delta_{r,s,i} - \theta_sP_{r,s}^i - y_{s,r}^i - \frac{1}{4k_i}y_{i,r}^iy_{i,s}^i + \sum_{j=1}^N\frac{1}{2k_j}y_{r,j}^iy_{j,s}^j &= 0\,,\\
            \partial_t d_r^i + \biggl(\sum_j \frac{b_j}{2k_j}d_j^j- \frac{b_i}{k_i}d_i^i\biggr)\delta_{r,i} + \frac{b_r}{2k_r}d_r^i - \frac{1}{2k_i}N_{r,i}^i d_i^i\hspace{20mm}\\
            + \sum_j \biggl(\omega_1^jd_r^i + \frac{1}{2k_j}N_{r,j}^id_j^j + \frac{1}{2k_j}d_j^iN_{r,j}^j - \omega_1^jN_{r,j}^i\biggr) &= 0\,,\\
             \partial_t v^i + 2\kappa_i \omega_0^i\omega_1^i - \theta v^i + \omega_1^jv^i + \frac{1}{2k_j} d_j^ix_j^j + \frac{1}{2k_j}d_j^jx_j^i\hspace{20mm}\\
            - \frac{1}{2k_i}d_i^ix_i^i + \sum_j \biggl(2\omega_0^jg^i - \omega_0^j d_j^i -\omega_1^jx_j^i\biggr) &= 0\,,\\
            \partial_t w_r^i - d_r^i - \frac{1}{2k_i}d_i^iy_{i,r}^i - \theta_r w_r^i + \sum_j \biggl(\omega_1^jd_r^i + \frac{1}{2k_j}d_j^i y_{j,r}^i + \frac{1}{2k_j}d_j^jy_{j,r}^i - \omega_1^j y_{j,r}^i\biggr) &= 0\,,\\
            \partial_t x_r^i + \biggl(1 - \frac{b_i}{2k_i}x_i^i + \sum_j \frac{b_j}{2k_j} x_j^j\biggr)\delta_{i,r} - \theta x_r^i + \frac{b_r}{2k_r}x_r^i - \frac{1}{2k_i}N_{r,i}^ix_i^i\hspace{20mm}\\
            + \sum_j \biggl(\omega_0^jd_r^i - \omega_0^jN_{r,j}^i + \frac{1}{2k_j}N_{r,j}^i x_j^j + \frac{1}{2k_j}N_{r,j}^jx_j^i\biggr) &= 0\,,\\
            \partial_t y_{r,s}^i - \frac{b_i}{2k_i}y_{i,s}^i \delta_{r,i} - N_{r,s}^i - \theta_s y_{r,s}^i - \frac{1}{2k_i}N_{r,i}^iy_{i,s}^i + \frac{b_r}{2k_r}y_{r,s}^i\hspace{20mm}\\
            + \sum_j \biggl(\frac{b_j}{2k^j}y_{j,s}^j\delta_{r,i} + \frac{1}{2k_j}N_{r,j}^iy_{j,s}^j + \frac{1}{2k_j}N_{r,j}^jy_{j,s}^i\biggr) &= 0\,,\\
            \partial_t z_r^i - x_r^i - \theta z_r^i - \theta_rz_r^i - \frac{1}{2k_i}x_i^iy_{i,r}^i + \sum_j \biggl(\omega_0^jw_r^i - \omega_0^jy_{j,r}^i + \frac{1}{2k_j}x_j^iy_{j,r}^j + \frac{1}{2k_j}y_{j,r}^ix_j^j\biggr) &= 0\,,
        \end{split}
    \end{align}
    }
    where $\omega_0$ and $\omega_1$ are as in \eqref{eqn:FunctionG}, $N^i_{r,s} = n^i_{r,s} + n^i_{s,r}$, $P^i_{r,s} = p^i_{r,s} + p^i_{s,r}$, and where all terminal conditions are $0$ except $n_{r,s}^i(T) = -a_i\delta_{r,s,i}$. Then the system of PDEs in \eqref{eqn:HJB_h} has solution given by
    \begin{align}
        \begin{split}
            h^i(t, \alpha, q^I, \bs{q}^B, \bs{u})&= f^i(t) + g^i(t)(q^I)^2 + m^i(t) \alpha^2 + \sum_{r=1}^N\sum_{s=1}^Nn_{r,s}^i(t)q^rq^s\\
            &\hspace{5mm} + \sum_{r=1}^N\sum_{s=1}^Np_{r,s}^i(t)u^ru^s + \sum_{r=1}^Nd_r^i(t)q^Iq^r+ v^i(t)q^I\alpha + \sum_{r=1}^Nw_r^i(t)q^Iu^r\\
            &\hspace{10mm} + \sum_{r=1}^N x_r^i(t)q^i\alpha + \sum_{r=1}^N\sum_{s=1}^N y_{r,s}^i(t)q^ru^s + \sum_{r=1}^Nz_r^i(t)\alpha u^r\,.
        \end{split}
        \label{eqn:BrokerAnsatz2}
    \end{align}
    
\end{proposition}

\begin{proof}
    This is checked by direct substitution of \eqref{eqn:BrokerAnsatz2} into \eqref{eqn:HJB_h} and grouping terms by their dependence on the state variables $\alpha$, $q^I$, $\bs{q}^B$, and $\bs u$. \qed
\end{proof}
\begin{theorem}[Verification Theorem for broker's optimization] Let $H^i$ be given by \eqref{eqn:BrokerAnsatz1} where $h^i$ is as in Proposition \ref{prop:Broker}. Then $H^i$ coincides with the value function defined in equation \ref{eqn:BrokerValueFunction}, and the optimal feedback control for broker $i$ is given by
    \begin{align}
        \begin{split}
            \nu^{i*}(t, \alpha, q^I, \bs{q}^B, \bs{u}) &= \frac{x_i^i}{2k_i}\alpha + \frac{b_i}{2k_i}q^i + \frac{d_i^i}{2k_i}q^I + \sum_{j=1}^N \frac{n_{j,i}^i + n_{i,j}^i}{2k_i}q^j + \sum_{j=1}^N\frac{y_{i,j}^i}{2k_i}u^j\,.
        \end{split}
        \label{eqn:OptimalBrokerTradingSpeed}
    \end{align}
\end{theorem}
\begin{proof}
    The function $H^i$ in \eqref{eqn:BrokerAnsatz1} is continuous on $\mathcal{T}$ with continuous first and second partial derivatives with respect to spacial variables, and additionally satisfies a quadratic growth condition. The function $\nu^{i*}$ in \eqref{eqn:OptimalBrokerTradingSpeed} is measurable and achieves the maximum in \eqref{eqn:HJB}. Under these controls the SDE's for $\alpha_t$, $S^{\nu^*}_t$, $Q^{I,\omega^*}_t$, $Q^{i,\nu^{i*},\omega^{i*}}_t$, $X^{i,\nu^{i*},\omega^{i*}}_t$, $u^i_t$ admit unique solutions. We then apply Theorem 3.5.2 of \cite{pham2009continuous} and conclude that $H^i$ is the value function and $\nu^{i*}$ is optimal. \qed
\end{proof}
\begin{remark}\label{rem:CoefficientsFunctionBroker}
    As $t\to T$, the coefficient functions $x^i_{i}$, $d^i_{i}$, $y^i_{i,j}$ and the cross-inventory terms $n^i_{j,i}$ and $n^i_{i,j}$ for 
    $j\neq i$ vanish, whereas the own-inventory coefficient $n^i_{i,i}$ does not. In line with the informed trader's behavior described in remark \ref{rem:FunctionG}, brokers therefore prioritize inventory liquidation over alpha exploitation near maturity. Moreover, because the coefficients multiplying $q^j(j\neq i)$ and $u^j$ go to zero, the influence of other brokers and of the uninformed trading flow on broker $i$'s control also becomes negligible as time approaches maturity.
\end{remark}

\section{Numerical Experiments}\label{sec:numerical}

In this section we showcase the performance of the strategies\footnote{To compute the trading strategies of the brokers, we solve the system of ODEs in Proposition \ref{prop:Broker} numerically using an explicit finite difference method with $5000$ time steps.} in the case of there being two brokers. The brokers are symmetric except for the transaction costs charged to the informed trader. We choose $\kappa_2=2\cdot \kappa_1$, i.e. the informed trader needs to pay a higher liquidity cost when trading with the second broker. Additionally, in order to focus on relationships between certain processes, we reduce the influence of the uninformed trading speeds by setting them equal ($u_t:=u^1_t=u^2_t$).\footnote{This also necessitates $\theta_1=\theta_2$, $\eta_1=\eta_2$, and $\rho_{1,2} = \rho_{2,1} = 1$.}

\begin{figure}
    \centering
    \includegraphics[width=0.32\textwidth]{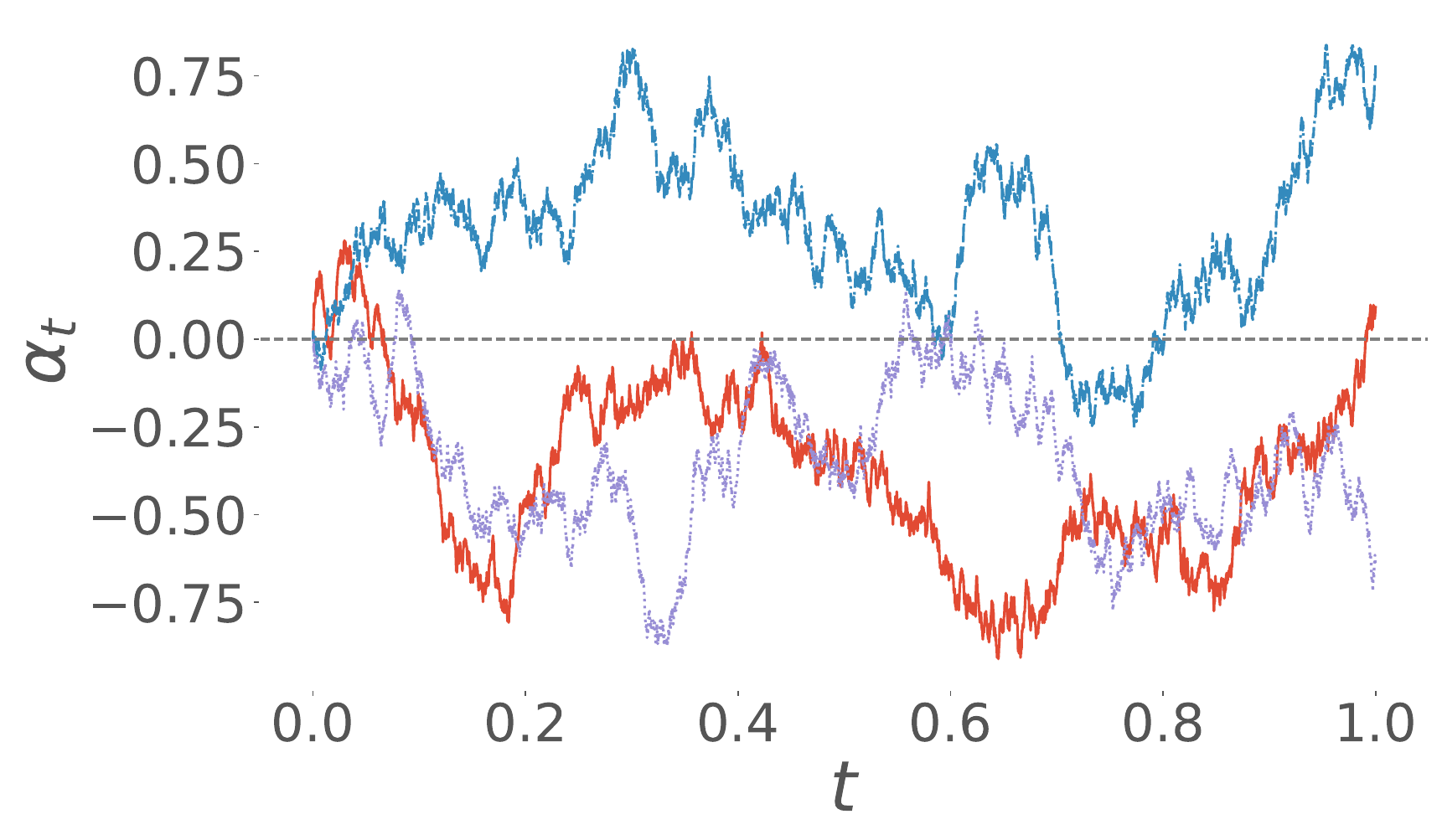}
    \includegraphics[width=0.32\textwidth]{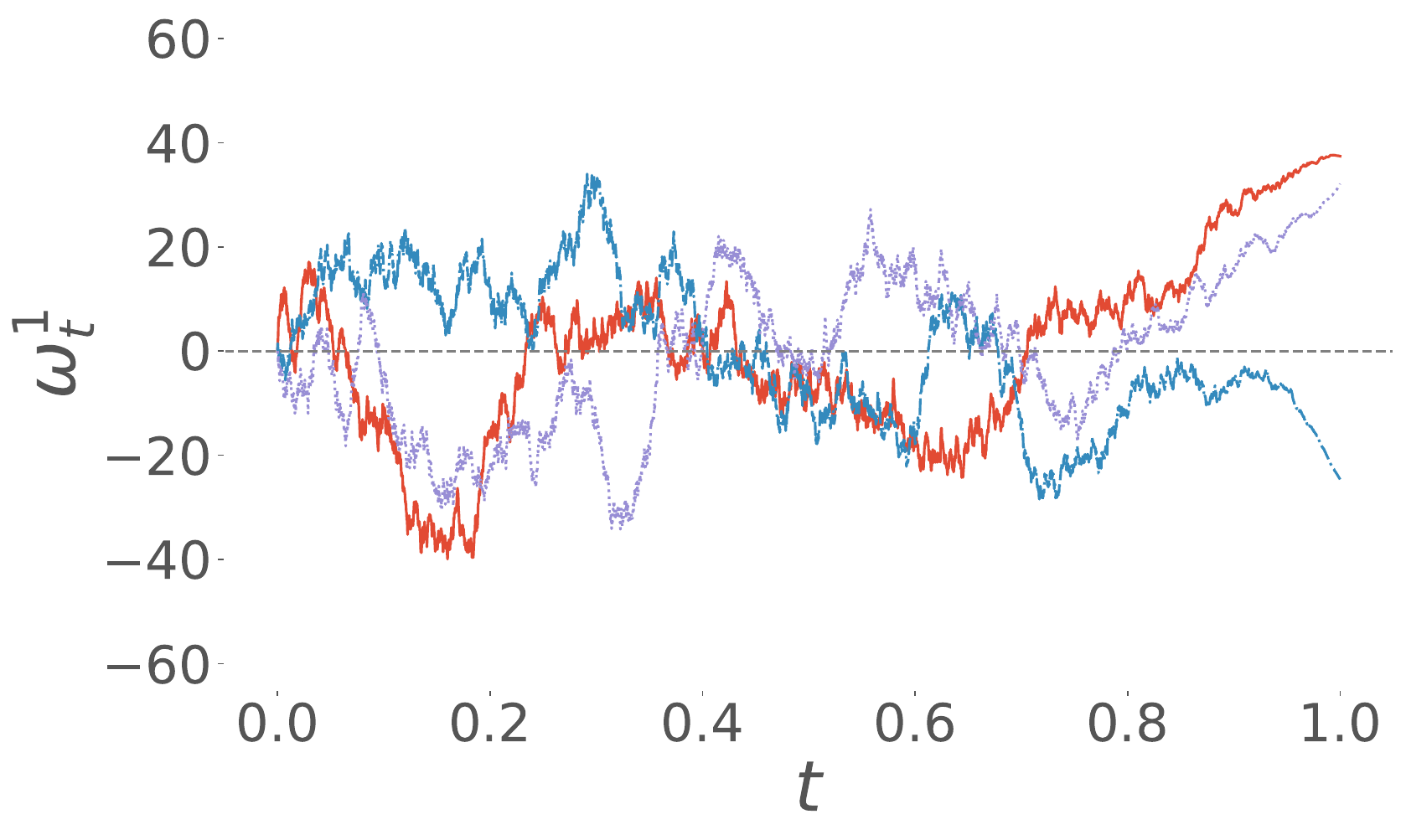}
    \includegraphics[width=0.32\textwidth]{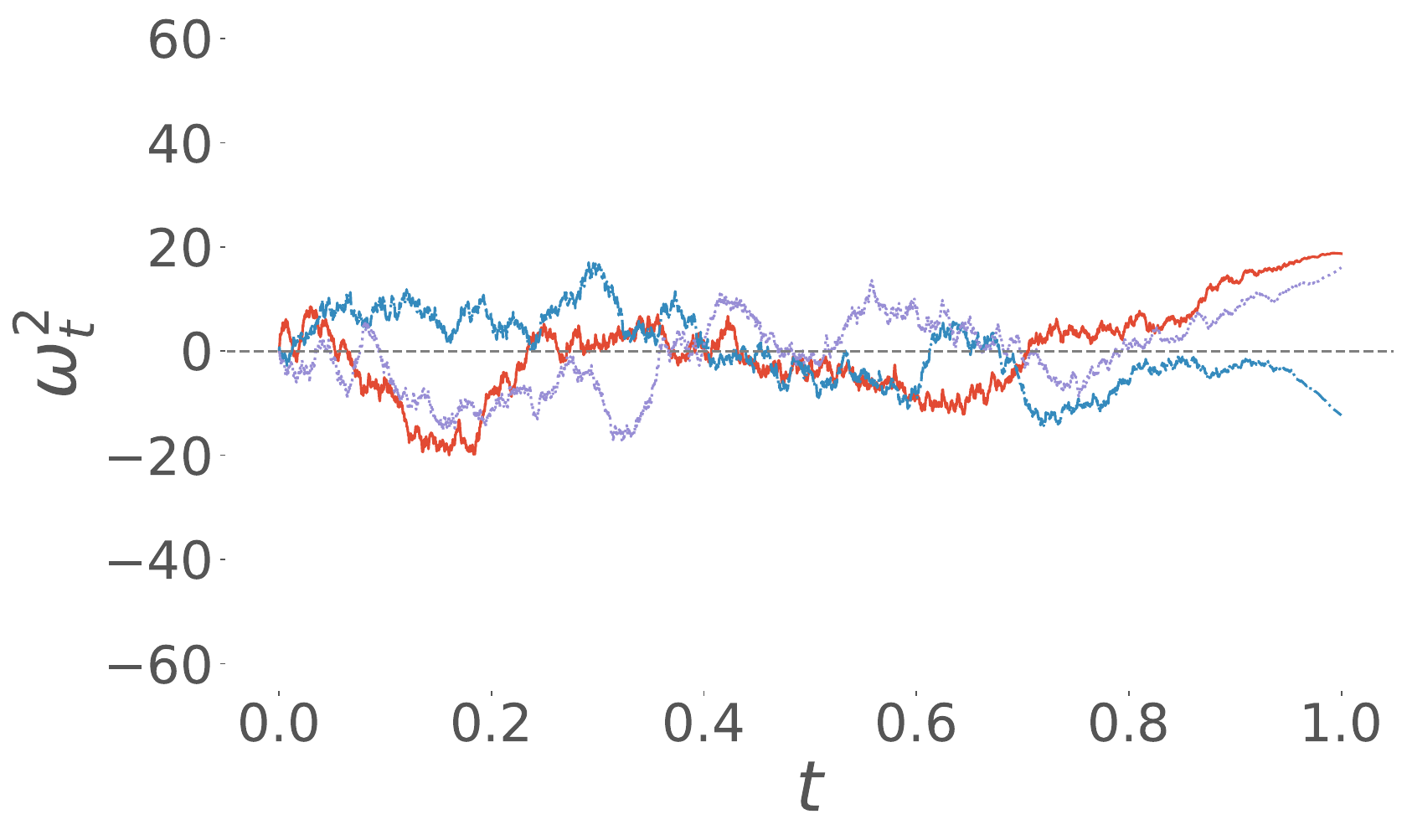}
    \includegraphics[width=0.32\textwidth]{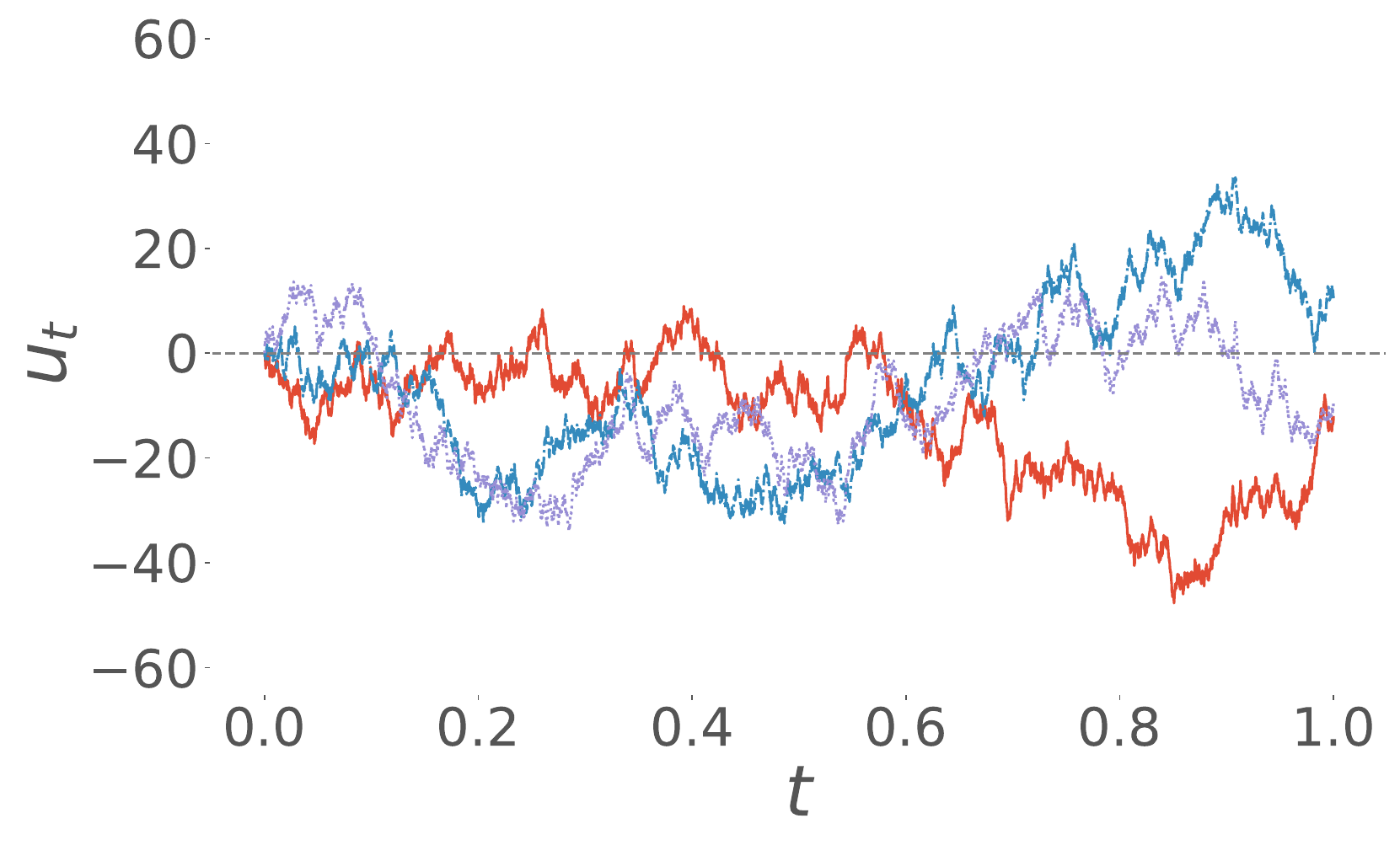}
    \includegraphics[width=0.32\textwidth]{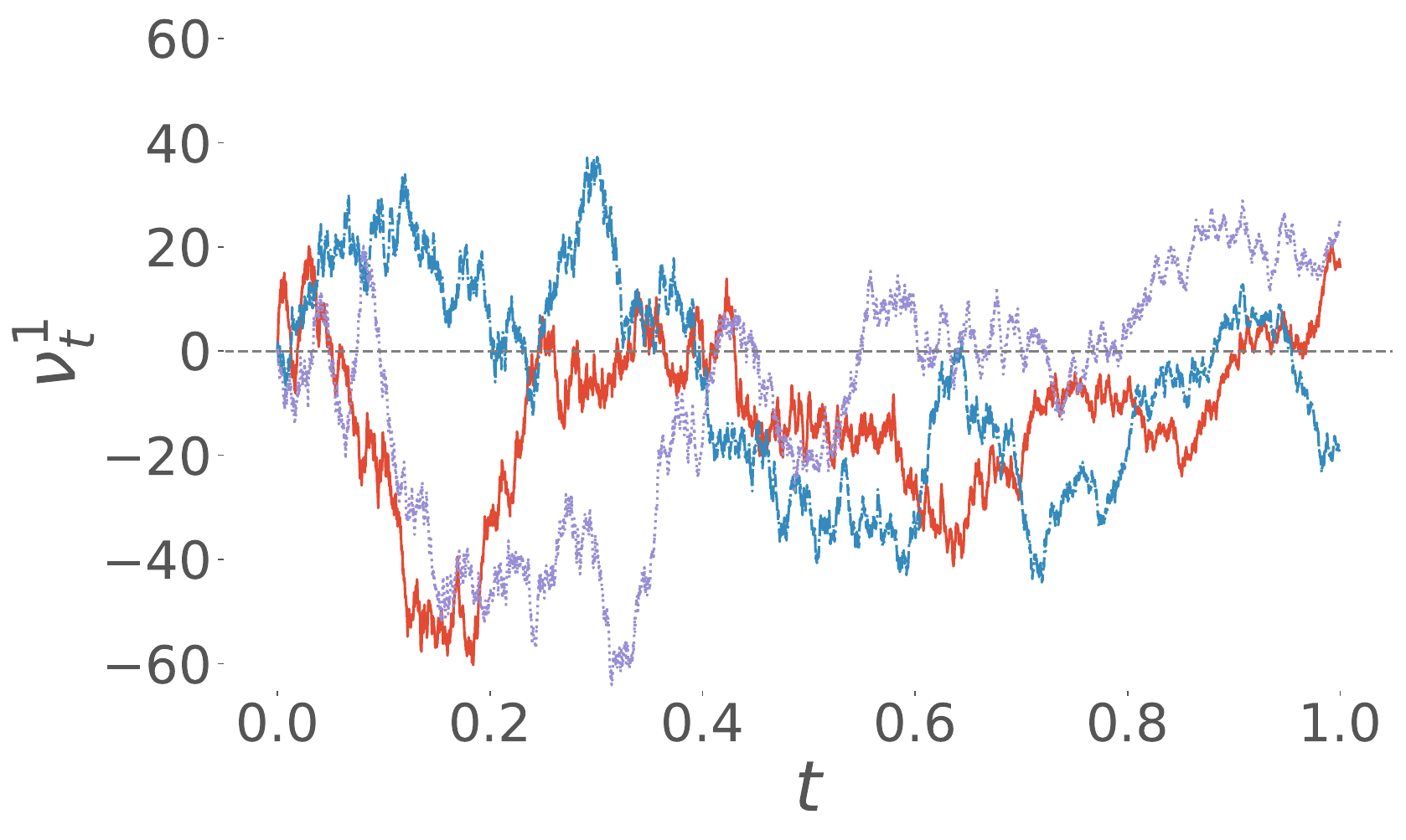}
    \includegraphics[width=0.32\textwidth]{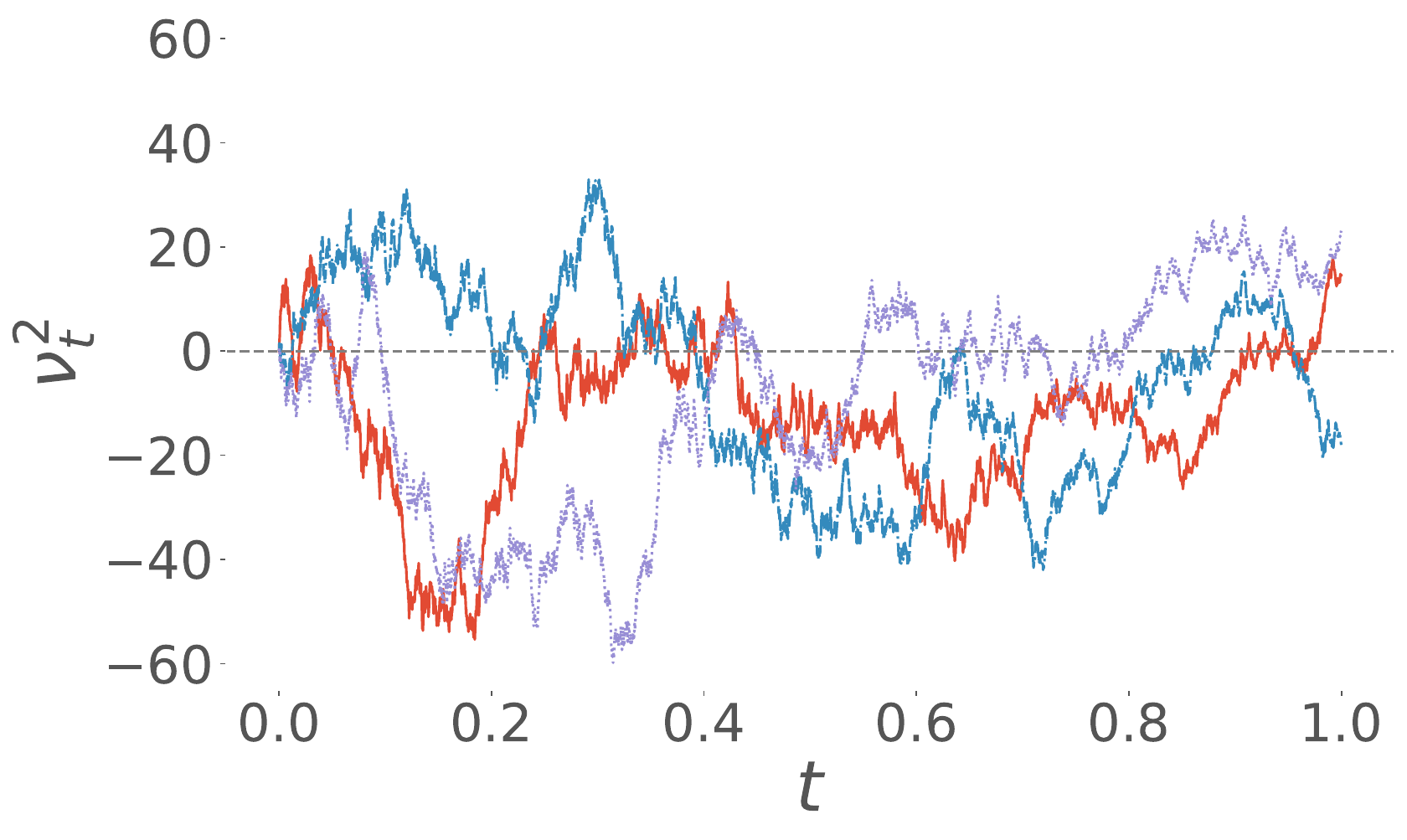}
    \includegraphics[width=0.32\textwidth]{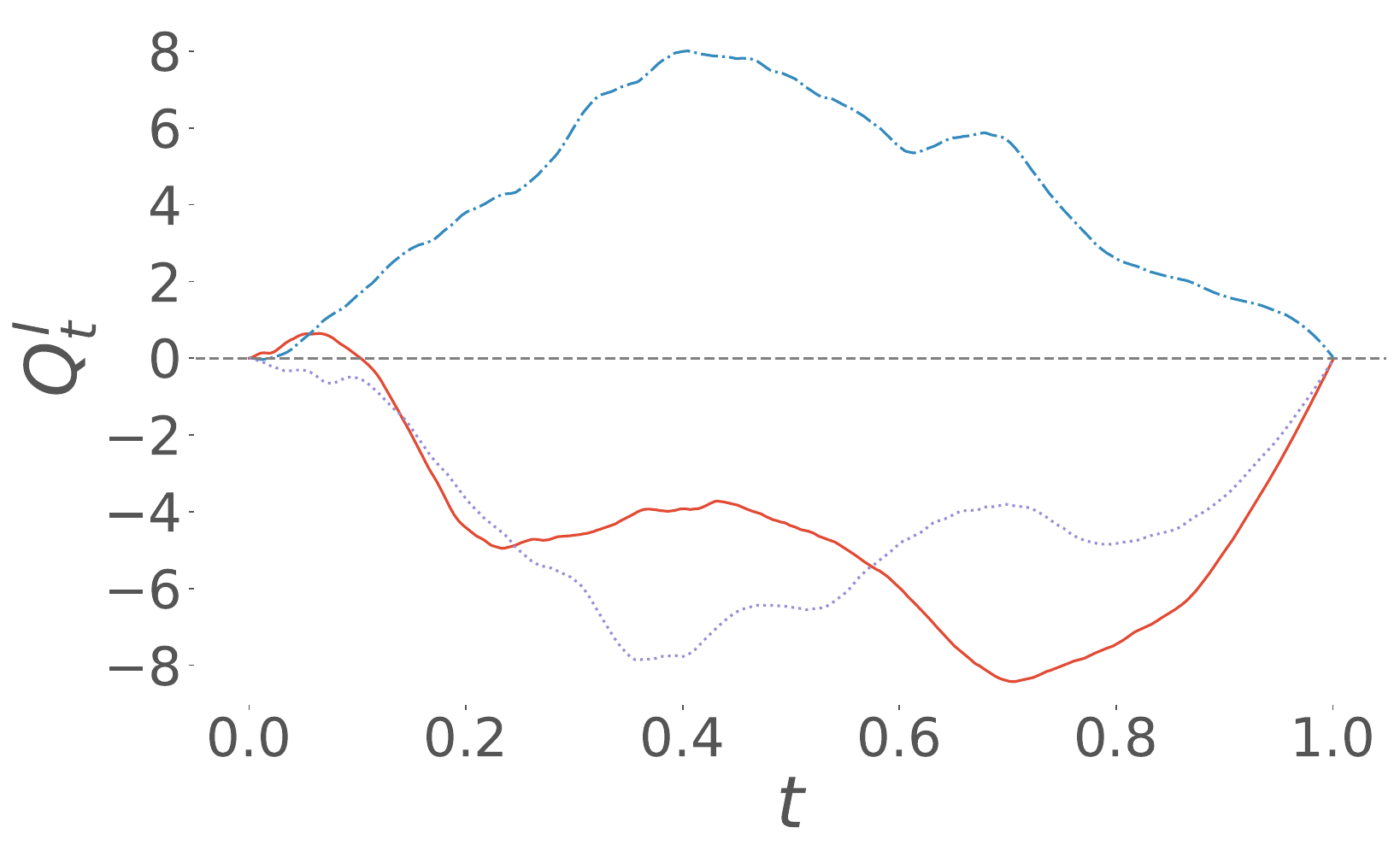}
    \includegraphics[width=0.32\textwidth]{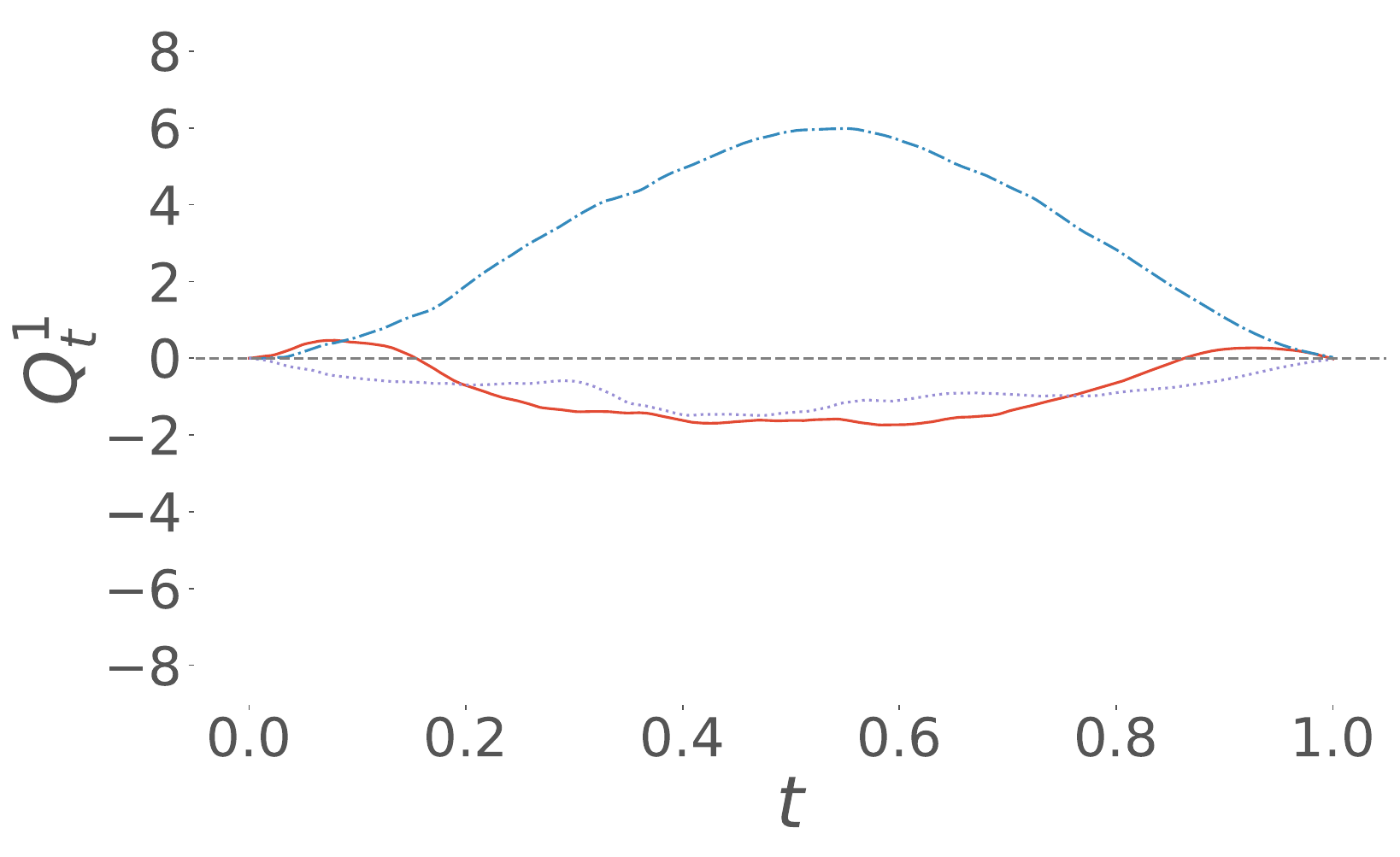}
    \includegraphics[width=0.32\textwidth]{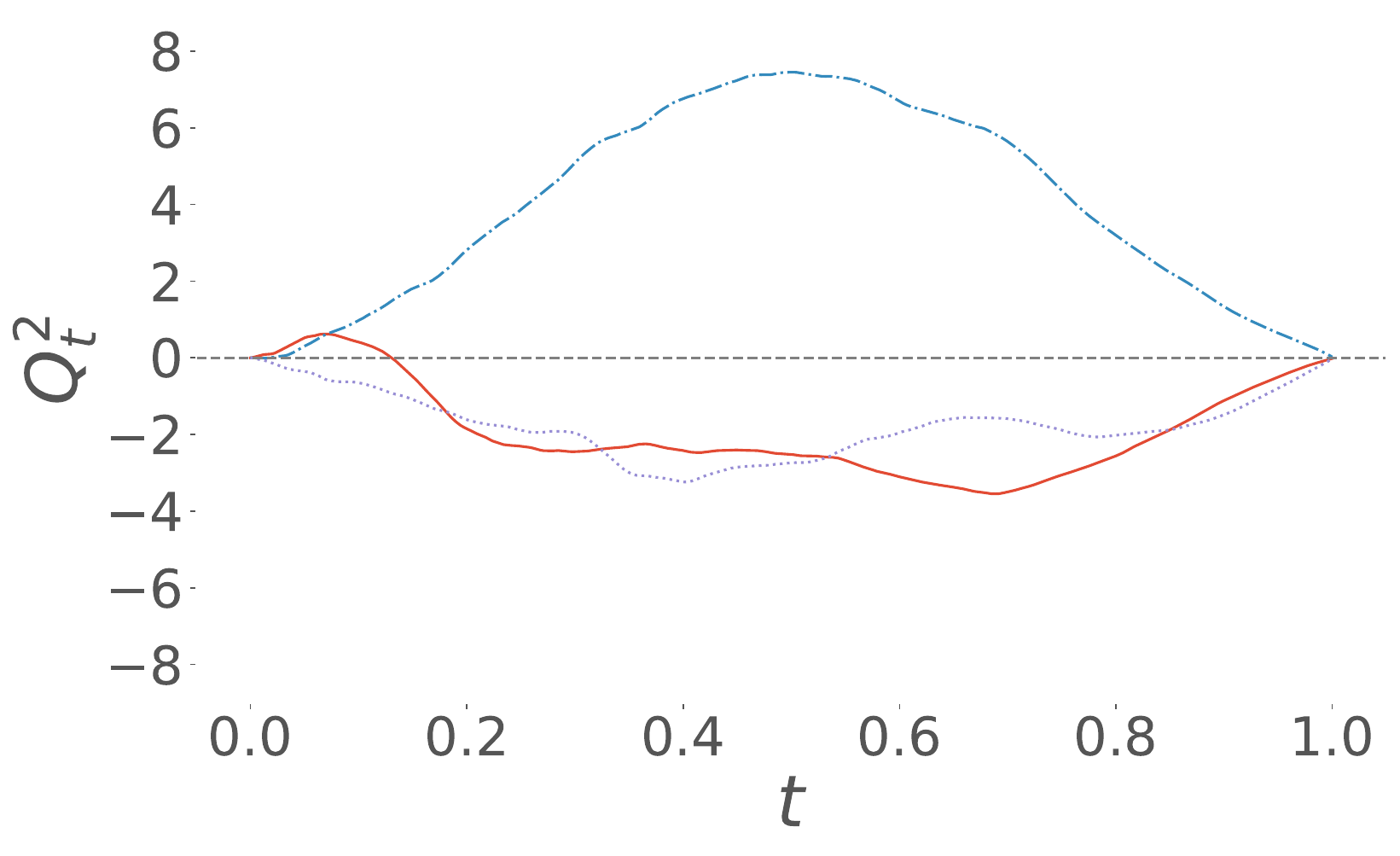}
    \caption{Sample paths for alpha $\alpha_t$, noise trading speed $u_t$, informed trading speed $\omega^{1,2}_t$, brokers' trading speed $\nu^{1,2}_t$, informed inventory $Q^I_t$, and brokers' inventory $Q^{1,2}_t$. Relevant parameter values are $S_0=100$, $\alpha_0 = 0$, $\theta = 3$, $\eta = 1$, $\sigma = 1$, $\kappa_1 = 0.001$, $\kappa_2 = 0.002$, $c_1 = c_2 = 0.001$, $k_1 = k_2 = 0.0012$, $b_1 = b_2 = 0.001$, $\phi_1 = \phi_2 = \phi_I = \psi_I = 0.01$, $a_1 = a_2 = a_I = 1$, $u^{1}_0 = u^{2}_0 = 0$, $\theta_1 = \theta_2 = 5$, $\eta_1 = \eta_2 = 50$, $Q^I_0 = Q^{1}_0 = Q^{2}_0 = 0$.}
    \label{fig:G-SamplePath}
\end{figure}

Figure \ref{fig:G-SamplePath} shows three sample paths for many processes relevant to the equilibrium. We observe that the informed trader's trading speed $\omega^2_t$ with broker 2 is half of $\omega^1_t$ with broker 1, which is a numerical demonstration of the result in \eqref{eqn:proportional_trades}. Perhaps surprisingly, broker 1's trading speed $\nu^1_t$ is not significantly larger in magnitude than broker 2's trading speed $\nu^2_t$. However, broker 2 tends to have larger magnitude of inventory compared to broker 1 because she trades less with the informed trader, and so the net trading of broker 2 is generally larger in magnitude than that of broker 1. During the first half of the time horizon, the informed trader and the brokers exhibit similar behaviour in their trading rates as each of them try to take advantage of the alpha component. However, in the later stages of the trading period, the brokers tend to exhibit trading speeds which are quite different from the informed trader. This is because they are also subject to the trades of the uninformed trader and must manage the terminal liquidation penalty posed by that order flow at the cost of not being able to fully take advantage of the alpha component.

\subsection{Robustness Analysis}

In the setting of the previous section, brokers are assumed to have full knowledge of the informed trader's model and parameter values. In particular, the feedback strategy of the brokers depends on the inventory on the informed trader. Given the relation in \eqref{eqn:total_speed} and subsequent discussion, this knowledge depends on being able to accurately estimate the informed trader's initial inventory $Q_0^I$. Note that in Figure \ref{fig:G-SamplePath} the informed trader's inventory path ends very close to zero in each sample, and if this style of dynamic game is to be repeated in a subsequent trading interval then an initial inventory estimate of $Q_0^I=0$ is reasonable. However, the brokers should be concerned with potential risks associated with the situation that the informed trader begins with non-zero inventory. In this section, we perform a robustness analysis which quantifies the loss with respect to both the performance criterion and expected P\&L due to suboptimal trading that arises from misspecification of initial inventory.

We compute these losses by performing a Monte Carlo simulation with the following methodology. Let $Q_0^I$ be the initial inventory of the informed trader who trades according to processes $\omega^{j*}$ in \eqref{eqn:OptimalInformedTradingSpeed}. Thus, the informed trader's inventory at all times is equal to
\begin{align*}
    Q_t^{I,\omega^*} &= Q_0^I + \int_0^t \omega_u^*\,du\,,\\
    \omega_t^* &= \sum_{j=1}^N\omega_t^{j*}\,,
\end{align*}
identical to the dynamic equilibrium demonstrated earlier. All brokers have an estimate\footnote{Each broker could have their own different estimate, but subsequently we will take all of the initial estimates to be zero and avoid the case of individual values.} of the informed trader's initial inventory denoted $\widehat{Q}_0^I$ and use this to construct an estimate of the full path by
\begin{align*}
    \widehat{Q}_t^{I,\omega^*} &= \widehat{Q}_0^I + \int_0^t \omega_u^*\,du\,.
\end{align*}
Estimating the inventory means that brokers also employ an estimate of the signal $\alpha$, which is denoted $\widehat{\alpha}$ and computed by an appropriately modified \eqref{eqn:AlphaExtraction}
\begin{align*}
    \widehat{\alpha}_t = \frac{\omega^{i*}_t-\omega^i_1(t)\,\widehat{Q}^{I,\omega^*}_t}{\omega^i_0(t)}\,.
\end{align*}
Using the estimates of inventory and signal, each broker then trades in a fashion that they believe results in equilibrium using the strategy
\begin{align*}
    \widehat{\nu}^i_t &= \nu^{i*}(t, \widehat{\alpha}_t, \widehat{Q}_t^{I,\omega^*}, \bs{Q}^{B,\widehat{\nu}}_t, \bs{u}_t)\,,
\end{align*}
where $\nu^{i*}$ is the feedback trading strategy given in \eqref{eqn:OptimalBrokerTradingSpeed}. With each broker using their prescribed strategy, we compute their performance criteria given by
\begin{align*}
    \widehat{H}^i &= \mathbb{E}\biggl[ X_T^{i,\widehat{\nu}^i,\omega^{i*}} + Q_T^{i,\widehat{\nu}^i}\,S_T^{\widehat{\nu}} - a_i\,(Q_T^{i,\widehat{\nu}^i})^2 - \phi_i\int_0^T (Q_t^{i,\widehat{\nu}^i})^2\,dt   \biggr]\,.
\end{align*}
As a point of comparison, we also compute their performance criteria in equilibrium when they correctly observe the initial inventory $Q_0^I$, given by
\begin{align*}
    H^{i*} &= \mathbb{E}\biggl[ X_T^{i,\nu^{i*},\omega^{i*}} + Q_T^{i,\nu^{i*}}\,S_T^{\nu^*} - a_i\,(Q_T^{i,\nu^{i*}})^2 - \phi_i\int_0^T (Q_t^{i,\nu^{i*}})^2\,dt   \biggr]\,.
\end{align*}
We also compute their expected P\&L under both the fully informed and misspecified scenarios given by
\begin{align*}
    {P\&L}^i &= \mathbb{E}\biggl[ X_T^{i,\nu^{i*},\omega^{i*}} + Q_T^{i,\nu^{i*}}\,S_T^{\nu^*}\biggr]\,,\\
    \widehat{P\&L}^i &= \mathbb{E}\biggl[ X_T^{i,\widehat{\nu}^i,\omega^{i*}} + Q_T^{i,\widehat{\nu}^i}\,S_T^{\widehat{\nu}}\biggr]\,.    
\end{align*}
In all computations, we take $\widehat{Q}_0^I = 0$, but we vary the actual initial inventory so that each performance metric above depends on $Q_0^I$. The range of values chosen are meant to represent typical inventory values of the informed trader as in the bottom left panel of Figure \ref{fig:G-SamplePath}. In Figure \ref{fig:RobustnessAnalysisValueFunction} we plot $H^{i*}-\widehat{H}^i$ and ${P\&L}^i-\widehat{P\&L}^i$ to demonstrate the effect of misspecification on performance in the two-broker scenario. In both cases we see that performance is maximized when the initial inventory is estimated correctly, as expected, and that the level of suboptimality is convex in the estimation error. This is an indication that it is important for the brokers to accurately estimate inventory, especially for large deviations from the typical starting value of zero.

\begin{figure}
    \centering
    \includegraphics[width=0.47\textwidth]{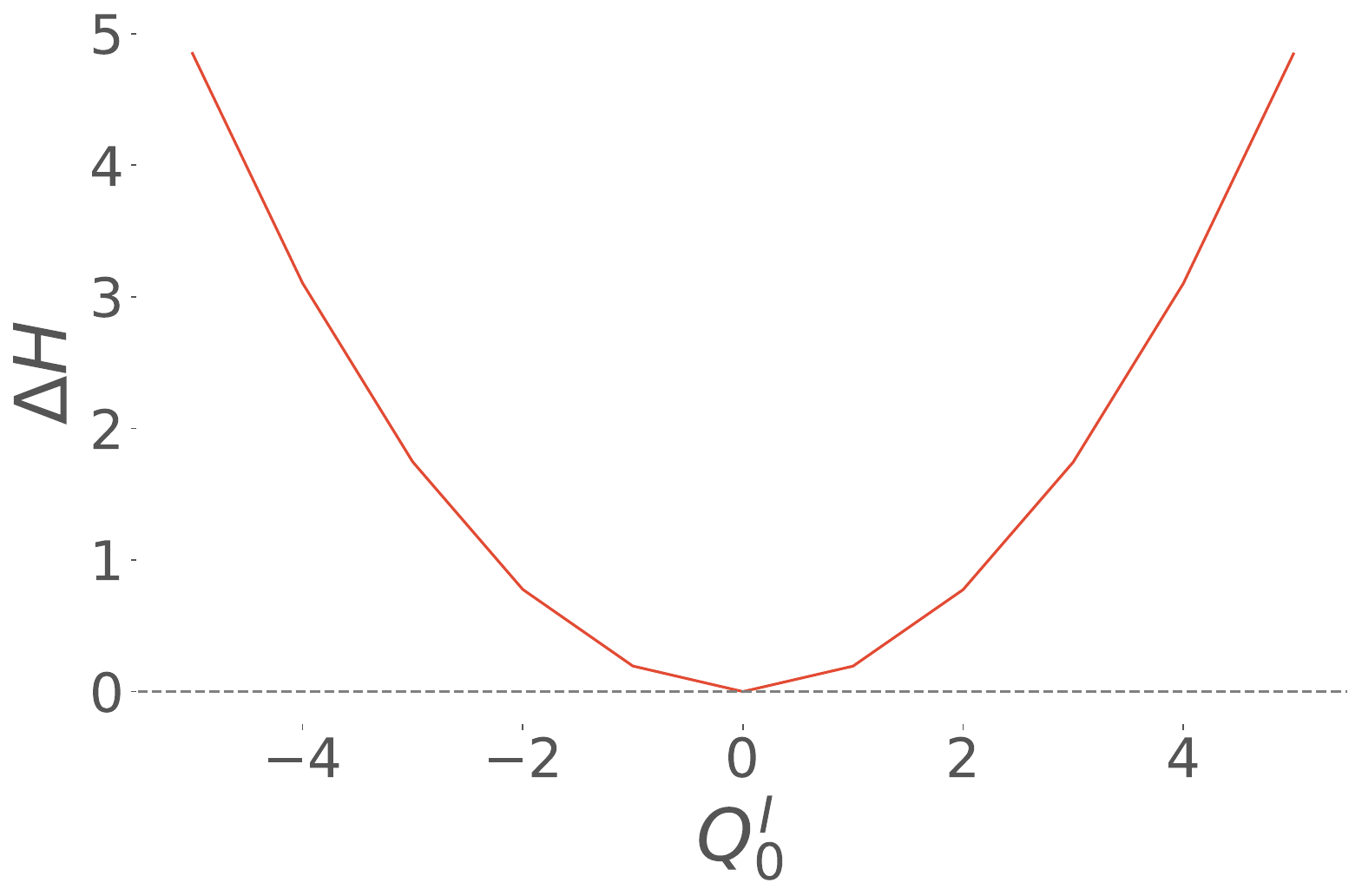}
    \includegraphics[width=0.47\textwidth]{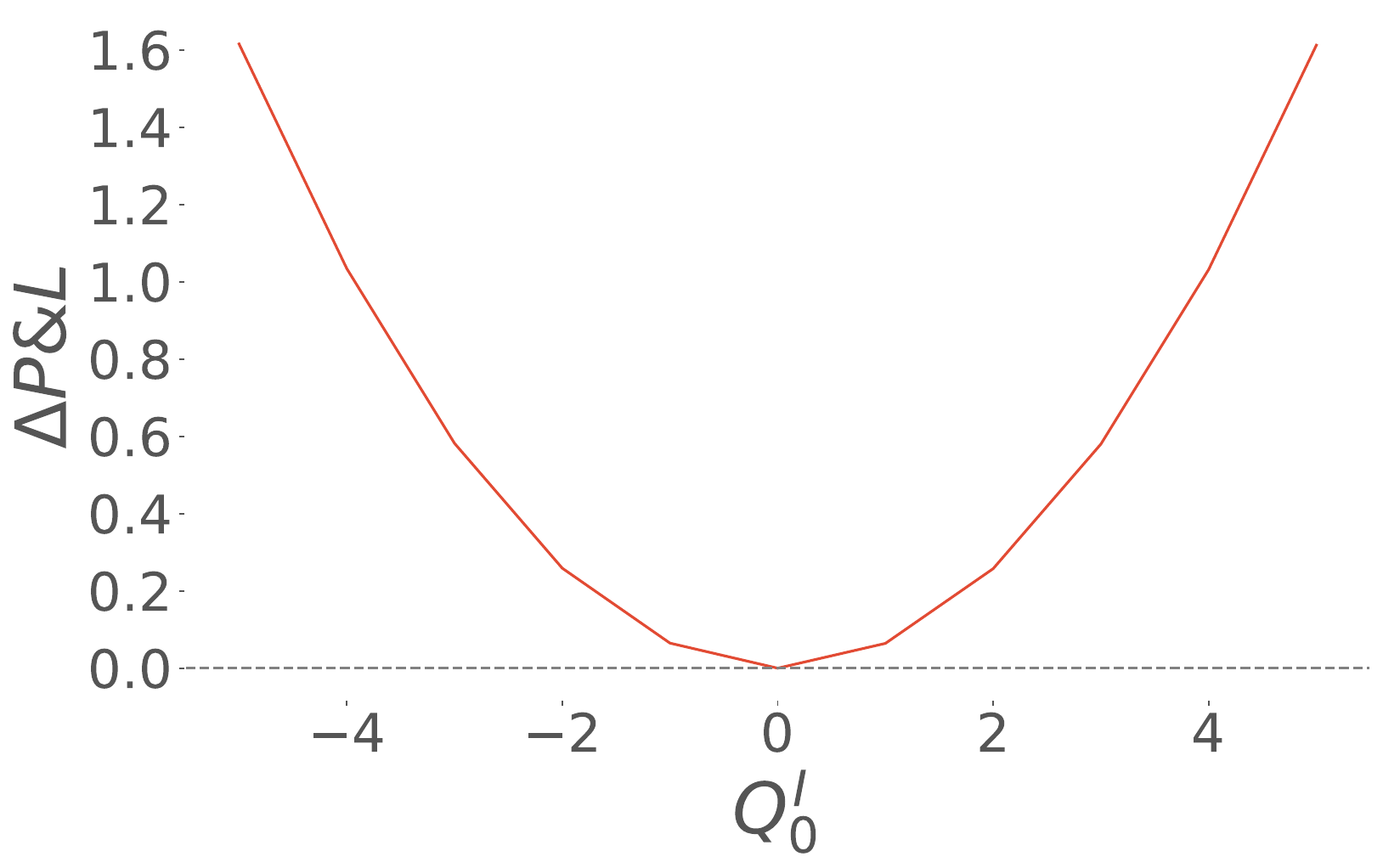}
    \caption{Change in performance criteria and expected P\&L for broker 1 under misspecification with different values of $Q^I_0$. Other parameter values are $S_0=100$, $\alpha_0 = 0$, $\theta = 3$, $\eta = 1$, $\sigma = 1$, $\kappa_1 =\kappa_2 = 0.001$, $c_1 = c_2 = 0.001$, $k_1 = k_2 = 0.0012$, $b_1 = b_2 = 0.001$, $\phi_1 = \phi_2 = \phi_I = \psi_I = 0.01$, $a_1 = a_2 = a_I = 1$, $u^{1}_0 = u^{2}_0 = 0$, $\theta_1 = \theta_2 = 5$, $\eta_1 = \eta_2 = 50$, $Q^{1}_0 = Q^{2}_0 = 0$, $\widehat{Q}^{I}_0=0$.}
\label{fig:RobustnessAnalysisValueFunction}
\end{figure}

\section{Liquidity Price Competition Between Brokers}\label{sec:liquidity_game}

Here we introduce the mechanism by which the collection of brokers compete for order flow from the informed trader to maximize their expected risk-adjusted profit. Before any trading takes place, the brokers specify the price of liquidity that they will offer to the informed trader by selecting the parameter $\kappa_i$\footnote{The liquidity price $\kappa_i$ applies only to the informed trader and changing $\kappa_i$ does not allow broker $i$ to extract additional wealth from the uninformed trader. The liquidity parameters $c_i$, along with the associated uninformed order flow $u^i$ in \eqref{eqn:uninformed_trading}, remain unaffected by the choice of $\kappa_i$.}. After all of the brokers simultaneously select this linear demand curve, trading between all agents proceeds as outlined in Section \ref{sec:model}. That is, the brokers, informed trader, and uninformed trader act according to \eqref{eqn:OptimalBrokerTradingSpeed}, \eqref{eqn:OptimalInformedTradingSpeed}, and \eqref{eqn:uninformed_trading}, respectively. Given that every agent acts according to these controls, define
\begin{align}
    V^i(\kappa_i, \kappa_{-i}) &= H^i(0, \alpha_0, S_0, Q^I_0, \bs{Q}^B_0, \bs{X}^B_0, \bs{u}_0; \nu^{-i*}, \bs{\kappa})\,,\label{eqn:BrokerInitialValueFunction}
\end{align}
where $\bs{\kappa} = (\kappa_1, \dots, \kappa_N)$ is a vector of liquidity prices between each broker and the informed trader, $\kappa_{-i}$ represents the collection of all liquidity prices other than the one set by broker $i$, and $H^i$ is the dynamic value function of broker $i$ given in \eqref{eqn:BrokerValueFunction}. All other parameters in the dynamics and performance criteria are fixed constants.

\begin{definition}[Liquidity Price Nash Equilibrium]\label{def:equilibrium}
	A liquidity price Nash equilibrium is a set of parameters $\bs{\kappa}^* = (\kappa_1^*, \dots, \kappa_N^*)$ such that for all $i\in\{1,\dots,N\}$
    \begin{align*}
        V^{i}(\kappa_i^*, \kappa_{-i}^*) \ge V^{i}(\kappa_i, \kappa_{-i}^*), \mbox{ for all } \kappa_i\neq \kappa^*_i\,.
    \end{align*}
\end{definition}
The interpretation of the equilibrium given in Definition \ref{def:equilibrium} is that we seek a set of liquidity cost parameters $\kappa_1^*,\dots,\kappa_N^*$ set by the brokers such that no broker can increase their performance by deviating from the corresponding $\kappa_i^*$. Once we have a liquidity price Nash equilibrium between the brokers, the initial value function of the informed trader can be acquired from equation \eqref{eqn:InformedValueFunctionCandidate}, and we denote it by 
\begin{align}
    V^I(\bs{\kappa}^*)&= H^I(0, \alpha_0, S_0, Q^I_0, X^I_0; \bs{\kappa}^*)\,,\label{eqn:InformedInitialValueFunction}
\end{align}
where $H^I$ is the informed trader's dynamic value function defined in \eqref{eqn:InformedValueFunction}.

\subsection{Numerical Experiment}

In this section, we numerically construct liquidity price Nash equilibria and perform sensitivity and other quantitative analysis. For reasons related to computational complexity and presentability of the results, we only consider the situation of $N=2$ brokers.

We propose the following algorithm for numerically computing an equilibrium\footnote{Although the existence and uniqueness of an equilibrium according to Definition \ref{def:equilibrium} has not been proven, our numerical experiments using the algorithm specified here consistently converge to the same limit (within numerical tolerance) regardless of the choice of initialization value $\kappa^{(0)}_2$.} given that other model parameters are fixed and denoted by $\Theta$.
\begin{enumerate}
    \item Divide the time interval $[0,T]$ into an equidistant time grid. Also specify a grid for $\kappa_1$ and $\kappa_2$. Initialize the liquidity price of the second broker as $\kappa^{(0)}_2$ which is one of the corresponding grid points.
    \item For an arbitrary $\kappa^{(n)}_2$, evaluate the function $V^{1}(\kappa_1, \kappa^{(n)}_2)$ for each value of $\kappa_1$ at its corresponding grid points. This requires numerically solving of the ODEs in Proposition \ref{prop:Broker} for which we use an explicit finite difference method. 
    \item We find the value of $\kappa_1$ which maximizes $V^{1}(\kappa_1, \kappa^{(n)}_2)$ and denote this value by $\kappa^{(n)}_1 = \kappa^*_1(\kappa^{(n)}_2)$.    
    \item In similar fashion, for a given $\kappa^{(n)}_1$ we compute $\kappa_2$ which maximizes $V^2(\kappa^{(n)}_1,\kappa_2)$ and denote this value by $\kappa^*_2(\kappa_1^{(n)})$.
    \item We introduce a learning rate parameter $\gamma \in [0,1)$ to update $\kappa^{(n+1)}_2$ for the next iteration. To reduce oscillations in searching for the equilibrium, we choose $\gamma$ to be small, and we update $\kappa_2^{(n+1)}$ according to
    \begin{align*}
        \kappa_2^{(n+1)} = (1-\gamma)\, \kappa_2^{(n)} + \gamma\,\kappa^{*}_2(\kappa^{*}_1(\kappa^{(n)}_2)).
    \end{align*}
    \item Steps 2 through 5 are repeated until convergence of the sequences $\{\kappa^{(n)}_2\}_{n\geq 0}$ and $\{\kappa_1^*(\kappa_2^{(n)})\}_{n\geq 0}$ in max absolute error within a specified tolerance. We drop the counting index in the final $\kappa_2$ and compute the corresponding $\kappa_1^*(\kappa_2)$.    
\end{enumerate}
Once convergence is attained we end up with a pair $\kappa^*_1(\Theta)$ and $\kappa^*_2(\Theta)$ which satisfy the conditions (to within numerical tolerance) of equilibrium given in Definition \ref{def:equilibrium}.

Below we showcase two symmetric broker's initial value functions to check the existence of the equilibrium numerically. Model parameters for the price dynamics and price impact are
\begin{align*}
    \alpha_0 &= 0\,, & S_0 &= 100\,,       & \theta &= 0.1\,,  & \eta &= 0.1\,,\\
    \sigma &= 0.1\,, & c_1 &= c_2=0.1\,, & k_1 &= k_2=0.2\,, & b_1 &= b_2=0.14\,, 
\end{align*}
and the penalty and uninformed trading parameters are
\begin{align*}
    a_I &= 1\,,            & a_1 &= a_2=40\,,        & \phi_I &= 0.01\,,             & \psi_I &= 0.1\,,\\
    \phi_1 &= \phi_2=20\,, & u^{1}_0 &= u^{2}_0=0\,, & \theta_1 &= \theta_2=0.001\,, & \eta_1 &= \eta_2=0.01\,. 
\end{align*}
Additionally, all initial inventories are set to $0$ when computing a liquidity price Nash equilibrium. This ensures that the source of value from each agent in the market is not due to the value of their inventory or potential value extracted from a different agent's inventory. All of the value stems from the market structure and optimal actions of each agent.

Taking first broker 1 as an example, the left panel of Figure \ref{fig:GBTG-3DMap&Single} shows the surface of broker 1's initial value function $V^{1}(\kappa_1, \kappa_2)$ over the two axes of $\kappa_1$ and $\kappa_2$. We can observe that for fixed $\kappa_2$, initial value function $V^{1}$ is concave with respect to $\kappa_1$. This concavity represents the trade-off for broker 1 between large but infrequent revenue of high liquidity prices, versus small but frequent revenue of low liquidity prices. In addition, for fixed $\kappa_1$, the initial value function $V^{1}$ is increasing in $\kappa_2$, which shows broker 1 can benefit from broker 2's increasing cost for the informed trader. This happens because as broker 2 increases their liquidity price, the informed trader will redirect a larger proportion of their order flow to broker 1. Another interesting observation is that $\kappa_1^*(\kappa_2)$ is slightly increasing in $\kappa_2$, represented by the slope of the dotted line projected on the plane. This means that if broker 2 charges a higher liquidity cost for the informed trader, broker 1 can follow with a mild, economically sensible increase in cost. In the right panel we display both surfaces of $V^{1}$ and $V^{2}$ and the corresponding lines of $\kappa^*_1(\kappa_2)$ and $\kappa_2^*(\kappa_1)$. The intersection of these two curves, marked by a point, represents a liquidity price Nash equilibrium. At this point both brokers achieve their best performance given the behaviour of the other broker.

\begin{figure}
    \centering
    \includegraphics[width=0.49\textwidth]{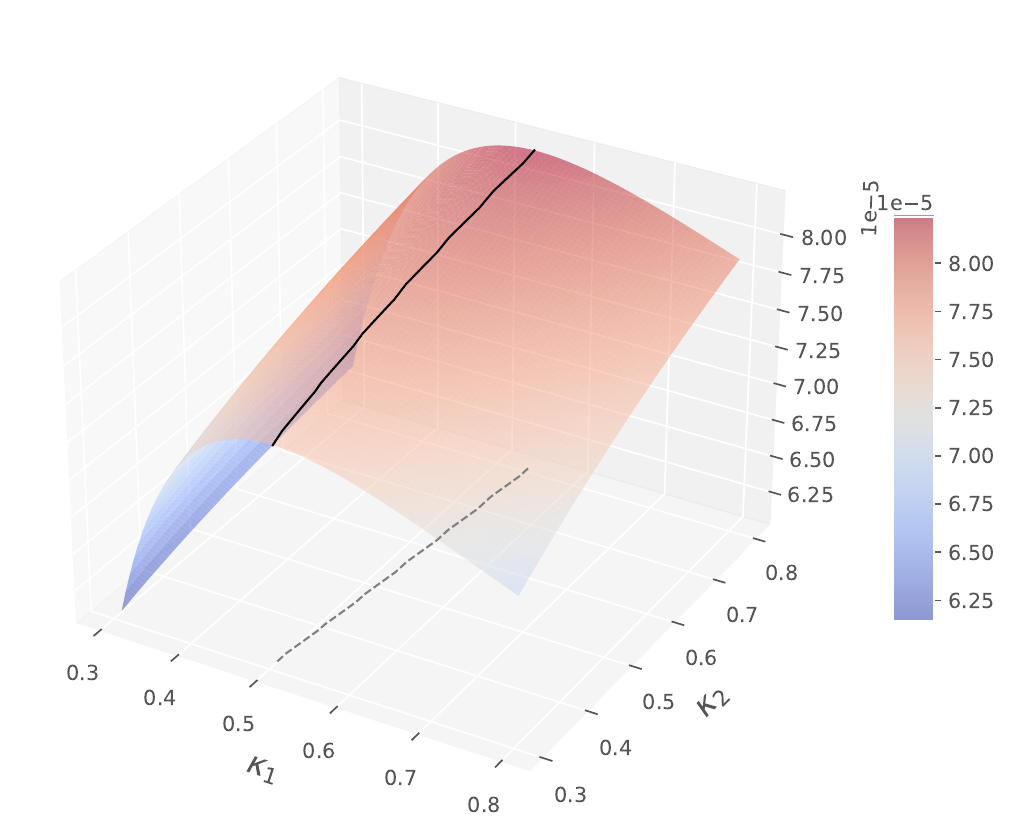}
    \includegraphics[width=0.49\textwidth]{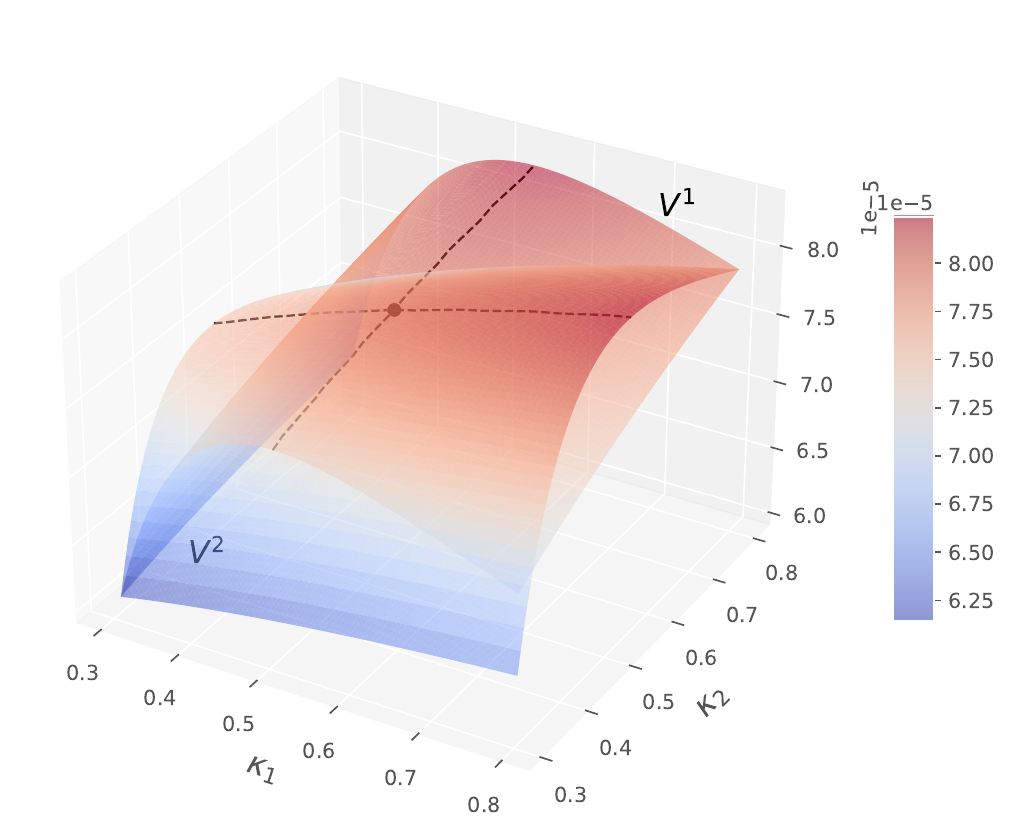}
    \caption{Left: Initial value function $V^{1}$ of broker 1. The solid line on the surface is the optimal $\kappa^{*}_1(\kappa_2)$ which maximizes $V^{1}$ given $\kappa_2$, i.e. $\kappa^*_1(\kappa_2)=\argmax_{\kappa_1} V^{1}(\kappa_1, \kappa_2)$, and the dotted line is its projection on the $\kappa_1\kappa_2$-plane. Right: Broker's initial value function $V^{1}$ and $V^{2}$. The dotted line on $V^{1}$ is $\kappa^*_1(\kappa_2)$, and the dotted line on $V^{2}$ is $\kappa^*_2(\kappa_1)$. The dot on the surface is the equilibrium.}
\label{fig:GBTG-3DMap&Single}
\end{figure}

An important consideration when investigating equilibrium in financial markets is the global welfare of such a state. For this reason we are interested in whether the liquidity price Nash equilibrium demonstrated in the right panel of Figure \ref{fig:GBTG-3DMap&Single} is Pareto efficient for the brokers. To achieve this we compute the difference between each broker's value function $V^i(\kappa_1,\kappa_2)$ for arbitrary $\kappa_1$ and $\kappa_2$ and the equilibrium value $V^i(\kappa_1^*,\kappa_2^*)$. The positive value of this difference is plotted for each broker in the left panel of Figure \ref{fig:GBTG-3DMap-Difference&Better-Choice}. For visual clarity, the region in which both of these differences are positive is shown in the right panel of Figure \ref{fig:GBTG-3DMap-Difference&Better-Choice}. Since both brokers have higher value function than their equilibrium value at all points in this region, each of these pairs $(\kappa_1,\kappa_2)$ represents a Pareto improvement. This is an indication that policies determined by an external party, such as through regulation or other restrictions, could result in greater value to the group compared to when each individual acts in their self interest.

\begin{figure}
    \centering
    \includegraphics[width=0.49\textwidth]{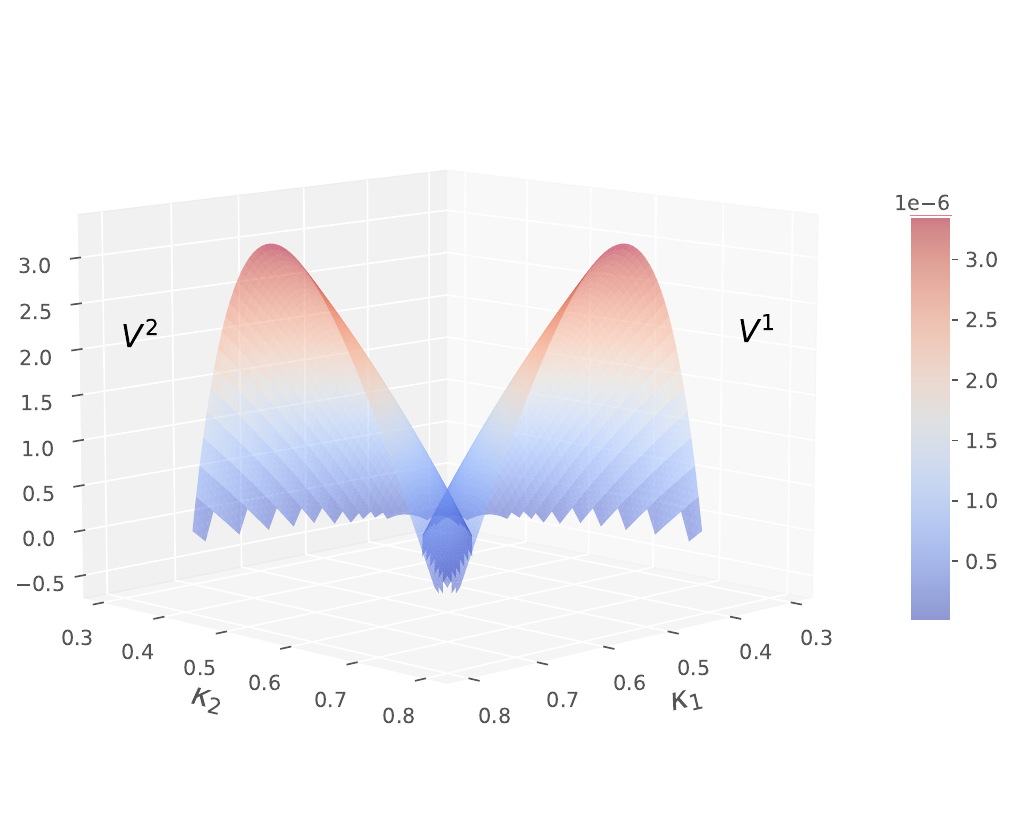}
    \includegraphics[width=0.49\textwidth]{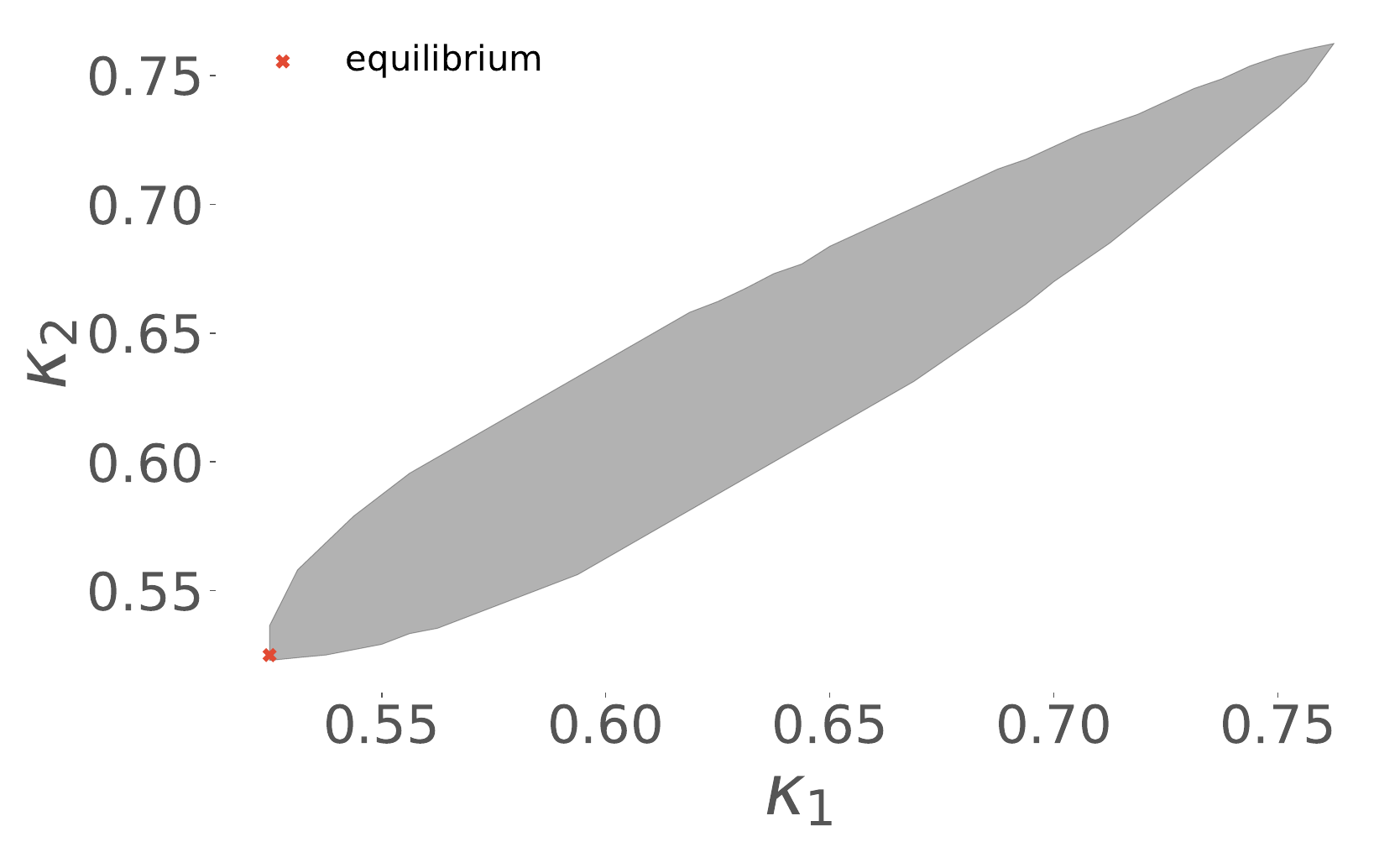}
    \caption{Left: Positive difference between $V^i(\kappa_1,\kappa_2)$ and $V^i(\kappa^*_1,\kappa^*_2)$ of two brokers. Right: Pairs of $(\kappa_1,\kappa_2)$ where positive difference surfaces overlap.}
\label{fig:GBTG-3DMap-Difference&Better-Choice}
\end{figure}

\subsubsection{Effect of Brokers' Running Inventory Penalty on the Equilibrium}\label{sec:RunningInventoryNumericalExperiment}

In this section, we investigate how a broker's running inventory penalty parameter $\phi_i$ affects the equilibrium as per Definition \ref{def:equilibrium} in the case of two asymmetric brokers. More specifically, we fix the running inventory penalty parameter $\phi_2$ of broker 2 (as well as all other parameters) and vary $\phi_1$ of broker 1.

\begin{figure}[!htp]
    \centering
    \includegraphics[width=0.49\textwidth]{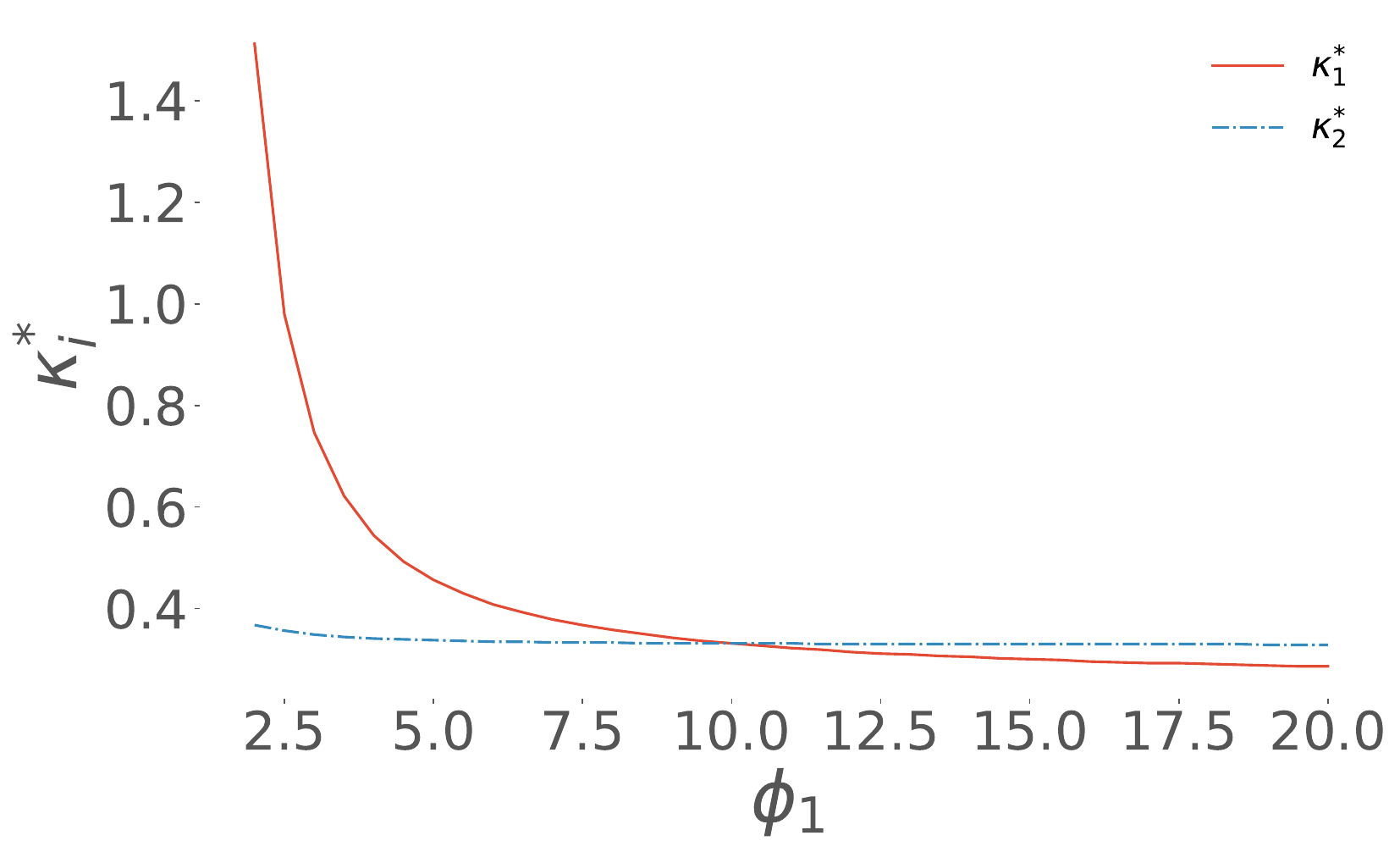}
    \includegraphics[width=0.49\textwidth]{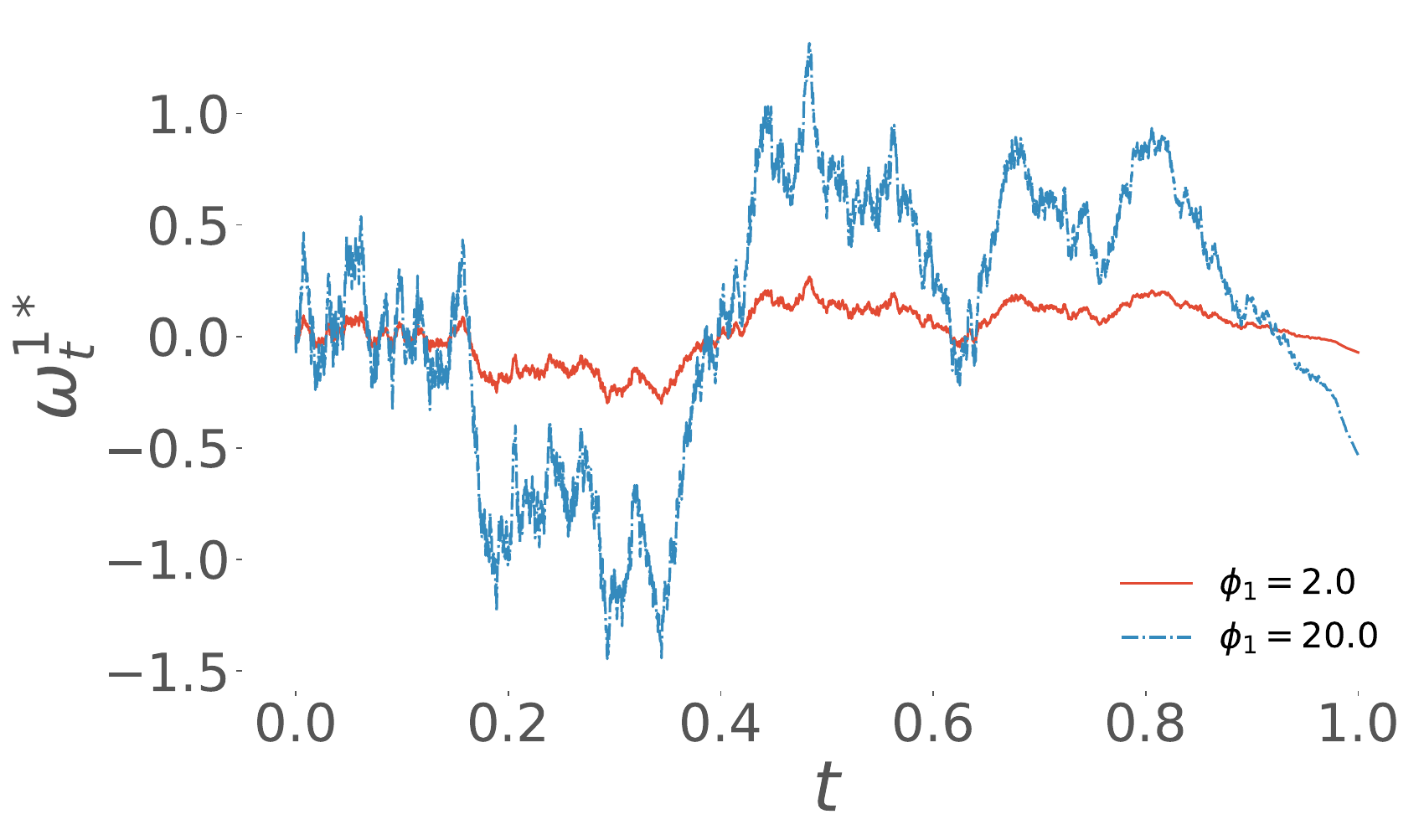}
    \caption{Left: Equilibrium with different values of broker 1's running inventory penalty parameter $\phi_1$. Right: Sample path for the informed trader's optimal trading speed $\omega^{1*}_t$ with broker 1 in the equilibrium with different values of $\phi_1$. The other parameters are $\phi_2=10,\alpha_0=0, S_0=5,\theta=1,\eta=5,\sigma=1, c_1=c_2=0.1,k_1=k_2=0.12,b_1=b_2=0.1,a_I=1,\phi_I=0.01,\psi_I=0.01,a_1=a_2=20,u^{1}_0=u^{2}_0=0,\theta_1=\theta_2=1,\eta_1=\eta_2=1,Q^I_0=Q^{1}_0=Q^{2}_0=0,X^I_0=X^{1}_0=X^{2}_0=0$.}
    \label{fig:GBTG-equilibrium-phi_B}
\end{figure}

In Figure \ref{fig:GBTG-equilibrium-phi_B}, we can see in the left panel that the optimal liquidity price $\kappa^*_1$ of broker 1 for the informed trader decreases when broker 1 receives higher running penalty parameter $\phi_1$. Therefore, the informed trader tends to trade more with broker 1 when the cost is lower, as shown in the right panel. This seems counter-intuitive as we expect broker 1 will quote higher liquidity cost to the informed trader to discourage transactions with higher penalty. However, given that brokers can monetize from the alpha of the informed trader, such speculative investment may change when penalty parameter changes. Hence, we need to consider their speculative investment at the same time. To investigate further, we define the net trading speed $\tilde{\nu}^i_t$ of broker $i$ as
\begin{align*}
    \tilde{\nu}^i_t \coloneqq \nu^i_t-\omega^i_t-u^i_t\,.
\end{align*}
With the informed trader's optimal strategy \eqref{eqn:OptimalInformedTradingSpeed} and brokers' optimal strategies \eqref{eqn:OptimalBrokerTradingSpeed}, we rewrite $\tilde{\nu}^{1*}_t$ as
\begin{align}
    \tilde\nu^{1*}_t &= r^1_1(t)\,\alpha_t + r^1_2(t)\,Q^{1*}_t + r^1_3(t)\,Q^{2*}_t + r^1_4(t)\,Q^{I*}_t + r^1_5(t)\,u^1_t + r^1_6(t)\,u^2_t\,,\label{eqn:NetTradingSpeed}
\end{align}
where
\begin{align*}
    r^1_1(t) &= \frac{x^1_1(t)}{2k_i} - \omega^1_0(t) &  r^1_2(t) &= \frac{b_1(t) + 2n_{1,1}^1(t)}{2k_1}\\
    r^1_3(t) &= \frac{n^1_{2,1}(t) + n^1_{1,2}(t)}{2k_1} & r^1_4(t) &= \frac{d^1_1(t)}{2k_1} - \omega^1_1(t)\\
    r^1_5(t) &= \frac{y^1_{1,1}(t)}{2k_1} - 1 & r^1_6(t) &= \frac{y^1_{1,2}(t)}{2k_1}
\end{align*}
The first term $r^1_1(t)\,\alpha_t$ is the speculative behavior of broker 1. The forth term $r^1_4(t)\,Q^{I*}_t$ and the fifth term $r^1_5(t)\,u^1_t$ account for how much of the informed and uninformed flow is internalized by broker 1. Figure \ref{fig:GBTG-phi_B-v_B_composition} plots each term of broker 1's net trading speed $\tilde{\nu}^{1*}_t$ in equation \eqref{eqn:NetTradingSpeed} when $\phi_1=2$ and $\phi_1=20$. From the lower left panel, we can see that even though the informed trader has higher trading volume with broker 1 when the penalty parameter is higher, the component related to the informed trader's inventory $r^1_4(t)\,Q^{I*}_t$ does not change much except for times close to $T$. From this, we can deduce that broker 1 is accepting more informed trader's transactions and liquidating faster in the case of higher running inventory penalty. Furthermore, the upper left panel shows that the speculative term associated with alpha $r^1_1(t)\,\alpha_t$ shrinks when broker 1 faces higher running inventory penalty. The other panels indicate that all the other terms are smaller in absolute size for higher penalty. Overall, we see from Figure \ref{fig:GBTG-phi_B-v_B_1_net&Q_B_1} that broker 1's net trading speed  $\tilde{\nu}^{1*}_t$ decreases in size when $\phi_1$ is higher, and therefore broker 1 indeed has a lower inventory level. Additionally, when broker 1 has a higher penalty and sets a lower transaction price for the informed trader, she actually liquidates orders from the informed trader faster and reduces her volume of speculative trading.

\begin{figure}
    \centering
    \includegraphics[width=0.32\textwidth]{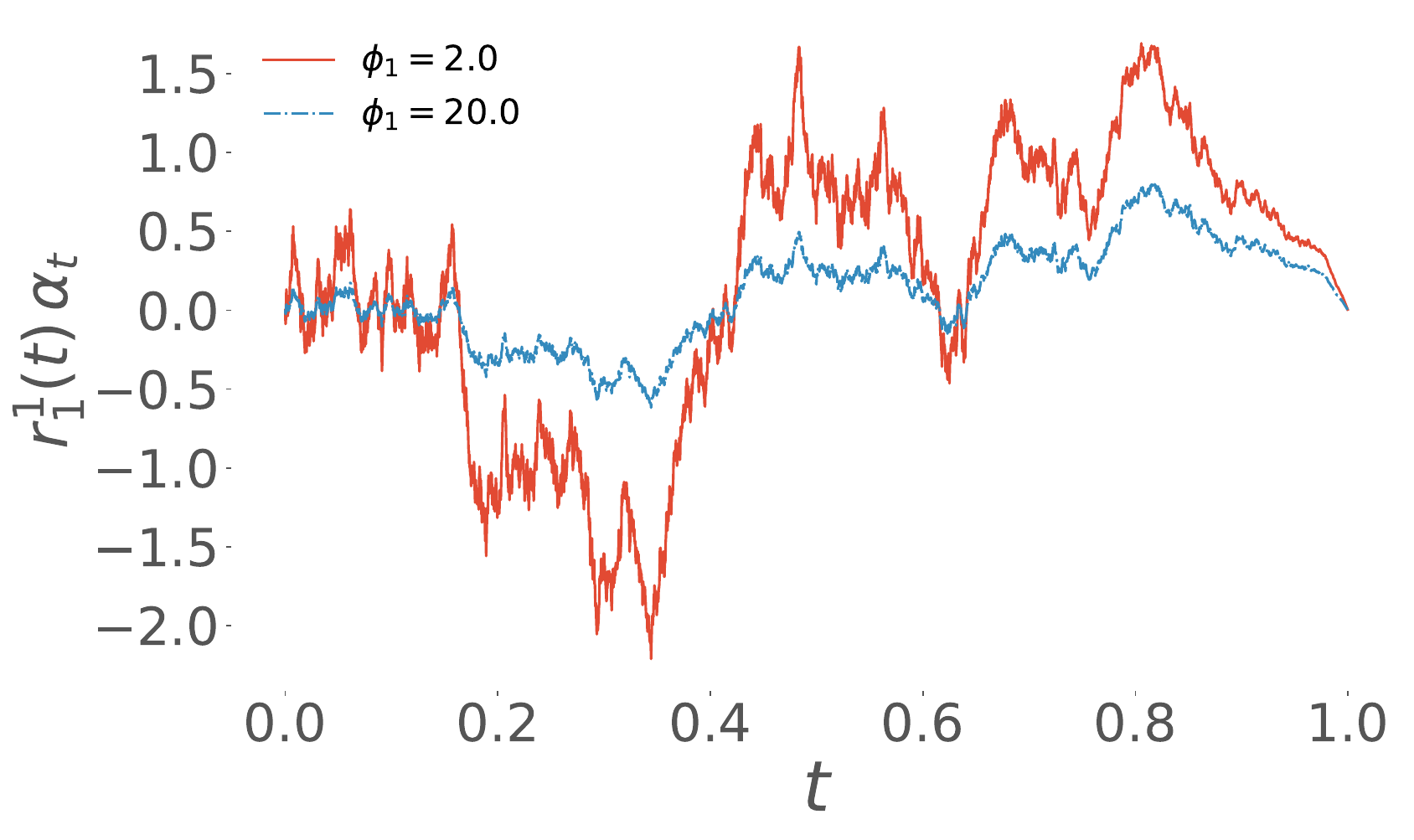}
    \includegraphics[width=0.32\textwidth]{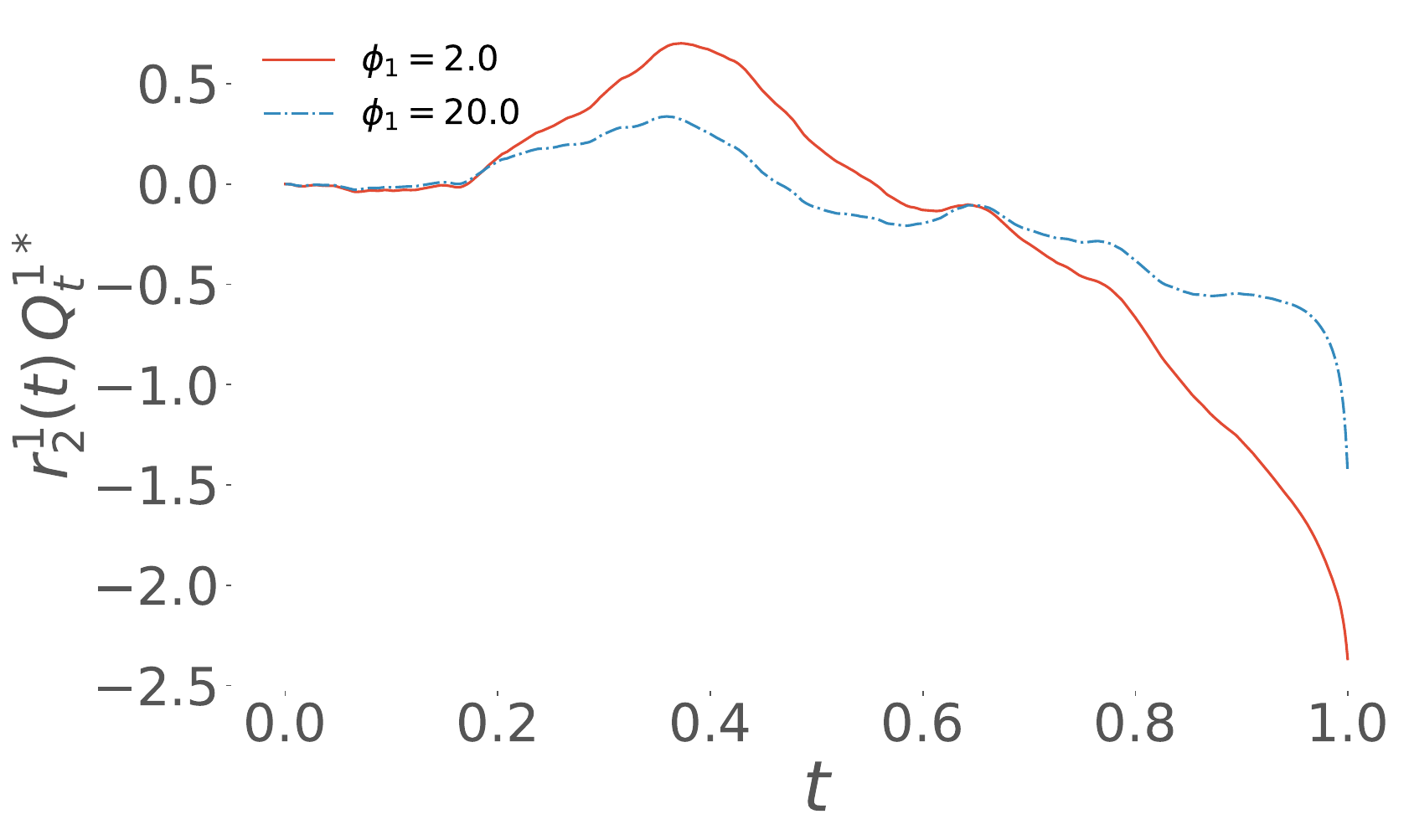}
    \includegraphics[width=0.32\textwidth]{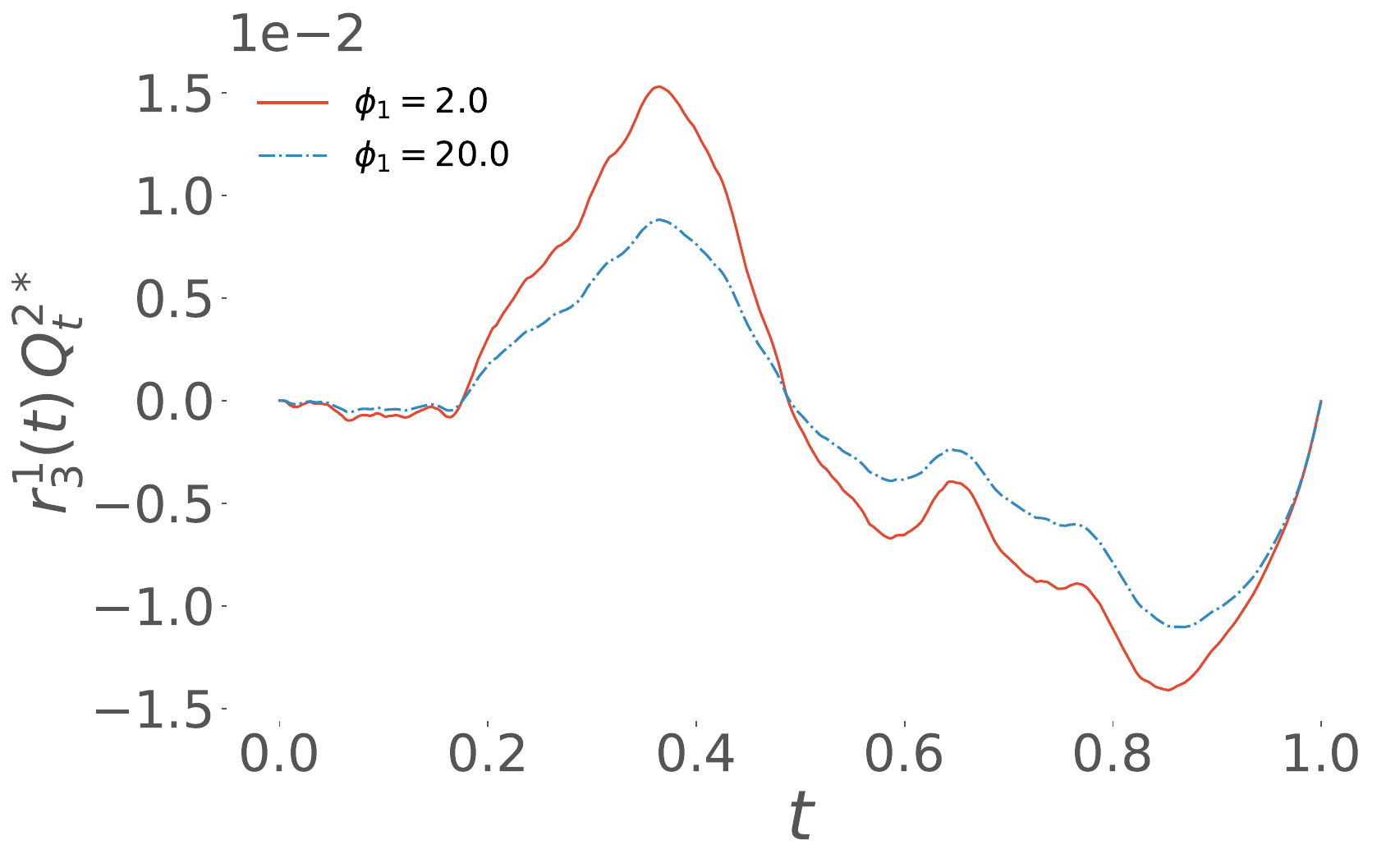}
    \includegraphics[width=0.32\textwidth]{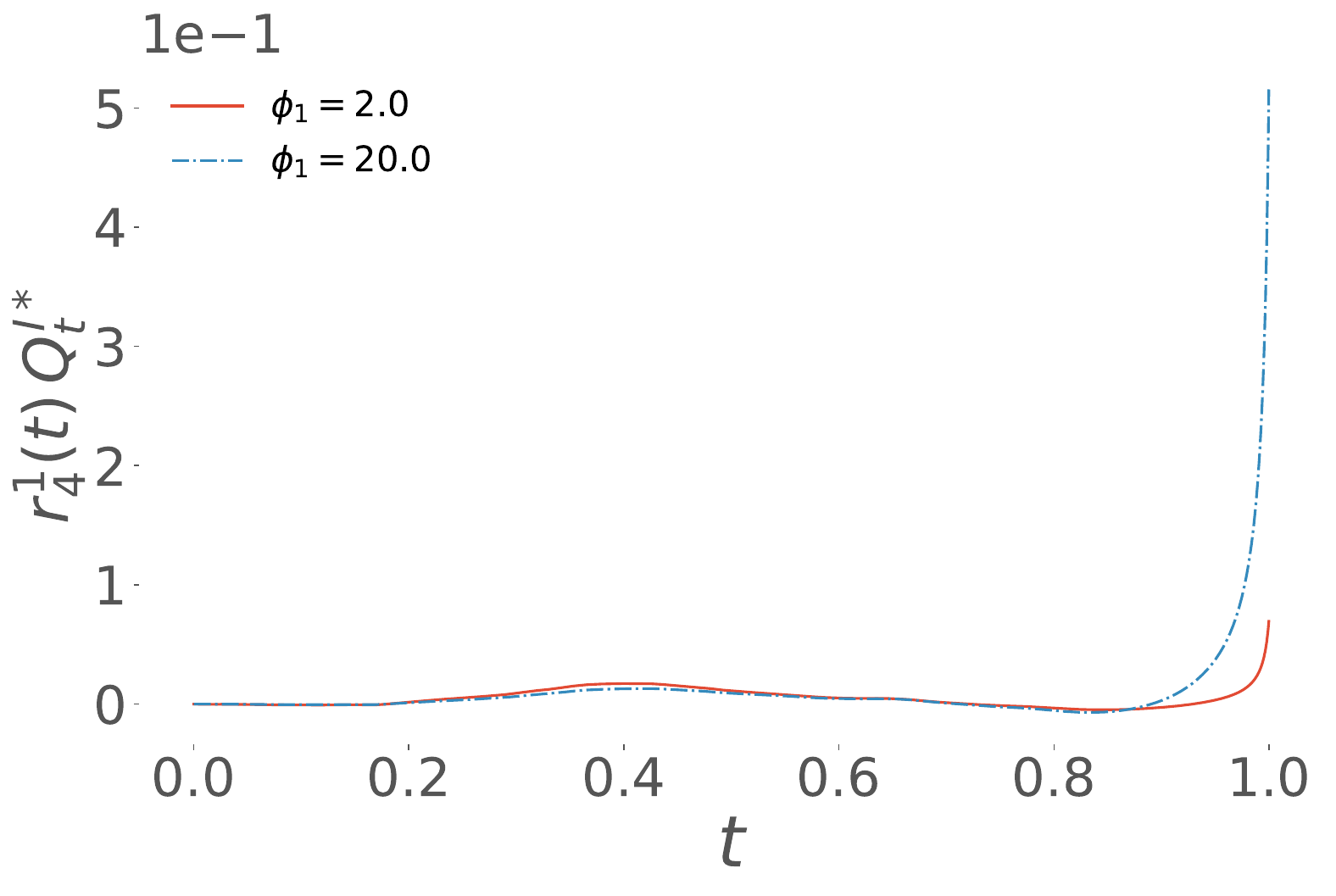}
    \includegraphics[width=0.32\textwidth]{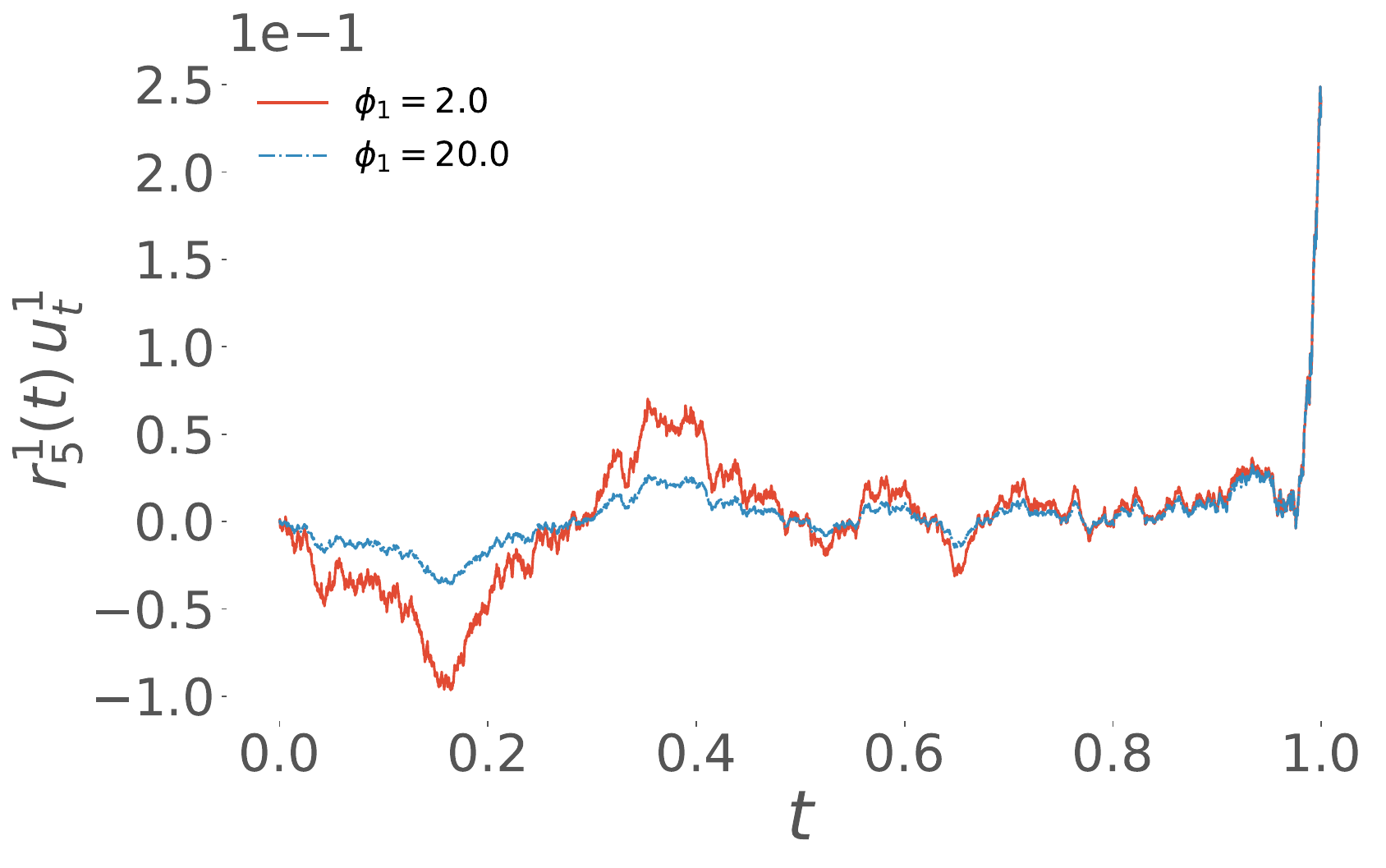}
    \includegraphics[width=0.32\textwidth]{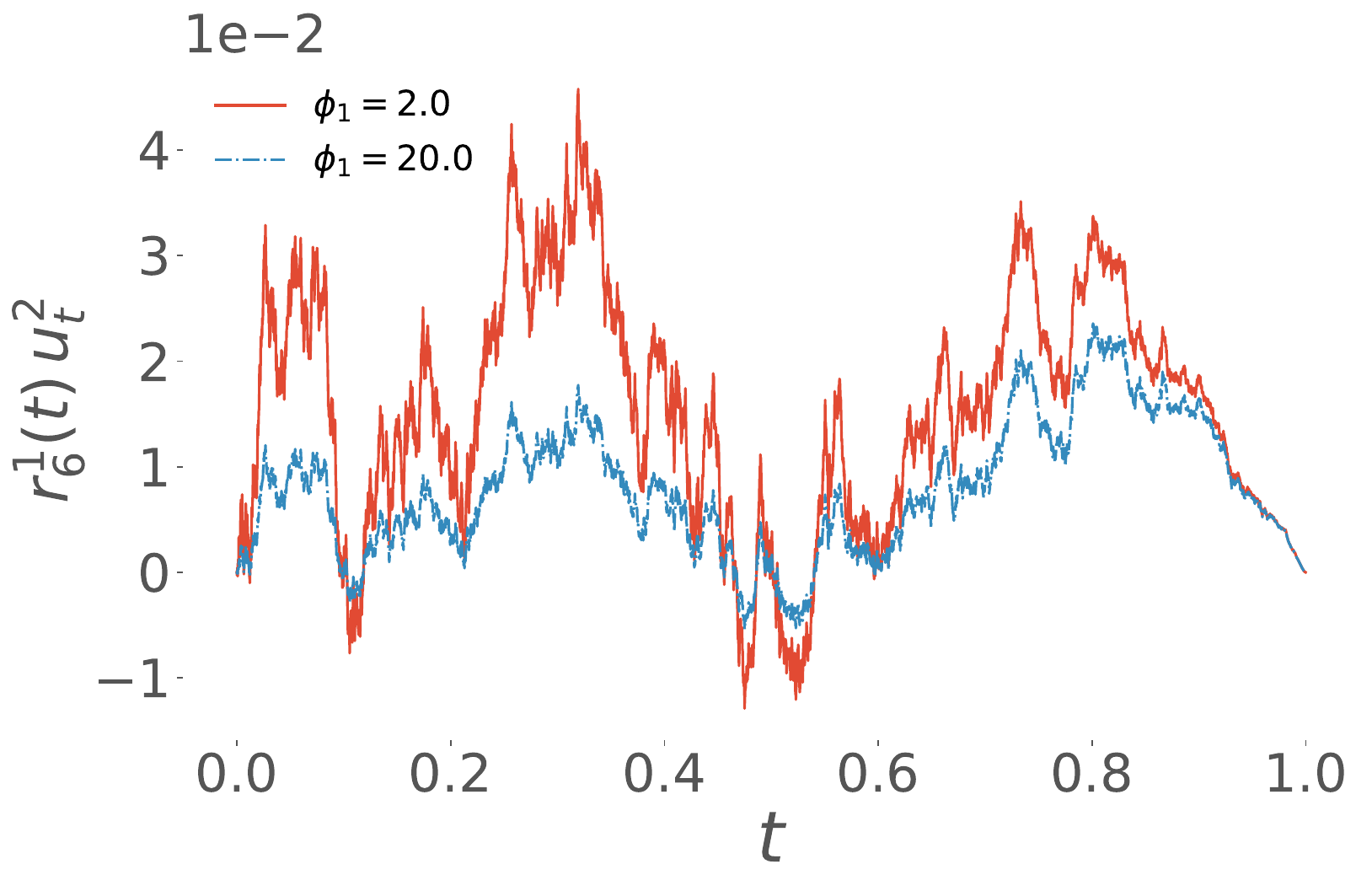}
    \caption{Sample paths for each component of broker 1's trading rate $\nu^{1*}_t$ for different values of running inventory penalty parameter $\phi_1$. The other parameters are the same as in Figure \ref{fig:GBTG-equilibrium-phi_B}. }
    \label{fig:GBTG-phi_B-v_B_composition}
\end{figure}

\begin{figure}
    \centering
    \includegraphics[width=0.49\textwidth]{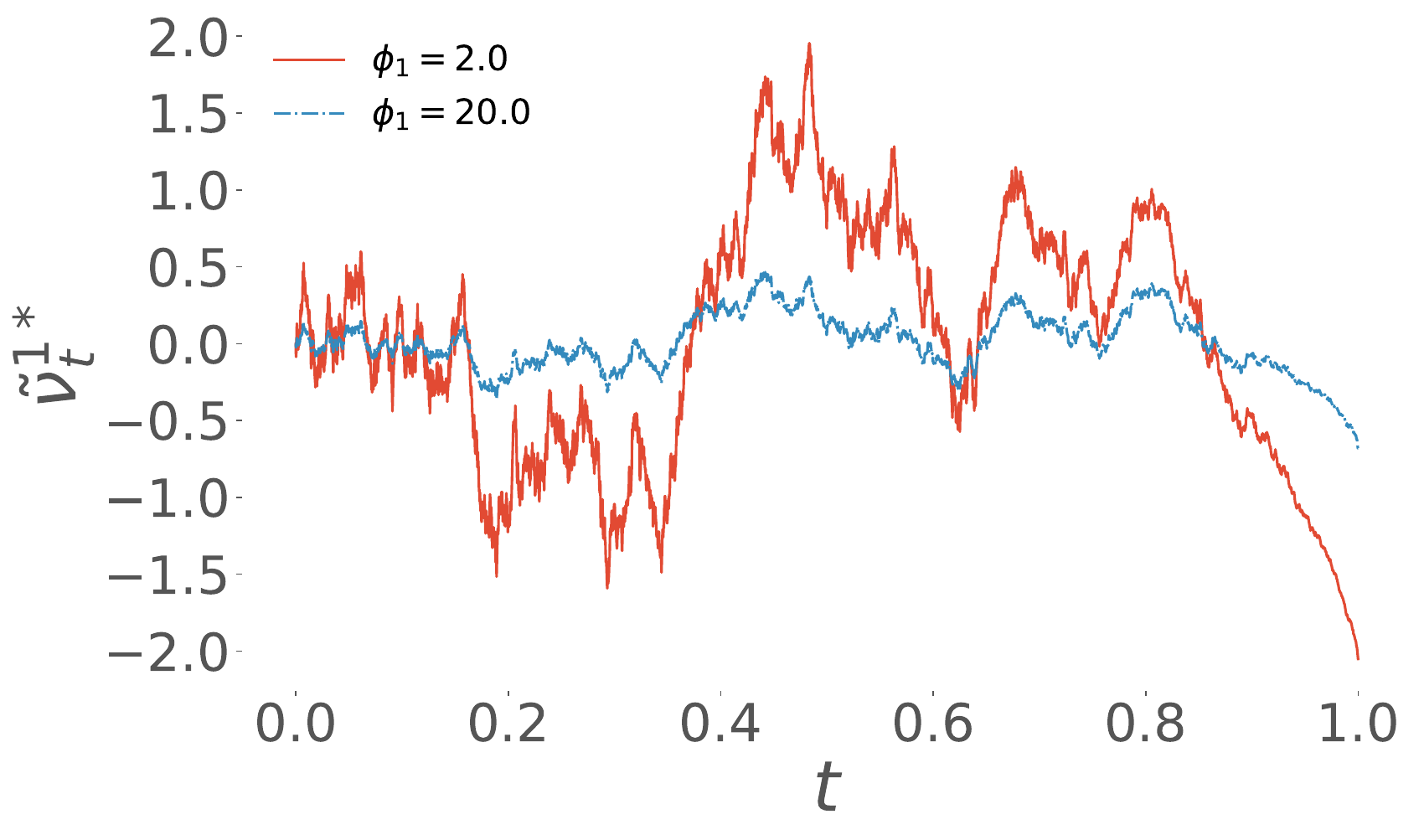}
    \includegraphics[width=0.49\textwidth]{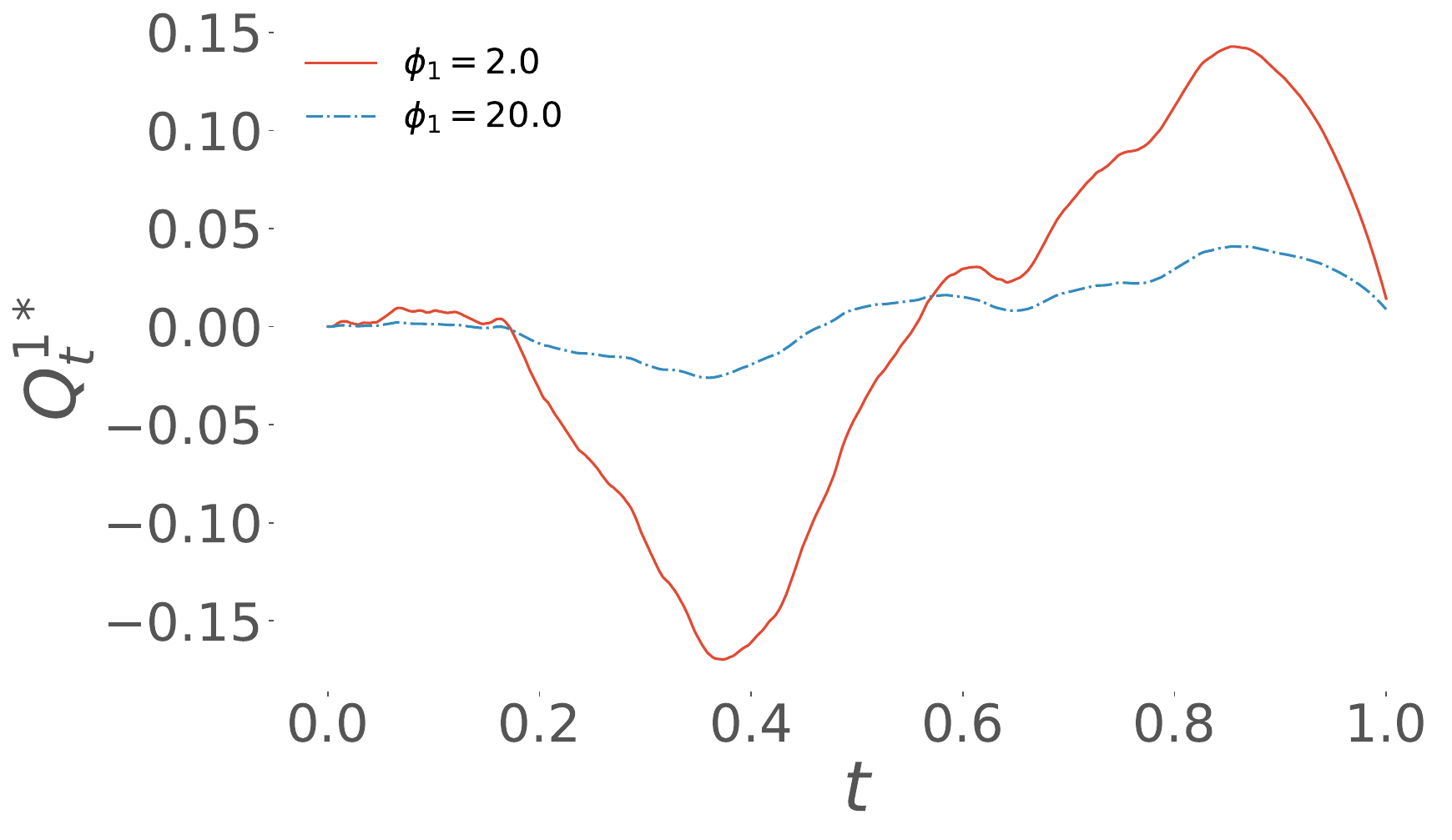}
    \caption{Sample paths for broker 1's net trading rate $\tilde{\nu}^{1*}_t$ and inventory $Q^{1*}_t$ for different values of running inventory penalty parameter $\phi_1$. The other parameters are the same as in Figure \ref{fig:GBTG-equilibrium-phi_B}.}
    \label{fig:GBTG-phi_B-v_B_1_net&Q_B_1}
\end{figure}

To confirm consistency of this conclusion, we define several quantities to aggregate and study across many realizations:
\begin{itemize}
    \item $Z^I \coloneqq \int^T_0 |\omega^1_t|\,dt$, the total trading volume of the informed trader with broker 1.
    \item $\bar{Q} \coloneqq \frac{1}{T}\int^T_0 |Q^1_t|\,dt$, broker 1's average inventory level during the entire trading horizon.
    \item $Y_{tot}^\alpha \coloneqq \int^T_0|r^1_1(t)\,\alpha_t|\,dt$, the total speculative trading volume of broker 1. This is the integral of the first term in \eqref{eqn:NetTradingSpeed}.
    \item $Y_{tot}^I \coloneqq \int^T_0|r^1_4(t)\,Q^{I*}_t|\,dt$, the total trading volume related to the informed trader's inventory of broker 1. This is the integral of the forth term in \eqref{eqn:NetTradingSpeed}.
\end{itemize}

We run 4000 simulations and investigate how these quantities change in liquidity price Nash equilibrium when $\phi_1$ increases from 2 to 20 with results summarized in Tables \ref{table:phi_B_1_t_test} and \ref{table:phi_B_1_t_test_}. For each quantity we conduct a two-sample t-test. The benefit of this is that we can check if the behaviour of the broker is really different before and after the increase in running inventory penalty. We see that the average inventory level of broker 1 decreases from 0.17 to 0.02, which indicates broker 1's trading behavior is indeed suppressed due to the higher running inventory penalty. However, we also see that the informed trader's trading volume with broker 1 increases significantly from 0.16 to 0.77, implying that broker 1 reduces her trading flow associated with alpha. This is confirmed by the decrease of speculative trading volume of broker 1 from 1.08 to 0.35. It is worth noticing that the total trading volume related to the informed trader's inventory also decreases, but this reduction (0.07) is significantly smaller compared to the decrease (0.73) in speculative volume. All of the t-stats are greater than 2 in absolute value showing high confidence that the broker's behaviour changes with the increase in penalty parameter. These observations rationalize the decreasing tendency of $\kappa^{1*}$ with respect to $\phi_1$ in Figure \ref{fig:GBTG-equilibrium-phi_B}: broker 1 reduces her trading activities to avoid a higher running inventory penalty, and she prefers to decrease her speculative trading rather than transactions with the informed trader.

\begin{table}
\begin{center}
    \begin{tabular}{ l|l|l|l|l|l }
     Char. & t-stat & p-value & average-1 & average-2 \\
     \hline
     $Z^I$ & -100.40 & 0.00 & 0.16 & 0.78 \\
     $\bar{Q}$ & 84.83 & 0.00 & 0.17 & 0.02 \\
     $Y_{tot}^\alpha$ & 81.17 & 0.00 & 1.08 & 0.35 \\
     $Y_{tot}^I$ & 4.82 & $1.44\times 10^{-6}$ & 1.42 & 1.35
     \end{tabular}
    \caption{Two sample t-test for characteristics in Scenario 1 ($\phi_1=2$) and Scenario 2 ($\phi_1=20$). The other parameters are the same as in Figure \ref{fig:GBTG-equilibrium-phi_B}.}\label{table:phi_B_1_t_test}
\end{center}
\end{table}

\begin{table}
\begin{center}
    \begin{tabular}{ l|l|l|l|l|l }
     Char. & average-1 & std-1 & average-2 & std-2 \\
     \hline
     $Z^I$ & 0.16 & 0.08 & 0.78 & 0.38 \\
     $\bar{Q}$ & 0.17 & 0.11 & 0.02 & 0.01 \\
     $Y_{tot}^\alpha$ & 1.08 & 0.54 & 0.35 & 0.18 \\
     $Y_{tot}^I$ & $1.42\times10^{-5}$ & $7.14\times 10^{-6}$ & $1.35\times10^{-5}$ & $6.75\times 10^{-6}$
     \end{tabular}
    \caption{Average and standard deviation for characteristics in Scenario 1 ($\phi_1=2$) and Scenario 2 ($\phi_1=20$). The other parameters are the same as in Figure \ref{fig:GBTG-equilibrium-phi_B}.}\label{table:phi_B_1_t_test_}
\end{center}
\end{table}

Figure \ref{fig:GBTG-phi_B-V_B&V_I} shows the dependency of each agent's value function with respect to the running inventory penalty of broker 1. On the left panel, we can observe that both brokers lose some initial value on average. This is very intuitive as broker 1 loses value with more penalty and broker 2 loses value from the malicious competition of broker 1, since she is forced to reduce her fee as shown in the left panel of Figure \ref{fig:GBTG-equilibrium-phi_B}. Meanwhile, the informed trader benefits from the increasing running inventory penalty parameter $\phi_1$, as she can conduct more transactions in this scenario. Also, from the perspective of competition of monetizing alpha, the informed trader benefits from less speculative trading of broker 1.

\begin{figure}
    \centering
    \includegraphics[width=0.49\textwidth]{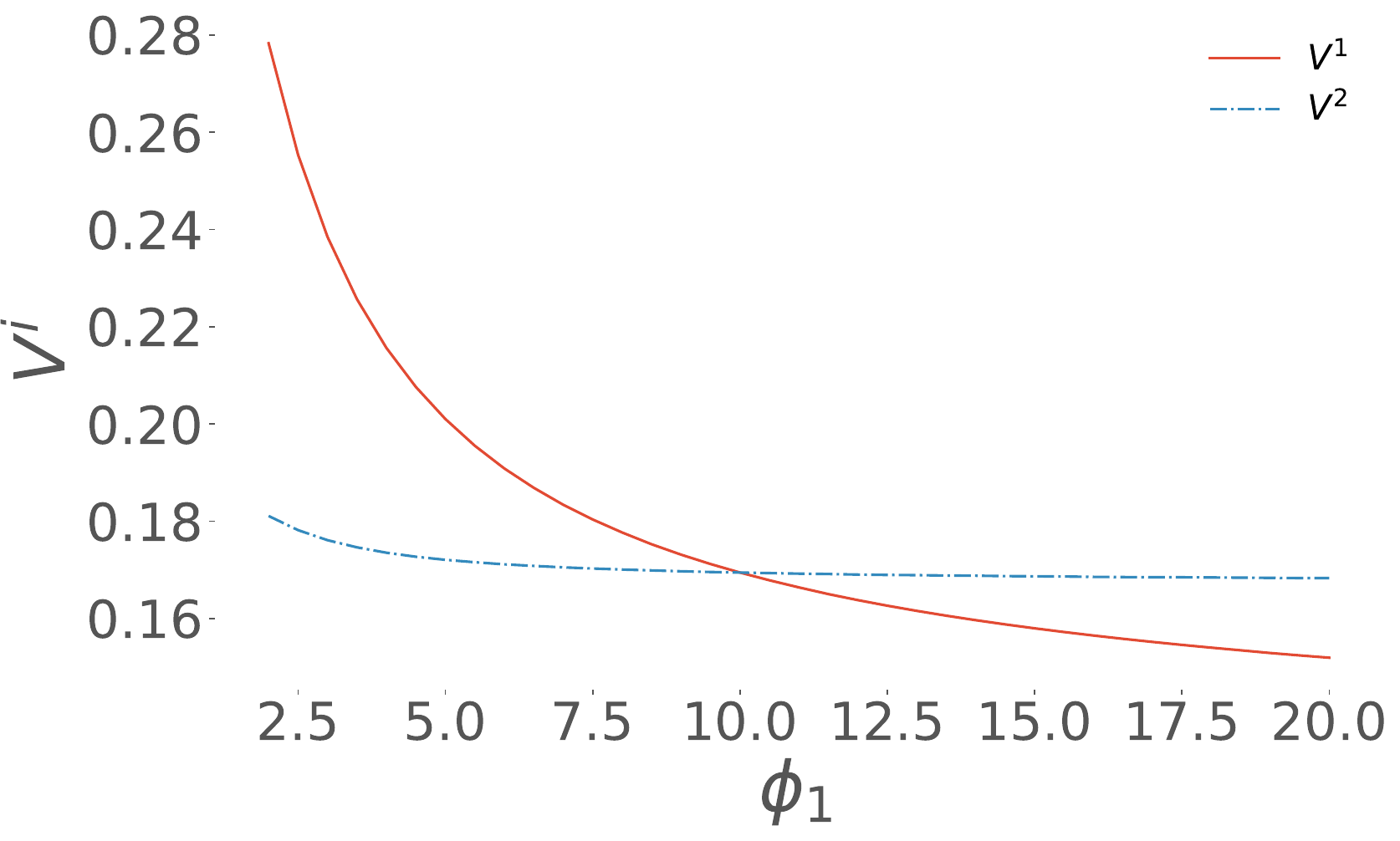}
    \includegraphics[width=0.49\textwidth]{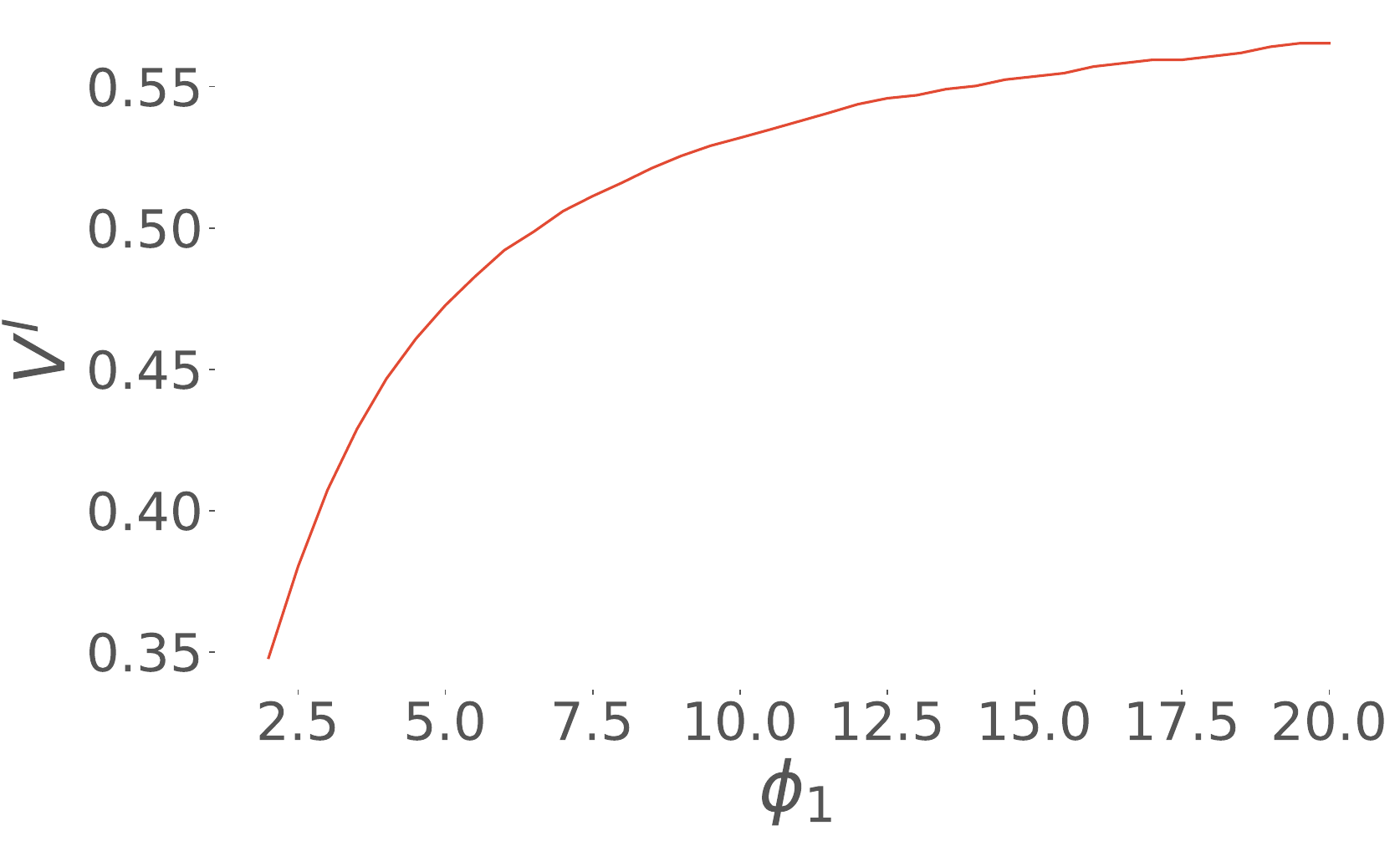}
    \caption{Brokers' and the informed trader's equilibrium value function $V^i(\kappa_1^*,\kappa_2^*)$ and $V^I(\kappa_1^*,\kappa_2^*)$ for different values of running inventory penalty parameter $\phi_1$. The other parameters are the same as in Figure \ref{fig:GBTG-equilibrium-phi_B}.}
    \label{fig:GBTG-phi_B-V_B&V_I}
\end{figure}

\subsubsection{Effect of Brokers' Terminal Penalty on the Equilibrium}\label{sec:penalty_effect}

In this section, we investigate how the terminal penalty parameter $a_i$ affects the equilibrium in the case of two asymmetric brokers. More specifically, we fix the terminal penalty parameter $a_2$ of broker 2 along with all other parameters and vary $a_1$ of broker 1.

\begin{figure}[!htp]
    \centering
    \includegraphics[width=0.49\textwidth]{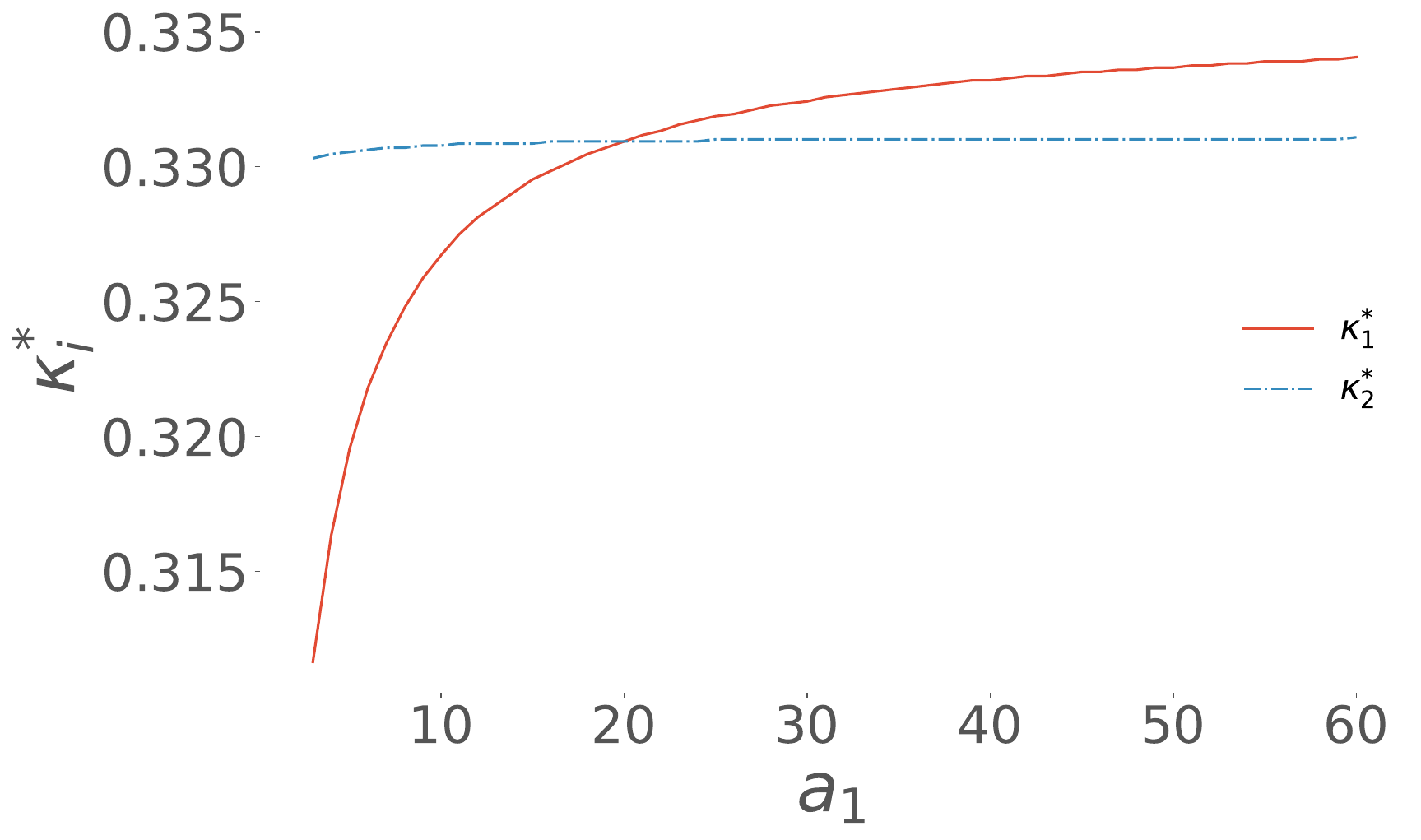}
    \includegraphics[width=0.49\textwidth]{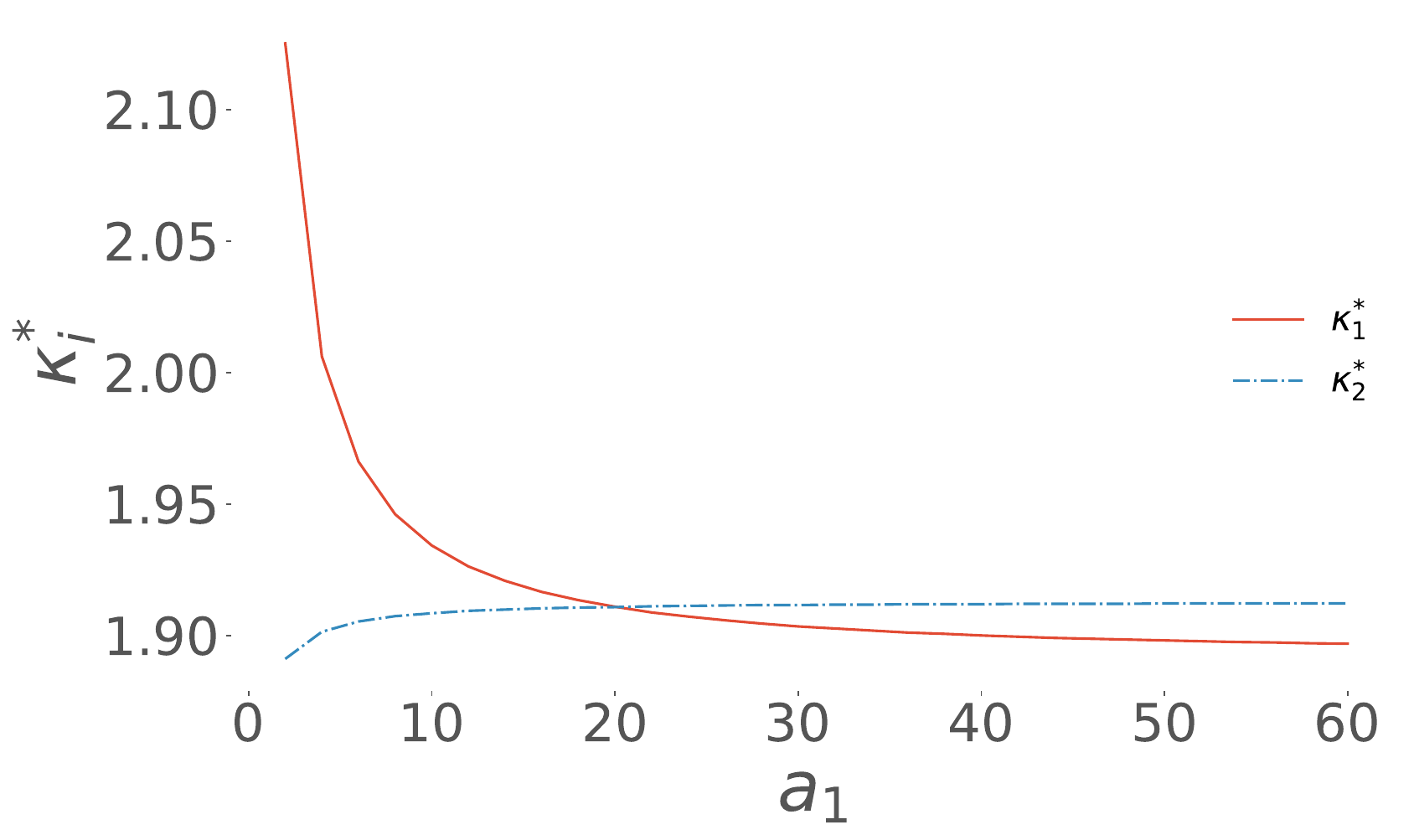}
    \caption{Equilibrium liquidity prices for differing values of broker 1's terminal inventory penalty parameter $a_1$. Left: $\phi_1=\phi_2=10$ (high). Right: $\phi_1=\phi_2=1$ (low). Other parameters are $a_2=20$, $T=1$, $N=5000$, $S_0=5$, $\sigma=1$, $\alpha_0=0$, $\theta=1$, $\eta=5$, $\phi_I=0.01$, $a_I=1$, $\psi_I=0.01$, $\theta_1=\theta_2=1$, $\eta_1=\eta_2=1$, $u^1_0=u^2_0=0$, $c_1=c_2=0.1$, $b_1=b_2=0.1$, $k_1=k_2=0.12$, $Q^I_0=Q^{1}_0=Q^{2}_0=0$, $X^I_0=X^{1}_0=X^{2}_0=0$.}
    \label{fig:GBTG-equilibrium-a_B}
\end{figure}

An interesting result is that the relationship between the equilibrium and terminal inventory penalty parameter $a_1$ also depends on the value of running inventory penalty parameters $\phi_i$ of both brokers. On the left panel of Figure \ref{fig:GBTG-equilibrium-a_B}, we can see that the optimal liquidity price $\kappa^{*}_1$ set by broker 1 increases for higher terminal penalty $a_1$, at high running inventory penalty $\phi_1$. This is economically sensible as higher liquidity price $\kappa^*_1$ will result in lower trading volume from the informed trader, and in turn reduce the terminal inventory pressure. However, on the right panel, for low running inventory penalty $\phi_1$, the optimal cost $\kappa^*_1$ is decreasing in terminal inventory penalty $a_1$. This is somewhat counterintuitive since in general we expect the broker to increase cost with higher penalty, so as to reduce traded volume. We display the t-test results for the same illustrative quantities as in Section \ref{sec:penalty_effect} in the two cases $a_1=2$ and $a_1=20$ and under two different regimes $\phi_i=1$ or $\phi_i=10$ respectively in Table \ref{table:a_B_1_t_test}. 
We conduct 4000 simulations for $\phi_i=1$ and $\phi_i=10$ respectively.

First we attempt to reconcile the counter-intuitive observation in the right panel. In the low penalty regime ($\phi_i=1$) we note that between two scenarios $a_1 = 2$ and $a_1 = 20$ we have:
\begin{itemize}
    \item As $Z^I$ increases and $Y_{tot}^I$ decreases from Scenario 1 to Scenario 2, broker 1 allows more informed trading while liquidating the informed flow faster.
    \item As $Y_{tot}^\alpha$ increases, broker 1 reduces her speculative trading volume.
    \item As $\bar{Q}$ decreases, broker 1 in general holds less inventory in the face of higher terminal penalty.
\end{itemize}
Similar to Section \ref{sec:RunningInventoryNumericalExperiment}, here we also reconcile the observation for the increase in terminal penalty. Continuing, we focus on the reason why broker 1 behaves differently in the two panels of Figure \ref{fig:GBTG-equilibrium-a_B} under different running inventory penalty. We make three additional observations from Table \ref{table:a_B_1_t_test}:
\begin{enumerate}
    \item The ratio $Y_{tot}^\alpha / Z^I$ is significantly higher at $\phi_i=1$ ($\frac{1.353}{0.118}\approx11.5$) than at $\phi_i=10$ ($\frac{0.501}{0.646}\approx0.8$). This indicates that when the running inventory penalty $\phi_i$ is low, broker 1 primarily profits through alpha. Conversely, when $\phi_i$ is high, broker 1 relies more on transactions with the informed trader.
    \item For low $\phi_i$ with higher dependency on speculative trading, broker 1 reduces $Y_{tot}^\alpha$ by 7.33\% . However, for high $\phi_i$, she barely changes $Y_{tot}^\alpha$.
    \item At high $\phi_i$, broker 1 cuts down informed trading, her main channel of profit, as indicated by the decrease of $Z^I$ by 6.48\%. 
\end{enumerate}
Therefore, we conclude that the different behaviours of broker 1 is caused by the distinct focus of trading (speculative trading or informed trading).

\begin{table}
\begin{center}
    \begin{tabular}{ l|l|l|l|l|l|l|l|l }
     Char. & average-1 & & std-1 & & average-2 & & std-2 & \\
     \hline
     & $\phi_i=1$ & $\phi_i=10$ & $\phi_i=1$ & $\phi_i=10$ & $\phi_i=1$ & $\phi_i=10$ & $\phi_i=1$ & $\phi_i=10$ \\
     \hline
     $Z^I$ & 0.118 & 0.646 & 0.051 & 0.287 & 0.131(11.17\%) & 0.604(-6.48\%) & 0.056 & 0.268 \\
     $\bar{Q}$ & 0.254 & 0.045 & 0.164 & 0.027 & 0.225(-11.24\%) & 0.043(-5.99\%) & 0.146 & 0.025 \\
     $Y_{tot}^\alpha$ & 1.353 & 0.501 & 0.669 & 0.251 & 1.253(-7.33\%) & 0.498(-0.47\%) & 0.620 & 0.250 \\
     $Y_{tot}^I$ & 0.008 & 0.064 & 0.005 & 0.046 & 0.005(-34.36\%) & 0.028(-56.23\%) & 0.003 & 0.019
     \end{tabular}
    \caption{Average and standard deviation for characteristics in Scenario 1 ($a_1=2$) and Scenario 2 ($a_1=20$) when $\phi_i=1$ and $\phi_i=10$ respectively. The numbers in parentheses represent the percentage decrease. The other parameters are the same as in Figure \ref{fig:GBTG-equilibrium-a_B}.}\label{table:a_B_1_t_test}
\end{center}
\end{table}

In Figure \ref{fig:GBTG-a_B-V_B&V_I}, broker 1 has a decreasing initial value function due to higher terminal inventory penalty both for $\phi_i=1$ and $\phi_i=10$, as expected. However, although broker 2 slightly increases her optimal liquidity price $\kappa^*_2$ to the informed trader both for $\phi_i=1$ and $\phi_i=10$ in Figure \ref{fig:GBTG-equilibrium-a_B}, she has an increasing initial value function when $\phi_i=10$ (the first panel) and a decreasing initial value function when $\phi_i=1$ (the third panel). This discrepancy arises because broker 2 attracts more customers when the running inventory penalty is higher, but loses customers when the running inventory penalty is lower compared to broker 1. The tendency of the informed trader's initial value function is more influenced by her trading volume with the two brokers. Consequently, her initial value function decreases as $\kappa^*_1$ increases when $\phi_i=10$, and increases as $\kappa^*_1$ decreases when $\phi_i=1$.

\begin{figure}
    \centering
    \includegraphics[width=0.44\textwidth]{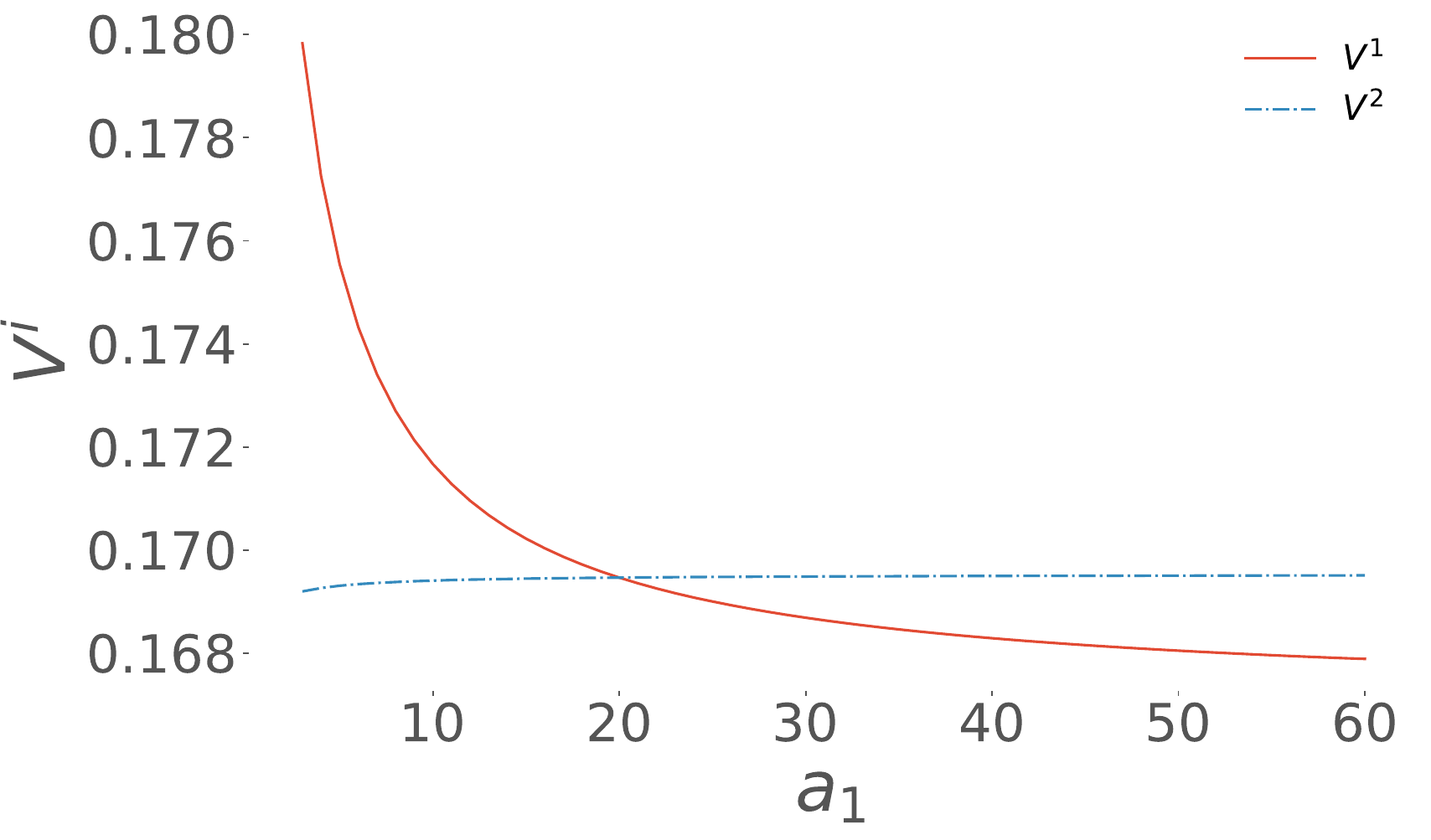}
    \includegraphics[width=0.44\textwidth]{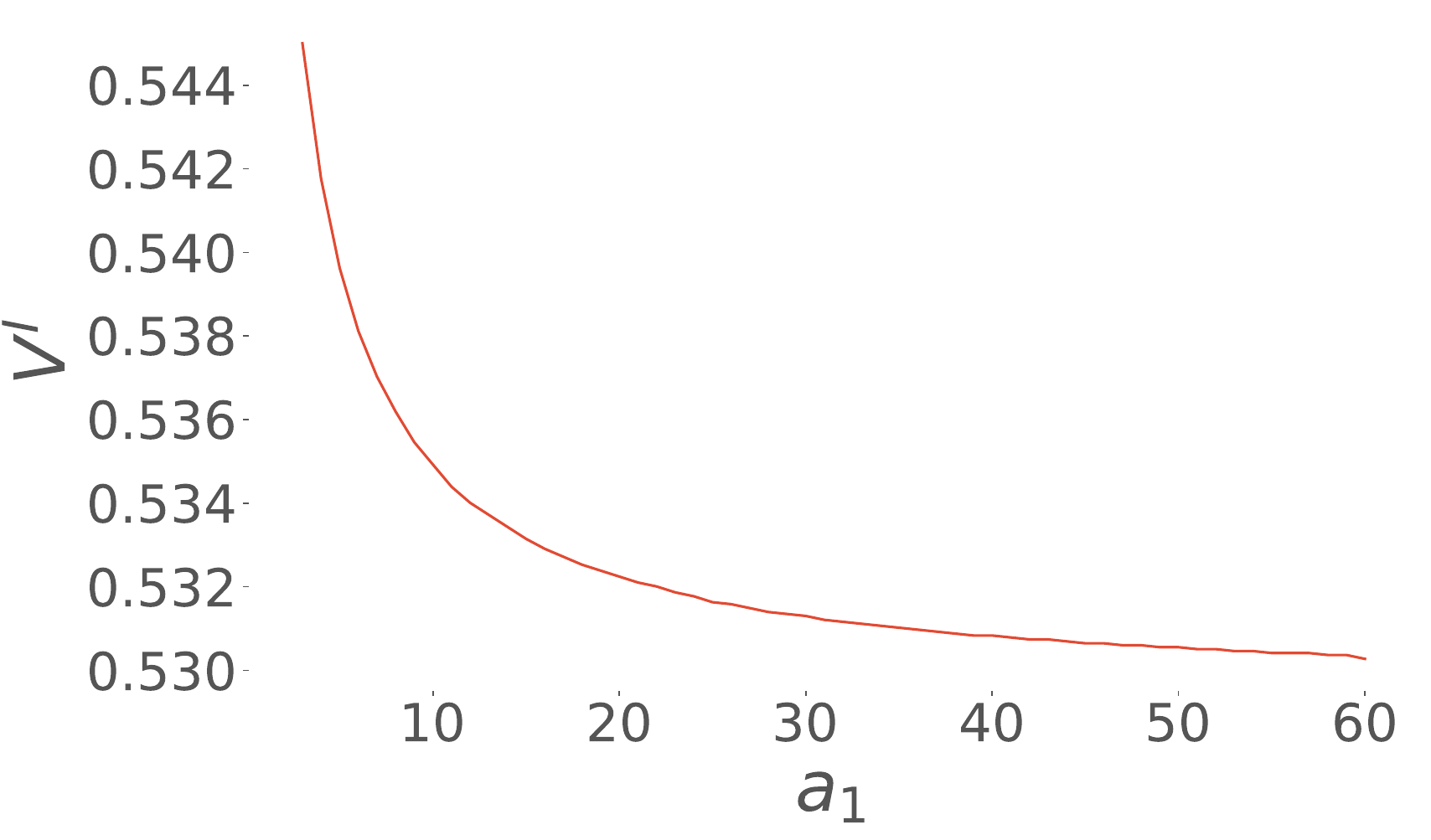}\\
    \includegraphics[width=0.44\textwidth]{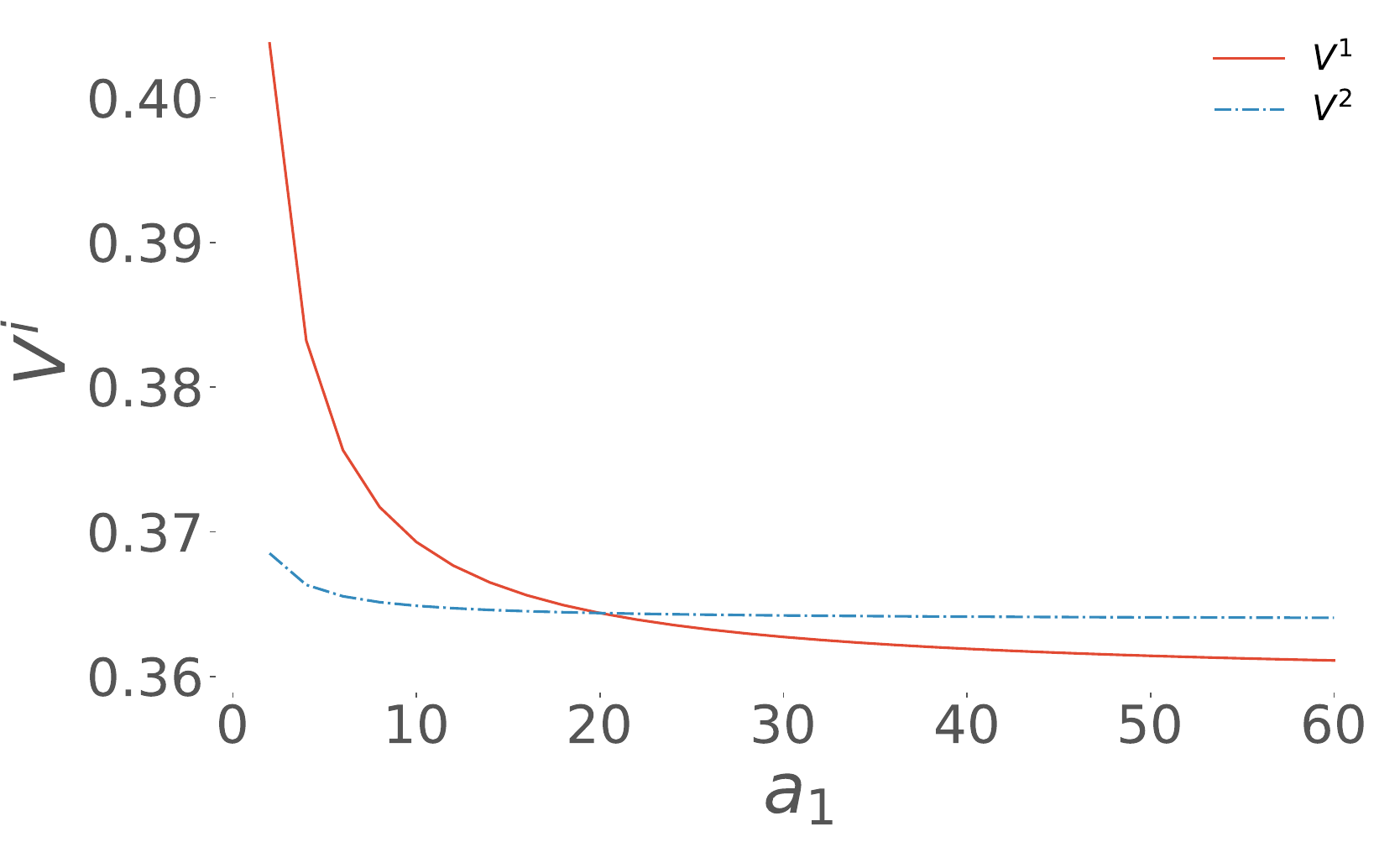}
    \includegraphics[width=0.44\textwidth]{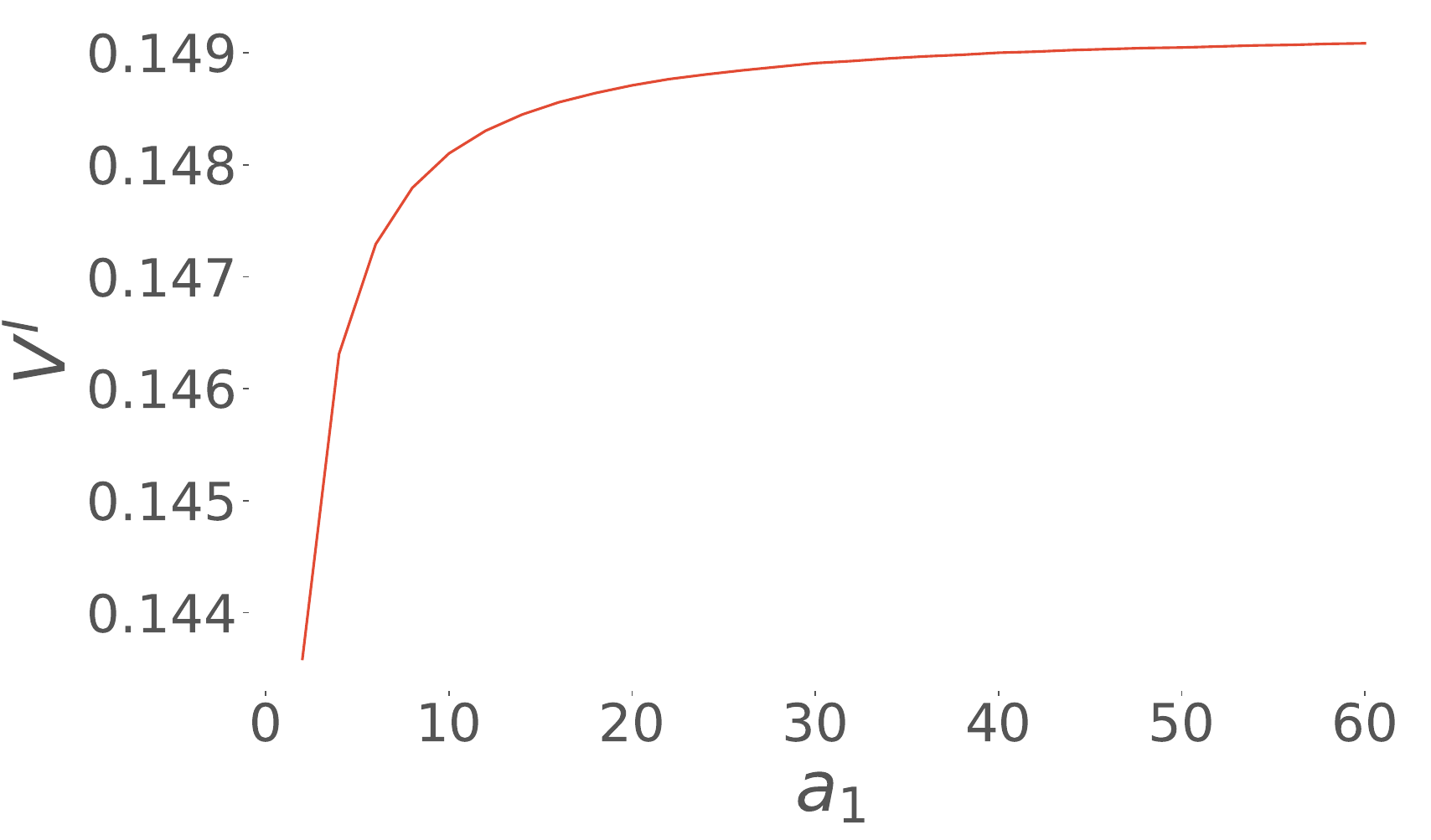}
    \caption{Initial value functions of the brokers ($V^B$) and the informed trader ($V^I$) in equilibrium for different values of terminal penalty parameter $a_1$. Two figures on the left: $\phi_1=\phi_2=10$. Two figures on the right: $\phi_1=\phi_2=1$. The other parameters are the same as in Figure \ref{fig:GBTG-equilibrium-a_B}.}
    \label{fig:GBTG-a_B-V_B&V_I}
\end{figure}

\section{Conclusion}\label{sec:conclusion}

We have studied a market making game among brokers trading with informed and uniformed counterparties. This Stackelberg-type game tells us how the informed trader decides his trading speed with each broker depending on the liquidity cost they charge, and demonstrates the internalization/externalization strategy for each broker. This is a direct trade-off between speculation and inventory risk. Furthermore, we study the competition among brokers through liquidity costs, before the sequential Stackelberg-type game between the brokers and informed trader. We numerically compute an equilibrium which we demonstrate is not Pareto efficient. The effect of model parameters (e.g. running inventory penalty) on the equilibrium and welfare of each player is then explored.

{\bf Disclosure Statement}

The authors report there are no competing interests to declare.

\bibliographystyle{chicago}
\bibliography{References}

\begin{thebibliography}{}

\bibitem[\protect\citeauthoryear{Almgren and Chriss}{Almgren and
  Chriss}{2001}]{almgren2001optimal}
Almgren, R. and N.~Chriss (2001).
\newblock Optimal execution of portfolio transactions.
\newblock {\em Journal of Risk\/}~{\em 3}, 5--40.

\bibitem[\protect\citeauthoryear{Barzykin, Boyce, and Neuman}{Barzykin
  et~al.}{2024}]{barzykin2024unwinding}
Barzykin, A., R.~Boyce, and E.~Neuman (2024).
\newblock Unwinding toxic flow with partial information.
\newblock {\em arXiv preprint arXiv:2407.04510\/}.

\bibitem[\protect\citeauthoryear{Bergault and S{\'a}nchez-Betancourt}{Bergault
  and S{\'a}nchez-Betancourt}{2025}]{bergault2025mean}
Bergault, P. and L.~S{\'a}nchez-Betancourt (2025).
\newblock A mean field game between informed traders and a broker.
\newblock {\em SIAM Journal on Financial Mathematics\/}~{\em 16\/}(2),
  358--388.

\bibitem[\protect\citeauthoryear{Bertsimas and Lo}{Bertsimas and
  Lo}{1998}]{bertsimas1998optimal}
Bertsimas, D. and A.~W. Lo (1998).
\newblock Optimal control of execution costs.
\newblock {\em Journal of financial markets\/}~{\em 1\/}(1), 1--50.

\bibitem[\protect\citeauthoryear{Cartea, Donnelly, and Jaimungal}{Cartea
  et~al.}{2018}]{cartea2018enhancing}
Cartea, A., R.~Donnelly, and S.~Jaimungal (2018).
\newblock Enhancing trading strategies with order book signals.
\newblock {\em Applied Mathematical Finance\/}~{\em 25\/}(1), 1--35.

\bibitem[\protect\citeauthoryear{Cartea, Duran-Martin, and
  S{\'a}nchez-Betancourt}{Cartea et~al.}{2023}]{cartea2023detecting}
Cartea, {\'A}., G.~Duran-Martin, and L.~S{\'a}nchez-Betancourt (2023).
\newblock Detecting toxic flow.
\newblock {\em arXiv preprint arXiv:2312.05827\/}.

\bibitem[\protect\citeauthoryear{Cartea, Jaimungal, and Ricci}{Cartea
  et~al.}{2014}]{cartea2014buy}
Cartea, {\'A}., S.~Jaimungal, and J.~Ricci (2014).
\newblock Buy low, sell high: A high frequency trading perspective.
\newblock {\em SIAM Journal on Financial Mathematics\/}~{\em 5\/}(1), 415--444.

\bibitem[\protect\citeauthoryear{Cartea and S{\'a}nchez-Betancourt}{Cartea and
  S{\'a}nchez-Betancourt}{2025}]{cartea2025brokers}
Cartea, {\'A}. and L.~S{\'a}nchez-Betancourt (2025).
\newblock Brokers and informed traders: dealing with toxic flow and extracting
  trading signals.
\newblock {\em SIAM Journal on Financial Mathematics\/}~{\em 16\/}(2),
  243--270.

\bibitem[\protect\citeauthoryear{Donnelly and Lorig}{Donnelly and
  Lorig}{2020}]{donnelly2020optimal}
Donnelly, R. and M.~Lorig (2020).
\newblock Optimal trading with differing trade signals.
\newblock {\em Applied Mathematical Finance\/}~{\em 27\/}(4), 317--344.

\bibitem[\protect\citeauthoryear{Easley, De~Prado, and O'Hara}{Easley
  et~al.}{2011}]{easley2011microstructure}
Easley, D., M.~M.~L. De~Prado, and M.~O'Hara (2011).
\newblock The microstructure of the" flash crash": Flow toxicity, liquidity
  crashes, and the probability of informed trading.
\newblock {\em Journal of Portfolio Management\/}~{\em 37\/}(2), 118.

\bibitem[\protect\citeauthoryear{Forsyth, Kennedy, Tse, and Windcliff}{Forsyth
  et~al.}{2012}]{forsyth2012optimal}
Forsyth, P.~A., J.~S. Kennedy, S.~T. Tse, and H.~Windcliff (2012).
\newblock Optimal trade execution: a mean quadratic variation approach.
\newblock {\em Journal of Economic dynamics and Control\/}~{\em 36\/}(12),
  1971--1991.

\bibitem[\protect\citeauthoryear{Glosten and Milgrom}{Glosten and
  Milgrom}{1985}]{glosten1985bid}
Glosten, L.~R. and P.~R. Milgrom (1985).
\newblock Bid, ask and transaction prices in a specialist market with
  heterogeneously informed traders.
\newblock {\em Journal of financial economics\/}~{\em 14\/}(1), 71--100.

\bibitem[\protect\citeauthoryear{Gu{\'e}ant, Lehalle, and
  Fernandez-Tapia}{Gu{\'e}ant et~al.}{2012}]{gueant2012optimal}
Gu{\'e}ant, O., C.-A. Lehalle, and J.~Fernandez-Tapia (2012).
\newblock Optimal portfolio liquidation with limit orders.
\newblock {\em SIAM Journal on Financial Mathematics\/}~{\em 3\/}(1), 740--764.

\bibitem[\protect\citeauthoryear{Herdegen, Muhle-Karbe, and Stebegg}{Herdegen
  et~al.}{2023}]{herdegen2023liquidity}
Herdegen, M., J.~Muhle-Karbe, and F.~Stebegg (2023).
\newblock Liquidity provision with adverse selection and inventory costs.
\newblock {\em Mathematics of Operations Research\/}~{\em 48\/}(3), 1286--1315.

\bibitem[\protect\citeauthoryear{Huang, Jaimungal, and Nourian}{Huang
  et~al.}{2019}]{huang2019mean}
Huang, X., S.~Jaimungal, and M.~Nourian (2019).
\newblock Mean-field game strategies for optimal execution.
\newblock {\em Applied Mathematical Finance\/}~{\em 26\/}(2), 153--185.

\bibitem[\protect\citeauthoryear{Kyle}{Kyle}{1985}]{kyle1985continuous}
Kyle, A.~S. (1985).
\newblock Continuous auctions and insider trading.
\newblock {\em Econometrica: Journal of the Econometric Society\/}, 1315--1335.

\bibitem[\protect\citeauthoryear{Pham}{Pham}{2009}]{pham2009continuous}
Pham, H. (2009).
\newblock {\em Continuous-time stochastic control and optimization with
  financial applications}, Volume~61.
\newblock Springer Science \& Business Media.

\bibitem[\protect\citeauthoryear{Schied and Zhang}{Schied and
  Zhang}{2019}]{schied2019market}
Schied, A. and T.~Zhang (2019).
\newblock A market impact game under transient price impact.
\newblock {\em Mathematics of Operations Research\/}~{\em 44\/}(1), 102--121.

\end{thebibliography}




\end{document}